\documentclass[11pt]{article}

\usepackage{latexsym}
\usepackage{a4wide}
\usepackage{amssymb,amsmath}
\usepackage{amsthm}
\usepackage{amsfonts}
\usepackage{MnSymbol}
\usepackage{rotating}

\theoremstyle{plain}
\newtheorem{theorem}{Theorem}
\newtheorem{corollary}[theorem]{Corollary}

\newtheorem{lemma}[theorem]{{\sc Lemma}}
\newtheorem{claim}[theorem]{{\sc Claim}}
\newtheorem{example}[theorem]{{\sc Example}}
\newtheorem{lclaim}{{\sc Claim}}[theorem]
\newtheorem{tlemma}{{\sc Lemma}}[theorem]

\newcommand{\mf}[1]{\mathbf{#1}}
\newcommand{\mprod}{\!\times\!}
\newcommand{\sqprod}{\!\times^{\textit{sq}}}    
\newcommand{\hh}{\mathbf{h}}
\newcommand{\vv}{\mathbf{v}}
\newcommand{\sizeh}{\textit{h\_size}}
\newcommand{\sizev}{\textit{v\_size}}

\newcommand{\rank}{\textit{rank}\,}
\newcommand{\ub}[2]{\textit{ub}^{#1}({#2})}

\newcommand{\weight}{\textit{weight}\,}
\newcommand{\nor}[1]{r_{#1}}
\newcommand{\noi}[1]{i_{#1}}
\newcommand{\noc}[1]{n_{#1}}
\newcommand{\nok}[1]{k_{#1}}
\newcommand{\con}[1]{\Gamma^{#1}}

\newcommand{\nN}{\mathbb{N}}
\newcommand{\nNp}{\mathbb{N}^+}
\newcommand{\infy}{\aleph_0}
\newcommand{\maxz}[1]{\textit{max}\,({#1})}
\newcommand{\minz}[1]{\textit{min}\,({#1})}
\newcommand{\maxS}[1]{\textit{max}\,({#1})}
\newcommand{\minS}[1]{\textit{min}\,({#1})}

\newcommand{\gf}{\mathcal{G}_{\F}}
\newcommand{\hfub}{\mathcal{H}_{\F}}
\newcommand{\hfubm}{\mathcal{H}_{\F}^{\textit{b}}}
\newcommand{\efub}{E^\textit{ub}_\F}
\newcommand{\tfub}{\mathcal{T}_{\F}}

\newcommand{\gfp}{\mathcal{G}_{\F}^+}
\newcommand{\scc}{\mathcal{S}}

\newcommand{\swall}{\textit{Sw}_{\F}}

\newcommand{\Pub}[2]{P_\textit{ub}^{#1}({#2})}

\newcommand{\hset}[1]{\textit{Set}({#1})}
\newcommand{\wset}[1]{\textit{Gen}({#1})}
\newcommand{\hsetx}[2]{\textit{Set}^{#1}({#2})}
\newcommand{\hsetlb}[1]{\textit{Set}_{\textit{lb}}({#1})}

\newcommand{\treec}[1]{\tfub({#1})}

\newcommand{\xlb}{X_{\textit{lb}}}
\newcommand{\ylb}{Y_{\textit{lb}}}
\newcommand{\yub}{Y^\textit{b}}

\newcommand{\Log}{{\sf Logic\_of}\,}
\newcommand{\Diff}{\mf{Diff}}
\newcommand{\K}{\mf{K}}
\newcommand{\Sfive}{\mf{S5}}
\newcommand{\comm}{[\Diff,\Diff]}

\newcommand{\sxd}{\Sfive\mprod\Diff}
\newcommand{\ssqxd}{\Sfive\sqprod\Diff}
\newcommand{\dxd}{\Diff\mprod\Diff}
\newcommand{\dsqxd}{\Diff\sqprod\Diff}

\newcommand{\cluster}{bi-cluster}
\newcommand{\clusters}{bi-clusters}
\newcommand{\grid}{grid of bi-clusters}
\newcommand{\grids}{grids of bi-clusters}
\newcommand{\bcno}[1]{no{#1}}
\newcommand{\bcinf}[1]{inf{#1}}
\newcommand{\bcvhsw}{v2hsw}
\newcommand{\bchvsw}{h2vsw}
\newcommand{\bceqsw}{$=$sw}
\newcommand{\bchstrict}{hstrict}
\newcommand{\bcvstrict}{vstrict}
\newcommand{\bchvstrict}{hvstrict}
\newcommand{\bcfree}{free}

\newcommand{\Dh}{\Diamond_\hh}
\newcommand{\Bh}{\Box_\hh}
\newcommand{\Dv}{\Diamond_\vv}
\newcommand{\Bv}{\Box_\vv}

\newcommand{\Rh}{R_\hh}
\newcommand{\Rv}{R_\vv}
\newcommand{\Rhp}{R_\hh^+}
\newcommand{\Rvp}{R_\vv^+}
\newcommand{\urel}[1]{\forall_{#1}}
\newcommand{\drel}[1]{\ne_{#1}}
\newcommand{\Qh}{Q_\hh}
\newcommand{\Qv}{Q_\vv}

\newcommand{\numin}{\nu_{\textit{min}}^{\F}}
\newcommand{\smin}{\xi_{\textit{min}}^{\F}}

\newcommand{\srdf}{\mathsf{sRdf}_2}
\newcommand{\srdfsq}{\mathsf{sRdf}_2^{\textit{sq}}}
\newcommand{\crs}{\mathsf{Crs}_n}

\newcommand{\F}{{\mathfrak F}}
\newcommand{\G}{{\mathfrak G}}
\newcommand{\Hh}{{\mathfrak H}}
\newcommand{\M}{{\mathfrak M}}
\newcommand{\N}{{\mathfrak N}}
\newcommand{\C}{\mathfrak C}
\newcommand{\Cc}{\mathcal{C}}

\newcommand{\rr}{%
\setlength{\unitlength}{.07cm}
\begin{picture}(5,3)
\thicklines
\put(0,1.2){\circle{3}}
\put(3,1.2){\circle{3}}
\end{picture}}
\newcommand{\ii}{%
\setlength{\unitlength}{.07cm}
\begin{picture}(5,3)
\thicklines
\put(0,1.2){\circle*{3.5}}
\put(3,1.2){\circle*{3.5}}
\end{picture}}
\newcommand{\ri}{%
\setlength{\unitlength}{.07cm}
\begin{picture}(5,3)
\thicklines
\put(0,1.2){\circle{3}}
\put(3,1.2){\circle*{3.5}}
\end{picture}}
\newcommand{\ir}{%
\setlength{\unitlength}{.07cm}
\begin{picture}(5,3)
\thicklines
\put(0,1.2){\circle*{3.5}}
\put(3,1.2){\circle{3}}
\end{picture}}

\newcommand{\ixx}{%
\setlength{\unitlength}{.07cm}
\begin{picture}(5,3)
\thicklines
\put(0,1.2){\circle*{3.5}}
\put(2.1,-.3){{\tiny *}}
\put(3,1.2){\circle{3}}
\end{picture}}

\newcommand{\xxi}{%
\setlength{\unitlength}{.07cm}
\begin{picture}(5,3)
\thicklines
\put(0,1.2){\circle{3}}
\put(-.9,-.3){{\tiny *}}
\put(3,1.2){\circle*{3.5}}
\end{picture}}

\newcommand{\avar}[1]{\mathsf{a}_{#1}}
\newcommand{\bvar}[1]{\mathsf{b}_{#1}}
\newcommand{\avarr}{\mathsf{a}}
\newcommand{\bvarr}{\mathsf{b}}
\newcommand{\cvarr}{\mathsf{c}}
\newcommand{\dvarr}{\mathsf{d}}
\newcommand{\Cvar}[1]{\mathsf{c}_{#1}}

\newcommand{\commf}{\mathsf{comm\_pse}}
\newcommand{\alphaC}{\mathsf{clash}_{\C}}
\newcommand{\betaf}{\mathsf{path}}

\newcommand{\deltaf}{\mathsf{large}}
\newcommand{\atp}[1]{\hat{\mathsf{a}}_{#1}}
\newcommand{\btp}[1]{\hat{\mathsf{b}}_{#1}}
\newcommand{\ctp}[1]{\hat{\mathsf{c}}_{#1}}
\newcommand{\ddtp}{\hat{\mathsf{d}}}
\newcommand{\cctp}{\hat{\mathsf{c}}}

\newcommand{\nof}{\mathsf{impossible}_\C}

\newcommand{\pathf}[1]{\mathsf{path}_{\F}^{#1}}

\newcommand{\firstf}[1]{\mathsf{small}_{\F}^{#1}}
\newcommand{\lastf}[1]{\mathsf{large}_{\F}^{#1}}
\newcommand{\badpathf}[1]{\mathsf{bad\_path}_{\F}^{#1}}

\newcommand{\xvar}{\mathsf{x}}
\newcommand{\yvar}{\mathsf{y}}
\newcommand{\hvar}[2]{\mathsf{a}^{#1}_{#2}}
\newcommand{\cvar}[2]{\mathsf{c}^{{#1}{#2}}}

\newcommand{\solf}{\mathsf{solution}_{\F}}
\newcommand{\xvarbar}{\overline{\mathsf{x}}}
\newcommand{\yvarbar}{\overline{\mathsf{y}}}
\newcommand{\zvarbar}{\overline{\mathsf{z}}}

\newcommand{\ctpp}[2]{\hat{\mathsf{c}}^{{#1}{#2}}}
\newcommand{\treef}[1]{\mathsf{tree}({#1})}
\newcommand{\treefi}[2]{\mathsf{tree}_{#1}({#2})}
\newcommand{\antefone}{\mathsf{upper\_bound}_\F}
\newcommand{\anteftwo}{\mathsf{switch}_\F}
\newcommand{\swf}[1]{\mathsf{switch}_\F^{#1}}
\newcommand{\antefthree}{\mathsf{lower\_bound}_\F}
\newcommand{\lbf}[1]{\mathsf{lower\_bound}_\F^{#1}}
\newcommand{\consf}{\mathsf{out}_\F}
\newcommand{\finalf}[1]{\mathsf{out}_{\F}({#1})}
\newcommand{\badsqf}{\mathsf{square\_bad}_\F}
\newcommand{\solfsub}{\mathsf{solution}_{\F^-}}

\begin{document}


\title{Non-finitely axiomatisable modal product logics with\\ infinite canonical axiomatisations}

\author{Christopher Hampson$^a$, Stanislav Kikot$^b$, Agi Kurucz$^a$ and S\'ergio Marcelino$^c$\\
\ \\
{\small ${}^a$Department of Informatics, King's College London, U.K.}\\
{\small ${}^b$Institute for Information Transmission Problems, Moscow, Russia}\\
{\small ${}^c$Instituto de Telecomunica\c{c}\~oes, Portugal}}

\date\

\maketitle

\begin{abstract}
Our concern is the axiomatisation problem for modal and algebraic logics that correspond to various 
fragments of two-variable first-order logic with counting quantifiers.
In particular,
we consider modal products with $\Diff$, the propositional unimodal logic of the difference operator.
We show that the two-dimensional product logic $\dxd$ is non-finitely axiomatisable, but can be axiomatised by infinitely
many Sahlqvist axioms. We also show that its `square' version 
(the modal counterpart 
of the substitution and equality free fragment of two-variable first-order logic with counting to two) is 
non-finitely axiomatisable over $\dxd$, but can be axiomatised by adding infinitely
many Sahlqvist  axioms. 
These are the first examples of products of finitely axiomatisable modal logics that are not finitely axiomatisable,
but axiomatisable by explicit infinite sets of canonical axioms.
\end{abstract}





\section{Introduction}

Ever since their introduction \cite{Segerberg73,Shehtman78,Gabbay&Shehtman98},
\emph{products of modal logics\/}---proposi\-tional multimodal logics determined by classes of product 
frames---have been extensively studied; see \cite{gkwz03} for a comprehensive exposition and 
further references.
In this paper we consider the problem of finding explicit infinite `nice' axiomatisations for non-finitely axiomatisable two-dimensional modal product logics. By `nice' here we mean formulas to which both the canonicity and
first-order correspondence properties of Sahlqvist formulas apply.

Canonicity is an important tool for proving Kripke completeness of propositional multimodal logics
\cite{Blackburnetal01,Goldblatt87}.
A modal logic is \emph{canonical} if it is valid in all its canonical frames.
The analogous algebraic notion of \emph{canonical extension} is central in
the theory of Boolean algebras with operators (BAOs) \cite{Jonsson&Tarski51}.
 A variety of BAOs is \emph{canonical} if it is closed under taking canonical extensions.
%
A modal formula is \emph{canonical} if the modal logic axiomatised by it is canonical.
Though in general canonicity of a formula is an undecidable `semantical' property \cite{Kracht96},
there exist known syntactical classes of canonical formulas, such as Sahlqvist formulas 
\cite{Sahlqvist75},
and their generalisations by Goranko and Vakarelov \cite{gv06}.

While any set of canonical formulas always axiomatises a canonical logic, Hodkinson and
Venema \cite{hv05} show that there are canonical logics that are \emph{barely canonical}
in the sense that every axiomatisation for such logics must contain infinitely many non-canonical
axioms. Further examples of barely canonical elementarily generated logics are given
in \cite{goldblatthodkinson07,BulianH13,Kikot15}.
Kikot \cite{Kikot15} also obtained the
following general dichotomy result: If a class of Kripke frames is definable by first-order formulas of the
form $\displaystyle\forall x_0\exists x_1\dots\exists x_n\bigwedge R_\lambda(x_i, x_j)$,  then the modal logic generated by such a class is either barely canonical or can be 
axiomatised by a single generalised Sahlqvist formula.
In this paper we show some elementarily generated modal logics that are outside of the scope of
this dichotomy.

It is well known that the two-dimensional (2D) modal product logic $\Sfive\mprod\Sfive$
has a finite axiomatisation with Sahlqvist axioms, describing two commuting $\Sfive$-modalities \cite{Henkinetal85}. 
($\Sfive\mprod\Sfive$ is the modal counterpart of the substitution and equality free fragment of  two-variable first-order logic, where only relational atomic formulas of the form $R(x_0,x_1)$ are allowed.)
On the other hand, for $n\geq 3$ the 
$n$-dimensional product logic $\Sfive^n$ is non-finitely axiomatisable \cite{Johnson69}
and barely canonical (even though it is canonical and recursively enumerable \cite{Henkinetal85}).
There are also known examples of recursively enumerable (even decidable) 2D products 
of finitely axiomatisable modal logics that are not finitely axiomatisable, such as ${\bf K4.3}\mprod\Sfive$ \cite{Kurucz&Marcelino12}.
However, so far no canonical axiomatisations for non-finitely axiomatisable products of
finitely axiomatisable logics have been known.

Instead of $\Sfive$ (the modal logic of all equivalence relations),
here we consider modal products with the finitely axiomatisable \cite{Segerberg80} logic $\Diff$ of all non-equality frames $(W,\drel{W})$.
An arbitrary frame for $\Diff$ is a \emph{pseudo-equivalence} relation: its equivalence classes might contain
both reflexive and irreflexive points. 
(In particular, equivalence relations are frames for $\Diff$, and so $\Diff\subseteq\Sfive$.)
It is easy to see that, unlike equivalence relations, the class of pseudo-equivalence relations is not Horn-definable. Therefore, the general theorem of Gabbay and Shehtman \cite{Gabbay&Shehtman98} on axiomatising
2D products of Horn-definable logics by their commutator does not apply to $\dxd$. 
However, as pseudo-equivalence relations form an elementary class, it does follow from general results \cite{goldblatt89,Kurucz10,Gabbay&Shehtman98}
that $\dxd$ is canonical and recursively enumerable.

We show that the 2D product logic $\dxd$ is non-finitely axiomatisable, but can be axiomatised by infinitely
many Sahlqvist axioms. We also show that its `square' version \mbox{$\dsqxd$} (the modal counterpart 
of the substitution and equality free fragment of two-variable first-order logic with counting to two) is 
non-finitely axiomatisable over $\dxd$, but can be axiomatised by adding infinitely 
many axioms that are generalised Sahlqvist \`a la Goranko and Vakarelov \cite{gv06}.
This way we give the first examples of products of finitely axiomatisable modal logics that are not finitely axiomatisable
but axiomatisable by explicit infinite sets of canonical axioms. 
By the correspondence theorem for (generalised) Sahlqvist formulas it follows that the classes
of all frames for both $\dxd$ and $\dsqxd$ are elementary
(unlike the frames for $\Diff^n$ and $\Sfive^n$ whenever $n\geq 3$, see \cite{hh09,Kurucz10}).
As $\Diff$-modalities are `self-reversive', it also follows \cite{gv01} that $\dsqxd$  in fact can be axiomatised by infinitely many Sahlqvist axioms.

Our results can also be formulated in an algebraic logic setting.
Given the full Boolean set algebra $\mathfrak{B}(U\times V)$ of all subsets of the Cartesian
product $U\times V$ of some non-empty sets $U,V$, one can define two additional unary operations 
$C^{\ne}_0$, $C^{\ne}_1$ on it by taking,
for every $X\subseteq U\times V$,
\begin{align*}
C^{\ne}_0 (X)  & = \bigl\{ (u,v) : \mbox{there is  $u'\ne u$ with $(u',v)\in X$}\bigr\}\\
C^{\ne}_1 (X) &  = \bigl\{ (u,v) : \mbox{there is $v'\ne v$ with $(u,v')\in X$}\bigr\}.
\end{align*}
Just like usual cylindrifications are algebraisations of the existential quantifier in first-order logic,
these \emph{strict-cylindrifications} algebraise the `there is a different' first-order quantifier.
We define $\srdf$ as the variety generated by all set algebras of this kind,
and $\srdfsq$ as the variety generated by those ones where $U=V$ for the Boolean unit $U\times V$.
Members of $\srdf$ ($\srdfsq$) might be referred to as \emph{two-dimensional rectangularly} (\emph{square}) \emph{representable diagonal-free strict-cylindric algebras\/}.
It follows from general considerations that both $\srdf$ and $\srdfsq$ are canonical varieties.
We show the following:
\begin{itemize}
\item
The equational theory of $\srdf$ is non-finitely axiomatisable, but it has an infinite Sahlqvist axiomatisation.
\item
The equational theory of $\srdfsq$ is non-finitely axiomatisable over that of $\srdf$, but it has an infinite generalised Sahlqvist axiomatisation.
\end{itemize}
While our varieties are the first such among `full rectangular' algebraisations of finite variable fragments of classical first-order logic, a similarly behaving `non-rectangular' algebraisation has been known.
Andr\'eka and N\'emeti \cite[5.5.12]{Henkinetal85} showed that the equational theory of 
the variety $\crs$ of $n$-dimensional 
\emph{relativised cylindric algebras}
is non-finitely axiomatisable whenever $n\geq 3$, 
while Resek and Thompson \cite{ResekThompson91,Monk2000} gave an infinite Sahlqvist axiomatisation for it, for any $n$.


\section{Our results and proof methods}\label{results}

\subsection{Non-finite axiomatisability}\label{proofmethodnonfin}


\begin{theorem}\label{co:nonfinax}
For any Kripke complete logic $L$ with $\K\subseteq L\subseteq\Sfive$, 
$L\mprod\Diff$ is not axiomatisable using finitely many propositional variables.
Thus, $\dxd$ is not finitely axiomatisable.
\end{theorem}

\begin{theorem}\label{t:sqnonfinax}
$\dsqxd$ is not axiomatisable over $\dxd$ using finitely many propositional variables.
\end{theorem}

After providing the necessary definitions in \S\ref{s:prelim} and some general tools in \S\ref{grids},
Theorems~\ref{co:nonfinax} and \ref{t:sqnonfinax} are proved in \S\ref{nonfinax}.
In our proofs, we will use the following pattern.
We show that every axiomatisation of a logic $L$ must contain infinitely many propositonal
variables by providing two infinite sequences of frames $\F_k$ and $\G_k$ such that
\begin{itemize}
\item
every $\F_k$ is a frame for $L$, while every $\G_k$ is not,
\item
but if $k$ is sufficiently large compared to $m$, then we cannot distinguish between $\F_k$ 
and $\G_k$ using $m$ many propositional variables.
\end{itemize}


\subsection{Infinite canonical axiomatisations}\label{proofmethod}

\begin{theorem}\label{t:axdxd}
\begin{itemize}
\item[{\rm (i)}]
There is an infinite axiomatisation for $\dxd$ consisting of Sahlqvist formulas. 

\item[{\rm (ii)}]
The class of all frames for $\dxd$ is elementary.

\item[{\rm (iii)}]
For every countable rooted frame $\F$, $\F$ is a frame for $\dxd$ iff 
 $\F$ is the p-morphic image of some product of two difference frames.
\end{itemize}
\end{theorem}

\begin{theorem}\label{t:sqax}
\begin{itemize}
\item[{\rm (i)}]
$\dsqxd$ can be axiomatised by adding infinitely many generalised Sahlqvist formulas to $\dxd$.

\item[{\rm (ii)}]
The class of all frames for $\dsqxd$ is elementary.

\item[{\rm (iii)}]
For every countable rooted frame $\F$, $\F$ is a frame for $\dsqxd$ iff 
 $\F$ is the p-morphic image of some product of two difference frames of the same size.
\end{itemize}
\end{theorem}

As each generalised Sahlqvist formula is axiomatically equivalent to
a Sahlqvist formula with inverse modalities \cite{gv01,gv06}  and
$\Diff$-modalities are `self-reversive', we have
the following (see \S\ref{dxdsqsahl} for more detail):
\begin{corollary}\label{co:sqS}
There is an infinite axiomatisation for $\dsqxd$ consisting of Sahlqvist formulas. 
\end{corollary}

Theorems~\ref{t:axdxd} and \ref{t:sqax} are proved in the respective \S\ref{dxd} and \S\ref{dxdsq}.
 In our proofs, we will use the following pattern.
In order to axiomatise $\Log\Cc$ for some class $\Cc$ of frames, we define a recursive
set $\Sigma$ of (generalised) Sahlqvist formulas, and prove that the following hold:

 \begin{itemize}
 \item[(ax1)]
 All formulas in $\Sigma$ are valid in every frame in $\Cc$.
 \item[(ax2)]
 For every countable rooted frame $\F$ that is not the p-morphic image of some frame in $\Cc$, there is some $\phi_\F\in\Sigma$ such that
 $\phi_\F$ is not valid in $\F$.
 \end{itemize}
 Then it follows that $\Log\Cc$ is axiomatised by $\Sigma$.
  Indeed, let $L$ be the smallest bimodal logic containing $\Sigma$. Then we clearly have 
 $L\subseteq\Log\Cc$ by (ax1).
 On the other hand, by the (generalised) Sahlqvist completeness theorem, $L$ is canonical,
 and so Kripke complete. By the (generalised) Sahlqvist correspondence theorem,
 the class of all frames for $L$ is an elementary class. Then it is easy to see by
 a L\"owenheim--Skolem type argument (see e.g.\ \cite[Thm.~1.6]{gkwz03}) that $L$ is the logic of its countable
 frames, and so by a standard modal logic argument $L$ is the logic of its countable rooted frames. 
 Now take some $\psi\notin L$. Then there is some countable rooted frame $\F$ such
 that $\F$ is a frame for $L$, but $\psi$ is not valid in $\F$.
 By (ax2), $\F$ is the p-morphic image of some frame in $\Cc$, and so $\psi\notin\Log\Cc$.
 
 Now the following two statements clearly follow from the above:
 %
 \begin{itemize}
 \item
 The class of all frames for $\Log\Cc$ is elementary.
 \item
 For every countable rooted frame $\F$, $\F$ is a frame for $\Log\Cc$ iff 
 $\F$ is the p-morphic image of some frame in $\Cc$.
 \end{itemize}
 

\section{Preliminaries and basic definitions}\label{s:prelim}

Our notation and terminology are mostly standard.
We denote the cardinality of a set $S$ by $|S|$. 
Natural numbers are considered as finite cardinals, and we use the usual multiplication operation
and ordering relations $<$ and $\leq$ among them and the infinite cardinal $\infy=|\omega|$.
We call $S$ \emph{countable} if $|S|\leq\infy$.
We denote the set of natural numbers by $\nN$ and its positive members by $\nNp$.
We will also use the usual functions
$\min(X)$, $\max(X)$ and $\sup(X)$ with respect to $<$, for $X\subseteq\nN\cup\{\infy\}$.


\subsection{Digraphs}

We assume that the reader is familiar with the basic notions about digraphs 
$\mathcal{G}=(N_{\mathcal{G}},\to)$
(see \cite{Berge73} for reference).
Below we summarise the notions used in the paper.
%
We call a node (vertex) in $N_{\mathcal{G}}$  an \emph{initial node} if it has no incoming $\to$ edges,
and a \emph{final node} if it has no outgoing $\to$ edges.
A digraph $\mathcal{G}^-$ is called a \emph{subgraph of\/} $\mathcal{G}$ if its nodes and edges are subsets
of the nodes and edges of $\mathcal{G}$, respectively. If the edges of a subgraph $\mathcal{G}^-$
consists of all the edges of $\mathcal{G}$ whose endpoints are nodes in $\mathcal{G}^-$,
then $\mathcal{G}^-$ is called an \emph{induced subgraph of\/} $\mathcal{G}$.
Given two nodes $z,z'$, a (\emph{directed})
\emph{path in $\mathcal{G}$ from $z$ to $z'$} is a finite sequence of subsequent edges,
the first one starting in $z$ and the last one ending in $z'$. The \emph{length} of a path is the number of
edges in it (we also consider paths of length $0$).
We call a path \emph{simple} if it does not contain the same edge twice.
A \emph{cycle} is a path starting and ending at the same node. 
$\mathcal{G}$ is called \emph{acyclic} if it does not contain any cycles.
A \emph{strongly connected component\/} is a maximal subgraph $\mathcal{S}$ such that
for all nodes $z$ and $z'$ in $\mathcal{S}$ there is a path from $z$ to $z'$.
A finite sequence $z_0,\dots,z_m$ of nodes is an \emph{undirected path between} $z_0$ \emph{and} $z_m$ \emph{in} $\mathcal{G}$ if for every $i<m$, either $z_i\to z_{i+1}$ or $z_{i+1}\to z_i$ is an edge in $\mathcal{G}$.

An acyclic digraph $\mathcal{G}$ is called a \emph{directed rooted tree} (or \emph{tree\/}, for short) 
if there is some inital node $r$ (the \emph{root\/}) such that for every node $z$ in $\mathcal{G}$ there is a unique path from $r$ to $z$. For each $z$, the length of this unique path is the \emph{height of\/} $z$.
If $z\to z'$ is an edge in a tree, then $z'$ is called a \emph{child of\/} $z$. A \emph{leaf} in a tree is a node
without children, that is, a final node.

Given a finite digraph $\mathcal{G}$ and a node $r$ in it, the 
\emph{tree unravelling of $\mathcal{G}$ with root\/} $r$ is the directed rooted tree $\mathcal{T}_{\mathcal{G},r}=(T_{\mathcal{G},r},\Rightarrow)$, where $T_{\mathcal{G},r}$ is the set of all paths in $\mathcal{G}$ starting at $r$ (with the length 0 path being the root of $\mathcal{T}_{\mathcal{G},r}$),  and $P\Rightarrow P'$ iff $P'$ can be obtained from $P$ by adding an additional $\to$ edge to its endpoint.
It is customary to identify each path $P\in T_{\mathcal{G},r}$ with a distinct \emph{copy} of its endpoint, in particular,
to identify the root of $\mathcal{T}_{\mathcal{G},r}$ with $r$.




\subsection{Unimodal and bimodal logics}\label{s:bimodal}

In what follows we assume that the reader is familiar with the basic notions in propositional multimodal logic
and its possible world semantics (see \cite{Blackburnetal01,cz} for reference).
Below we
summarise the necessary notions and notation for the bimodal case only, but we will use them
throughout for the unimodal case as well.
We define  \emph{bimodal formulas} by the following grammar:
\[
\phi:=\ p\mid\top\mid\bot\mid\neg\phi\mid\phi_1\land\phi_2\mid\phi_1\lor\phi_2\mid\Bh\phi\mid\Bv\phi\mid\Dh\phi\mid\Dv\phi,
\]
where $p$ ranges over a countably infinite set of propositional variables.
We use the usual abbreviations $\to$, $\leftrightarrow$, 
and also
%
\[
\Diamond_i^+\phi:=\ \phi\lor\Diamond_i\phi,\hspace*{3cm} \Box_i^+\phi:=\ \phi\land\Box_i\phi,
%
\]
for $i=\hh,\vv$. (The subscripts are indicative of the 2D intuition: $\hh$ for `horizontal' and
$\vv$ for `vertical'.)
%
%
Bimodal formulas are evaluated in bimodal \emph{frames}: relational structures of the form
$\F=(W,\Rh,\Rv)$, having two binary relations $\Rh$ and $\Rv$ on a non-empty set $W$.
A (\emph{Kripke}) \emph{model on\/} $\F$ is 
a function $\M$ mapping propositional variables to subsets of $W$. 
(With a slight abuse of notation, we identify the pair $(\F,\M)$ with $\M$.)
Given $m\in\nN$, we call a Kripke model $\M$ $m$-\emph{generated} if
there are at most $m$ different propositional variables $p$ such that
$\M(p)\ne\emptyset$.
The \emph{truth relation}
`$\M,w\models\phi$', connecting points in models and formulas, is defined as usual by induction 
on $\phi$. If $\M,w\models\phi$ for some model $\M$ on $\F$ and some point $w$ in $\M$, then
we say that $\phi$ \emph{is satisfied in\/} $\M$, and \emph{satisfiable in\/} $\F$.
Given a set $\Sigma$ of bimodal formulas, we write $\M\models\Sigma$ if
we have $\M,w\models\phi$,
for every $\phi\in\Sigma$ and every $w\in W$. (We write
just $\M\models\phi$ for $\M\models\{\phi\}$.) We say that $\phi$ is \emph{valid in} $\F$,
if $\M\models\phi$ for every model $\M$ based on $\F$.  If every formula
in a set $\Sigma$ is valid in $\F$, then we say that $\F$ is a \emph{frame for}~$\Sigma$.


The usual operations on unimodal frames and models can be defined on their bimodal counterparts
 as well. In particular, 
 given two frames $\F=(F,R_\hh^{\F},R_\vv^{\F})$ and $\G=(G,R_\hh^{\G},R_\vv^{\G})$,
a function $f:F\to G$ is called a \emph{p-morphism from} $\F$ \emph{to}
$\G$ if it satisfies the following conditions, for all $u,v\in F$,
$y\in G$, $i=\hh,\vv$:
\begin{itemize}
\item
$uR_i^{\F}v\ $ implies $f(u)R_i^{\G}f(v)$ (that is, $f$ is a \emph{homomorphism}),

\item
$f(u)R_i^{\G}y$ implies that there is some $v\in F$ such that $f(v)=y$ and $uR_i^{\F}v$
(the \emph{backward condition}).
\end{itemize}
If $f$ is onto then we say that $\G$ \emph{is a p-morphic image of} $\F$. 
Similarly to the unimodal case, validity of bimodal formulas in frames is preserved 
under taking p-morphic images.
For any model $\M$ on $\F$ and model $\N$ on $\G$, a p-morphism from $\F$ to $\G$ 
is called a \emph{p-morphism from} $\M$ \emph{to} $\N$ whenever, for all propositional
variables $p$ and points $x$ in $\F$, $x\in\M(p)$ iff $f(x)\in\N(p)$.
If $f$ is onto then we say that $\N$ \emph{is a p-morphic image of} $\M$. 
%

%
Given two frames $\F=(F,R_\hh^{\F},R_\vv^{\F})$ and $\G=(G,R_\hh^{\G},R_\vv^{\G})$,
$\F$ is a \emph{subframe of} $\G$ if $F\subseteq G$ and $R_i^{\F}=R_i^{\G}\cap(F\times F)$, for $i=\hh,\vv$.
Given some $x\in F$, the \emph{subframe} $\F^x$ \emph{of} $\F$ 
\emph{generated by point} $x$ is the subframe of $\F$ with the following set $F^x$ of
points:
%
\[
F^x=\{y\in F :  y\mbox{ is accessible from }x
\mbox{ along the reflexive and transitive closure of $R_\hh^{\F}\cup R_\vv^{\F}$}\}.
\]
%
We say that a frame $\F$ is \emph{rooted} if $\F=\F^r$ for some point $r$. Such a point $r$ is
called a \emph{root in\/} $\F$.


A set $L$ of bimodal formulas is called a (normal) \emph{bimodal logic} (or \emph{logic}, for short)
if it contains all propositional tautologies and the formulas
$\Box_i(p\to q)\to(\Box_i p\to\Box_i q)$, for $i=\hh,\vv$,
and is closed  under the rules of Substitution, Modus Ponens and 
Necessitation $\varphi/\Box_i\varphi$, for $i=\hh,\vv$. Given a class $\Cc$ of frames, we always
obtain a logic by taking
\[
\Log\Cc=\{\phi :\phi\mbox{ is a bimodal formula valid in every member of $\Cc$}\}.
\]
We say that $\Log\Cc$ is the \emph{logic of} $\Cc$. It is well known that
\begin{equation}\label{rooted}
\Log\Cc=\Log\{\mbox{rooted frames in $\Cc$}\}.
\end{equation}
A logic $L$ is called
\emph{Kripke complete} if $L=\Log\Cc$ for some class $\Cc$.
Given a bimodal logic $L$ and a recursive set $\Sigma$ of bimodal 
formulas, we say that $\Sigma$ \emph{axiomatises} $L$ if
$L$ is the smallest bimodal logic containing $\Sigma$. 
A logic $L$ is called \emph{finitely axiomatisable\/} whenever there is some finite $\Sigma$ axiomatising $L$.


\subsubsection{Sahlqvist and generalised Sahlqvist formulas}\label{s:Sahlqvist}

Below we recall the definition of Sahlqvist formulas \cite{Sahlqvist75}, and generalised (monadic) Sahlqvist formulas of 
Goranko and Vakarelov \cite[Def.24]{gv06} for our bimodal language.

A bimodal formula is
\emph{positive} (\emph{negative}) if every occurrence of a propositional variable in it is under the scope of an even (odd) number of negations $\neg$.
A \emph{boxed atom} is a formula $\Box_{i_1}\dots\Box_{i_n}p$ 
where $n\in\nN$, $i_1,\dots,i_n\in\{\hh,\vv\}$ and $p$ is a propositional variable.
A \emph{Sahlqvist antecedent} is a formula built up from $\top$, $\bot$, boxed atoms, and negative 
formulas, using $\lor$, $\land$, $\Dh$ and $\Dv$.

A \emph{boxed formula} is a formula of the form
\[
\Box^1\Bigl(\psi_1\to\Box^2\bigl(\psi_2\to\dots\Box^n(\psi_n\to\Box^0 p)\dots\bigr)\Bigr),
\] 
where $n\in\nN$, each $\Box^i$ is a finite (possibly empty) sequence of boxes $\Bh$ and $\Bv$,
each $\psi_i$ is a positive formula,
and $p$ is a propositional variable. The variable $p$ is called the \emph{head} of the
boxed formula, and all variables in any of the $\psi_i$ are called \emph{inessential variables\/}.
A \emph{potential generalised Sahlqvist antecedent} is a formula $\phi$ built up from $\top$, $\bot$, 
boxed formulas, and negative 
formulas, using $\lor$, $\land$, $\Dh$ and $\Dv$.
Given such a formula $\phi$, the \emph{dependency digraph of\/} $\phi$ is a digraph 
$\mathcal{D}(\phi)=(N_\phi,\Rrightarrow)$, where $N_\phi$ is the set of heads of the boxed
formulas in $\phi$, and $q\Rrightarrow p$ iff $q$ is an inessential variable in a boxed formula with head $p$. If $\mathcal{D}(\phi)$ is acyclic, then $\phi$ is called a \emph{generalised Sahlqvist antecedent\/}.

A  (\emph{generalised}) \emph{Sahlqvist implication} is of the form $\phi\to\psi$, where $\phi$ is a 
(generalised) Sahlqvist antecedent and $\psi$ is a positive formula. 
A (\emph{generalised}) \emph{Sahlqvist formula} is a formula that is built up from (generalised) Sahlqvist implications by freely applying
$\Bh$, $\Bv$ and $\land$, and by applying $\lor$ only between formulas that do not share any
propositional variables.

The (\emph{generalised}) \emph{Sahlqvist completeness theorem} says that every logic axiomatised by (generalised) Sahlqvist
formulas is canonical, and so Kripke complete. 
The (\emph{generalised}) \emph{Sahlqvist  correspondence theorem} says that every 
(generalised) Sahlqvist formula has a first-order correspondent.
(A first-order formula $A$ in the language having equality and binary predicate symbols $\Rh$ and $\Rv$
is called a \emph{correspondent} of a bimodal formula $\phi$, whenever for every frame $\F$,
$A$ is valid in $\F$ iff $\phi$ is valid in $\F$.) 
Kracht \cite{Kracht96} gives a syntactical description of
first-order correspondents of Sahlqvist formulas.
Kracht's characterisation is extended to generalised Sahlqvist formulas  
by Kikot  \cite{Kikot09}.


\subsubsection{Some unimodal logics}

The following well-known unimodal logics are mentioned in the paper:
\begin{align*}
\K & = \Log\{\mbox{all unimodal frames}\},\\
\Sfive & = \Log\{\mbox{all unimodal equivalence frames}\}.
\end{align*}
In order to avoid extensive use of $\times$, we denote the universal relation $W\times W$
on any non-empty set $W$ by $\urel{W}$.
By \eqref{rooted},
\[
\Sfive=\Log\bigl\{(W,\urel{W}) : \mbox{$W$ is a non-empty set}\bigr\}.
\]

Most of the paper is about two-dimensional modal product logics (see \S\ref{s:product} below) where one or both component
logics is the much-studied unimodal `logic of elsewhere'  $\Diff$
\cite{rijke92,Gargov&Goranko93,Gargov&Passy&Tinchev87}. This logic was introduced by 
Von Wright \cite{vonWright79} as the set  of unimodal formulas that are valid in all
\emph{difference frames}, that is, in frames $(W,\drel{W})$, where $\drel{W}$ is the
non-equality relation on some non-empty set $W$.
Segerberg \cite{Segerberg80} axiomatised $\Diff$ by the Sahlqvist formulas
%
\begin{align}
\label{symmf}
& p\to\Box\Diamond p,\\
\label{pstf}
& \Diamond\Diamond p\to (p\lor\Diamond p).
\end{align}
So an arbitrary frame for $\Diff$ is a \emph{pseudo-equivalence relation}, that is,
it may contain both reflexive and irreflexive points, but it is always symmetric and 
\emph{pseudo-transitive}:
\begin{equation}\label{pstrans}
\forall x, y, z\;\bigl( R(x,y)\land R(y,z)\to (x=z \lor R(x,z))\bigr). 
\end{equation}
In particular, equivalence relations are frames for $\Diff$, and so $\Diff\subseteq\Sfive$.
It is straightforward to see that every rooted frame $(W,R)$ for $\Diff$ is a p-morphic image of 
any  difference frame $(U,\drel{U})$ for which
\[
|U|\geq  2\cdot |\{w \in W : w\mbox{ is $R$-reflexive}\}| + |\{w \in W : w\mbox{ is $R$-irreflexive}\}|.
\]
In particular,
\begin{equation}\label{difftosfivepm}
\mbox{if $|U|\geq 2\cdot |W|$ then $(W,\urel{W})$ is a p-morphic image of $(U,\drel{U})$.}
\end{equation}
Note that one can express the \emph{universal\/}, the \emph{at least two} and the \emph{precisely one} modalities with the help of a difference modality:
\begin{align*}
\forall\phi:\quad& \phi\land\Box\phi,\\
\Diamond^{\geq 2}\phi:\quad& \Diamond(\phi\land\Diamond\phi),\\
\Diamond^{=1}\phi:\quad& (\phi\lor\Diamond\phi)\land\neg\Diamond(\phi\land\Diamond\phi).
\end{align*}
%


%
%


\subsubsection{Bimodal product frames and logics}\label{s:product}

Given unimodal frames 
$\F_\hh=(W_\hh,\Rh)$ and $\F_\vv=(W_\vv,\Rv)$, their (\emph{modal}) \emph{product} is defined to be
the bimodal frame
\[
\F_\hh\mprod\F_\vv= ( W_\hh\mprod W_\vv,\overline{R}_\hh,\overline{R}_\vv),
\]
where $W_\hh\mprod W_\vv$ is the Cartesian product of $W_\hh$ and $W_\vv$
and, for all $x,x'\in W_\hh$, $y,y'\in W_\vv$,
\begin{gather*}
(x,y) \overline{R}_\hh (x',y')\quad \text{ iff }\quad x\Rh x'\mbox{ and }y=y',\\
(x,y) \overline{R}_\vv (x',y')\quad \text{ iff }\quad y\Rv y' \mbox{ and }x=x'.
\end{gather*}
%
%
%
It is easy to see that both taking point-generated subframes and p-morphic images
commute with the product construction:
\begin{align}
\label{pgenprod}
& \mbox{For any $x_\hh$ in $\F_\hh$, $x_\vv$ in $\F_\vv$,
$(\F_\hh\mprod\F_\vv)^{(x_\hh,x_\vv)}=\F_\hh^{x_\hh}\mprod\F_\vv^{x_\vv}$.}\\
\nonumber
& \mbox{If $\F_i$ is a p-morphic image of $\G_i$ for $i\in\{\hh,\vv\}$}\\
\label{pmprod}
& \hspace*{3cm}\mbox{then $\F_\hh\mprod\F_\vv$ is a p-morphic image of $\G_\hh\mprod\G_\vv$.}
\end{align}

Given  Kripke complete unimodal logics $L_\hh$ and $L_\vv$ in the respective unimodal languages
having $\Dh,\Bh$ and $\Dv,\Bv$, 
their  \emph{product} is defined as the (Kripke complete) bimodal logic
\[
\Log\bigl\{\F_\hh\mprod \F_\vv : \mbox{$\F_i$ is a frame for $L_i$, for $i=\hh,\vv$}\bigr\}.
\]
%
%
%
%
%
%
%
%
%
In particular,
%
$\dxd=\Log\bigl\{\F\mprod\G : \F,\G\mbox{ are frames for $\Diff$}\bigr\}$.
%
We call a frame of the form $(U,\drel{U})\mprod(V,\drel{V})$, for some non-empty sets $U,V$,
a \emph{product of difference frames\/}.
Then, by \eqref{rooted}, \eqref{pgenprod} and \eqref{pmprod}, we have that 
%
\begin{equation}
\label{dxdbyprod}
\dxd=\Log\{\mbox{products of difference frames}\}.
\end{equation}
If $|U|=|V|>0$ then we call $(U,\drel{U})\mprod(V,\drel{V})$, 
a \emph{square product of difference frames\/}.
We define the `square' version of $\dxd$ as
%
\[
\dsqxd  =\Log\{\mbox{square products of difference frames}\}.
\]
Then by \eqref{dxdbyprod}, we have
\begin{equation}
\label{prodinsq}
\dxd\subseteq\dsqxd.
\end{equation}
As Theorem~\ref{t:sqnonfinax} shows, there is an infinite gap between these two logics.
(Note that it is easy to reduce the validity problem of both $\dxd$ and $\dsqxd$ to that of
two-variable first-order logic with counting, and so by the decidability of the latter  \cite{gor97}, both 
$\dxd$ and $\dsqxd$ are decidable.)

It is easy to see that the classes of (isomorphic copies of) products of difference frames and of
square products of difference frames are both closed under ultraproducts. Thus, by a general result
of \cite{goldblatt89}, both $\dxd$ and $\dsqxd$ are canonical logics. Note that while (isomorphic copies of) products of difference frames form a (finitely axiomatisable) elementary class by Corollary~\ref{co:prodelem} below,
it is not hard to show that the class of square products of difference frames is not closed under 
elementary equivalence,
and so is not elementary (cf.\ \cite[Thm.~4.1.12]{Chang&Keisler90}).

Let $\commf$ be the conjunction of the Sahlqvist formulas $(\Bh\Bv p\leftrightarrow\Bv\Bh p)$ and
\eqref{symmf}--\eqref{pstf} for both $\Bh$ and $\Bv$ (with the first-order correspondent 
saying that $\Rh$ and $\Rv$ are commuting pseudo-equivalence relations). 
Then the logic $\comm$ axiomatised by $\commf$ is canonical, and so Kripke complete.
It is
straightforward to see that $\commf$ is valid in every product of difference frames,
and so 
\begin{equation}\label{dxdcomm}
 \comm\subseteq\dxd.
 \end{equation}
As Theorem~\ref{co:nonfinax} shows, there is an infinite gap between these two logics.


\section{Rooted frames for $\comm$}\label{grids}

In this section, we have a closer look at rooted frames for $\comm$, that is,
rooted frames of the form $\F=(W,\Rh,\Rv)$, where
$\Rh$ and $\Rv$ are commuting pseudo-equivalence relations.

%

We begin with the simplest rooted frames of this kind.
A frame $\C=(C,\Rh,\Rv)$ is called a \emph{\cluster\/}, if $\drel{C}$ is a subset of $R_j$ for both $j=\hh,\vv$.
It is straightforward to see that a \cluster{} is a rooted frame for $\comm$.
For $j=\hh,\vv$, a point $c$ in $\C$ is called $R_j$-\emph{irreflexive} ($R_j$-\emph{reflexive})
if $c\neg R_jc$  ($cR_jc$) holds.
So there can be four kinds of points in a \cluster:
both $\Rh$- and $\Rv$-reflexive (denoted by \rr),
$\Rh$-irreflexive and  $\Rv$-reflexive \mbox{(\ir),}
$\Rh$-reflexive and  $\Rv$-irreflexive (\ri), and 
both $\Rh$- and $\Rv$-irreflexive (\ii).
We use $\ \ixx$ to indicate when a point is $\Rh$-irreflexive and it does not matter whether it is
$\Rv$-reflexive or $\Rv$-irreflexive.
Similarly, $\ \xxi$ will be used whenever a point is $\Rv$-irreflexive and it does not matter whether it is
$\Rh$-reflexive or $\Rh$-irreflexive.
(An example of a \cluster{} is depicted in Fig.~\ref{f:no}.)
In what follows we often identify a \cluster\ with its domain. In particular,
for every \cluster\ $\C=(C,\Rh,\Rv)$, we denote by $|\C|$ the cardinality of its domain. We also let 
\begin{align*}
\sizeh(\C) & =  2\cdot |\{w \in \C : w\mbox{ is $\Rh$-reflexive}\}| + |\{w \in \C : w\mbox{ is $\Rh$-irreflexive}\}|,
\\
\sizev(\C) & =  2\cdot |\{w \in \C : w\mbox{ is $\Rv$-reflexive}\}| + |\{w \in \C : w\mbox{ is $\Rv$-irreflexive}\}|.
\end{align*}
%


Next, let $\F=(W,\Rh,\Rv)$ be an arbitrary rooted frame for $\comm$, and let $\Rhp$ and $\Rvp$ be the respective reflexive closures of $\Rh$ and $\Rv$. It is easy to see that 
$\Rhp$ and $\Rvp$ are commuting equivalence relations.
We define an equivalence relation $\sim$ on $W$
by taking, for all $u,v\in W$,
\[
u\sim v\quad\mbox{iff}\quad u\Rhp v\mbox{ and  }u\Rvp v.
\]
For each $u\in W$, let $[u]$ denote its $\sim$-class, and let $W^\sim=\{[u] : u\in W\}$.
We say (with a slight abuse of notation)
that $\F$ \emph{is} (\emph{represented as}) \emph{a \grid} $(X,Y,g)$ whenever 
$g: X\mprod Y\to W^\sim$ is a bijection for some sets $X$, $Y$ 
such that the following hold for all $x\ne x'\in X$ and $y\ne y'\in Y$:
\begin{itemize}
\item[(gc1)]
$u\Rh v$ for all $u\in g(x,y)$ and $v\in g(x',y)$;
\item[(gc2)]
$u\Rv v$ for all $u\in g(x,y)$ and $v\in g(x,y')$.
\end{itemize}
Observe that a single \cluster{} is a special case of a \grid{} when $|X|=|Y|=1$.

Given two \grids{} $\F=(X,Y,g)$ and $\F^\star=(X^\star,Y^\star,g^\star)$, we say that
$\F$ is a \emph{subgrid of} $\F^\star$ if $X\subseteq X^\star$, $Y\subseteq Y^\star$, and
$g=g^\star|_{X\mprod Y}$.
For each $(x,y)\in X\mprod Y$, we will denote by $\F^{xy}$ the subgrid of 
$\bigl(\{x\},\{y\},g|_{\{x\}\mprod\{y\}}\bigr)$ of $\F$. 
Observe that $\F^{xy}$ is always a \cluster. For any \cluster{} $\C$, we say that $\F$ \emph{contains\/} $\C$,
if $\C$ is isomorphic to $\F^{xy}$ for some $x,y$.

Throughout, we draw \grids{} by depicting each \cluster{} as a rectangular box, depicting $X$ (and $\Rh$ between \clusters) horizontally and $Y$
(and $\Rv$ between \clusters) vertically; see, for example, Figs.~\ref{f:badgrids2} and \ref{f:badgrids1}.


%
%
%

%
%

\begin{lemma}\label{l:grid}
Every rooted frame $\F$ for $\comm$ is a \grid.
\end{lemma}

\begin{proof}
Suppose $\F=(W,\Rh,\Rv)$ is a rooted frame $\F$ for $\comm$, that is,
$\Rh$ and $\Rv$ are commuting pseudo-equivalence relations. 
Take any $r\in W$. 
As $\Rhp$ and $\Rvp$ are commuting equivalence relations,
it is easy to see that  $r$ is a root in $\F$, and
for all $v,w\in W$, if $r\Rhp v$ and $r\Rvp w$ then there is $u$ with $v\Rvp u$ and $w\Rhp u$.
So we let
\[
X=\bigl\{[v]\in W^\sim : r\Rhp v\bigr\}
\quad\mbox{ and }\quad
Y=\bigl\{[w]\in W^\sim : r\Rvp w\bigr\},
\]
and define a function $g:X\times Y\to W^\sim$ by taking, for all $[v]\in X$, $[w]\in Y$,
\[
g\bigl([v],[w]\bigr)=[u]
\quad\mbox{ iff }\quad
v\Rvp u\mbox{ and }w\Rhp u.
\]
As both $\Rhp$ and $\Rvp$ are equivalence relations, for all $v$, $w$, $u$, $v'$, $w'$, $u'$, if
$v\sim v'$, $w\sim w'$, $v\Rvp u$, $w\Rhp u$, $v'\Rvp u'$ and $w'\Rhp u'$ then $u\sim u'$ follows, and
so $g$ is well-defined. It is easy to see that $g$ is injective and both (gc1) and (gc2) hold.
Finally, we show that $g$ is surjective: Take some $[u]\in W^\sim$.  Then there exist
$n\in\nN$ and $u_0,\dots,u_n$ such that $u_0=r$, $u_n=u$ and for each $i<n$, either $u_i \Rhp u_{i+1}$
or $u_i\Rvp u_{i+1}$. As $\Rhp$ and $\Rvp$ are commuting equivalence relations,
it follows that there are $v,w$ such that 
$r\Rhp v\Rvp u$ and $r\Rvp w\Rhp u$, and so $[v]\in X$, $[w]\in Y$ and $g\bigl([v],[w]\bigr)=[u]$, as required.
\end{proof}

\begin{corollary}\label{co:prodelem}
For every frame $\F=(W,\Rh,\Rv)$, $\F$ is isomorphic to a product of difference frames iff
$\Rh$ and $\Rv$ are commuting irreflexive pseudo-equivalence relations and all \clusters{} in $\F$
are singletons.
\end{corollary}

Because of the proof-pattern described in \S\ref{proofmethod}, we are particularly interested
in those countable \grids{} that are p-morphic images of some product of difference frames.
The following lemma provides a general characterisation for them.

\begin{lemma}\label{l:fits}
A countable \grid{} $\F$ is a p-morphic image of a product of two difference frames iff $\F$ is such that
\begin{itemize}
\item
each of its \clusters{} is the  p-morphic image of a product of two
difference frames, and
\item
the sizes of the product preimages for each \cluster{} `fit'.
\end{itemize} 
More precisely, for any countable \grid{} $\F=(X,Y,g)$, we have the following:
\begin{itemize}
\item[{\rm (i)}]
If $h:(U,\drel{U})\mprod(V,\drel{V})\to\F$ is an onto p-morphism, then 
there exists a function $\xi_h:(X\cup Y)\to\bigl(\nNp\cup\{\infy\}\bigr)$ such that
for every $(x,y)\in X\times Y$, the \cluster{} $\F^{xy}$ in $\F$ is a p-morphic image of
$(U_{x},\drel{U_{x}})\mprod (V_{y},\drel{V_{y}})$ for some sets $U_{x},V_{y}$ with 
$|U_{x}|=\xi_h(x)$ and $|V_{y}|=\xi_h(y)$.

\item[{\rm (ii)}]
If $\xi:(X\cup Y)\to\bigl(\nNp\cup\{\infy\}\bigr)$ is a function such that
for every $(x,y)\in X\times Y$, the \cluster{} $\F^{xy}$ in $\F$ is a p-morphic image of
$(U_{xy},\drel{U_{xy}})\mprod (V_{xy},\drel{V_{xy}})$ for some sets $U_{xy},V_{xy}$ with 
$|U_{xy}|=\xi(x)$ and $|V_{xy}|=\xi(y)$, then there is an onto p-morphism
$h_\xi:(U,\drel{U})\mprod(V,\drel{V})\to\F$ for some sets $U,V$ with
$|U|=\sum_{x\in X}\xi(x)$ and $|V|=\sum_{y\in Y}\xi(y)$.
\end{itemize}
%
\end{lemma}

\begin{proof}
(i):
We let 
\begin{align*}
& U_x=\{u\in U: \mbox{there exist $y\in Y$, $v\in V$ with $h(u,v)\in\F^{xy}$}\},\ \ \mbox{for every $x\in X$,}\\
& V_y=\{v\in V: \mbox{there exist $x\in X$, $u\in U$ with $h(u,v)\in\F^{xy}$}\},\ \ \mbox{for every $y\in Y$}.
\end{align*}
Then it is straightforward to see that, for every $(x,y)\in X\times Y$, the restriction of
$h$ to $U_x\times V_y$ is a p-morphism onto $\F^{xy}$. So we can define $\xi_h$ by taking $\xi_h(x)=|U_x|$, for $x\in X$, and
$\xi_h(y)=|V_y|$, for $y\in Y$,  as required.

(ii):
For every $(x,y)\in X\times Y$,
suppose that $h^{xy}:(U_{xy},\drel{U_{xy}})\mprod (V_{xy},\drel{V_{xy}})\to\F^{xy}$
is an onto p-morphism.
As for every $x\in X$, we have $|U_{xy}|=\xi(x)=|U_{xy'}|$ for any $y,y'\in Y$, we may assume that
$U_{xy}$ and $U_{xy'}$ are the same set $U_x$. Similarly, for every $y\in Y$, we may assume that $V_{xy}$ and $V_{x'y}$ are
the same set $V_y$. We may also assume that all these sets are disjoint. Now let $U=\bigcup_{x\in X} U_x$,
$V=\bigcup_{y\in Y} V_y$, and let the function $h_\xi:U\times V\to \F$ be defined by taking $h_\xi(u,v)=h^{xy}(u,v)$
whenever $u\in U_x$ and $v\in V_y$. Then it is straightforward to check that $h_\xi$ is a p-morphism
from $(U,\drel{U})\mprod(V,\drel{V})$ onto $\F$.
\end{proof}


\subsection{`Good' and `bad' \clusters}\label{biclusters}

By Lemma~\ref{l:fits}, if a \grid{} is not the p-morphic image of a product of difference frames, 
then it is because its \clusters{} are not p-morphic images
of `fitting' product preimages.
In this subsection, we have a closer look at individual \clusters{} first: which of them can or cannot be obtained
as the p-morphic image of some product of difference frames, and what size-restrictions we have on possible 
product preimages.
We distinguish fifteen types of finite \clusters, depending on whether they contain $R_i$-reflexive points or not, for $i=\hh,\vv$ (see Table~\ref{tt:finclusters}).
In particular, finite \clusters{} of types (\bcno{1})--(\bcno{4}) will be called \emph{impossible \clusters} 
throughout. Lemma~\ref{l:clusterpm} below claims that every countable \cluster{} that is not 
impossible can be obtained as the p-morphic image of any product of difference frames 
validating some constraints.
(In \S\ref{sformulas} we will show that the converse of Lemma~\ref{l:clusterpm} also holds in the sense 
that whenever a countable \cluster{} $\C$ is a p-morphic image of a product of difference frames,
then $\C$ is not impossible, and the described constraints hold for the preimage product frame, see Corollary~\ref{co:hassolution}.)

\begin{lemma}\label{l:clusterpm}
\begin{itemize}
\item[{\rm (i)}]
Every countably infinite \cluster{} $\C$ is a p-morphic image of 
\mbox{$(\omega,\drel{\omega})\mprod(\omega,\drel{\omega})$.}

\item[{\rm (ii)}]
For every finite \cluster{} $\C$, if $\C$ is not an impossible \cluster, then
$\C$ is the p-morphic image of $(U,\drel{U})\mprod(V,\drel{V})$ for any countable sets
$U$, $V$ such that
the constraints of Table~\ref{tt:finclusters} hold for $x=|U|$ and $y=|V|$.
\end{itemize}
\end{lemma}
%
%
\begin{table}
\begin{center}
\begin{tabular}{|c|cccc|l|}
\hline
type of $\C$ & \ii & \ri & \ir & \rr & constraints on $x\times y$ size\\
&&&&& p-morphic preimage\\
\hline\hline
\mbox{(\bcno{1})} & -- & + & + & - & no such $x,y$\\[3pt]
\hline
\mbox{(\bcno{2})} & + & -- & + & - & no such $x,y$\\[3pt]
\hline
\mbox{(\bcno{3})} & + & + & -- & -- & no such $x,y$\\[3pt]
\hline
\mbox{(\bcno{4})} & + & + & + & -- & no such $x,y$\\[3pt]
\hline
\mbox{(\bcinf{1})} & -- & + & + & + & $x=\infy,\ y=\infy$\\[3pt]
\hline
\mbox{(\bcinf{2})} & + & -- & + & + & $x=\infy,\ y=\infy$\\[3pt]
\hline
\mbox{(\bcinf{3})} & + & + & -- & + & $x=\infy,\ y=\infy$\\[3pt]
\hline
\mbox{(\bcinf{4})} & + & + & + & + & $x=\infy,\ y=\infy$\\[3pt]
\hline
\mbox{(\bchvsw)} & -- & + & -- & + & $x\geq 2y,\ y\geq\sizev(\C)$\\
$x\to^2 y$ &&&&& \\[3pt]
\hline
\mbox{(\bcvhsw)} & -- & -- & + & + & $x\geq\sizeh(\C),\ y\geq 2x$\\
$y\to^2 x$ &&&&& \\[3pt]
\hline
\mbox{(\bceqsw)} & + & -- & -- & + & $x\geq y,\ y\geq x$,\\
$x\leftrightarrow^1 y$ &&&&& $x\geq\sizeh(\C),\ y\geq\sizev(\C)$\\[3pt]
\hline
\mbox{(\bchstrict)} & -- & -- & + & -- & $x=\sizeh(\C)=|\C|$,\\
&&&&& $y\geq\sizev(\C)=2\cdot |\C|$\\[3pt]
\hline 
\mbox{(\bcvstrict)} & -- & + & -- & -- & $x\geq\sizeh(\C)=2\cdot |\C|$,\\
&&&&& $y=\sizev(\C)=|\C|$\\[3pt]
\hline
\mbox{(\bchvstrict)} & + & -- & -- & -- & $x=\sizeh(\C)=|\C|$,\\ 
&&&&& $y=\sizev(\C)=|\C|$\\[3pt]
\hline
\mbox{(\bcfree)} & -- & -- & -- & + & $x\geq\sizeh(\C)=2\cdot |\C|$,\\ 
&&&&& $y\geq\sizev(\C)=2\cdot |\C|$\\[3pt]
\hline 
\end{tabular}

\bigskip
\begin{tabular}{l|l}
\underline{Notation in table:} & \underline{Terminology:}\\[4pt]
\rr\,=\,both $\Rh$- and $\Rv$-reflexive & \emph{impossible} \cluster
s: (\bcno{1})--(\bcno{4})\\
\ri\,=\,$\Rh$-reflexive, $\Rv$-irreflexive & \emph{infinity} \clusters: (\bcinf{1})--(\bcinf{4})\\
\ir\,=\,$\Rh$-irreflexive, $\Rv$-reflexive & \emph{switch} \clusters: (\bchvsw),\\
\ii\,=\,both $\Rh$- and $\Rv$-irreflexive & \hspace*{2.2cm}(\bcvhsw), and (\bceqsw) \\
+\;=\,there is such a point in $\C$ & \emph{strict} \clusters: (\bchstrict),\\
--\;=\,there isn't such a point in $\C\quad$ &  \hspace*{1.6cm}(\bcvstrict), and (\bchvstrict) 
\end{tabular}
\end{center}
\caption{Possible finite \clusters{} in a \grid.}\label{tt:finclusters}
\end{table}

\begin{proof}
We begin with a useful tool.
Given a \cluster{} $\C=(C,\Rh,\Rv)$,
we define an $\C$-\emph{network} to be a homomorphism 
$f:(U,\drel{U})\mprod (V,\drel{V})\to\C$, for some finite non-empty sets $U$ and $V$.
Given $\C$-networks $f_1:(U_1,\drel{U_1})\mprod (V_1,\drel{V_1})\to\C$ and 
$f_2:(U_2,\drel{U_2})\mprod (V_2,\drel{V_2})\to\C$, we write $f_1\subseteq f_2$ whenever
$U_1\subseteq U_2$, $V_1\subseteq V_2$ and $f_2|_{U_1\mprod V_1}=f_1$.
We define a \emph{game} $\mathbb{G}(\C)$ between two players, $\forall$ and $\exists$.
They build a countable sequence of $\C$-networks
$
f_0\subseteq f_1\subseteq \dots\subseteq f_k\subseteq\dots .
$
In round~0, $\forall$ picks any point $r$ in $\C$, and $\exists$ responds with 
$U_0=\{u_0\}$, $V_0=\{v_0\}$, and $f_0(u_0,v_0)=r$.
In round $k$ ($k\in\nNp$), some sequence $f_0\subseteq\dots\subseteq f_{k-1}$  of $\C$-networks
has already been built. $\forall$ picks a pair $(c^\ast,z)$ where
$c^\ast\in C$ and $z\in U_{k-1}\cup V_{k-1}$.
%
%
There are two cases:
\begin{itemize}
\item
$z\in V_{k-1}$. Then $\exists$ can respond in two ways:
 If either $c^\ast$ is $\Rh$-irreflexive and there is $u\in U_{k-1}$ with $f_{k-1}(u,z)=c^\ast$, 
 or $c^\ast$ is $\Rh$-reflexive and there are $u,u'\in U_{k-1}$, $u\ne u'$ with $f_{k-1}(u,z)=f_{k-1}(u',z)=c^\ast$, 
then she responds with $f_{k}=f_{k-1}$. Otherwise, 
she responds (if she can) with some \mbox{$\C$-network} $f_k\supseteq f_{k-1}$ 
such that $V_k=V_{k-1}$, $U_k=U_{k-1}\cup\{u^+\}$ for some fresh point $u^+$,
and $f_k(u^+,z)=c^\ast$.
In other words, she needs to find a sequence $\bigl(c_v : v\in V_{k-1}-\{z\}\bigr)$ of points in $\C$ such that,
for every $v\in V_{k-1}-\{z\}$, 
%

%
\begin{itemize}
\item
$c_v\Rh f_{k-1}(u,v)$ for every $u\in U_{k-1}$,
\item
$c_v\Rv c^\ast$, and $c_v\Rv c_{v'}$ for every $v'\in V_{k-1}-\{z\}$, $v'\ne v$.
\end{itemize}

\item
$z\in U_{k-1}$. Then again, $\exists$ can respond in two ways:
If either $c^\ast$ is $\Rv$-irreflexive and there is $v\in V_{k-1}$ with $f_{k-1}(z,v)=c^\ast$, 
 or $c^\ast$ is $\Rv$-reflexive and there are $v,v'\in V_{k-1}$, $v\ne v'$ with $f_{k-1}(z,v)=f_{k-1}(z,v')=c^\ast$, 
then she responds with $f_{k}=f_{k-1}$. Otherwise, 
she responds (if she can) with some \mbox{$\C$-network} $f_k\supseteq f_{k-1}$ 
such that $U_k=U_{k-1}$, $V_k=V_{k-1}\cup\{v^+\}$ for some fresh point $v^+$,
and $f_k(z,v^+)=c^\ast$.
In other words, 
she needs to find a sequence $\bigl(c_u : u\in U_{k-1}-\{z\}\bigr)$ of points in $\C$ such that,
for every $u\in U_{k-1}-\{z\}$, 
\begin{itemize}
\item
$c_u\Rv f_{k-1}(u,v)$ for every $v\in V_{k-1}$,

\item
$c_u\Rh c^\ast$, and $c_u\Rh c_{u'}$ for every $u'\in U_{k-1}-\{z\}$, $u'\ne u$.
\end{itemize}
\end{itemize}
If $\exists$ can respond in each round $k$ for $k\in\nN$ then \emph{she wins the play}. We say that
$\exists$ \emph{has a winning strategy in} $\mathbb{G}(\C)$ if she can win all plays, whatever moves
$\forall$ takes in the rounds.

\begin{lclaim}\label{c:game}
For every countable \cluster{} $\C$, 
player $\exists$ has a winning strategy in $\mathbb{G}(\C)$ iff
$\C$ is the p-morphic image of a product of two countable difference frames.
\end{lclaim}

\begin{proof}
On the one hand, it is easy to see that $\exists$ can use a p-morphism from
a product of two difference frames onto $\C$ to determine her winning strategy in $\mathbb{G}(\C)$.

For the other direction,
consider a play of the game $\mathbb{G}(\C)$ with the following property:
For all $k\in\nN$, $(u,v)\in U_k\mprod V_k$, $c^\ast\in C$, there exist $\ell_\hh,\ell_\vv\in\nN$ such
that $\ell_\hh,\ell_\vv> k$, $\forall$ picks $(c^\ast,u)$ in round $\ell_\hh$, and 
$\forall$ picks $(c^\ast,v)$ in round $\ell_\vv$ (since $\C$ is countable, he can do these). 
If $\exists$ uses her strategy, then the union $f:(U,\drel{U})\mprod (V,\drel{V})\to\C$ of 
the constructed countable ascending chain of $\C$-networks is a p-morphism.
Indeed,
take some $u^\ast\in U$, $v^\ast\in V$, $c^\ast\in C$ such that, say, $f(u^\ast,v^\ast)\Rh c^\ast$.
We need to find some $u\in U$ such that $u\ne u^\ast$ and $f(u,v^\ast)=c^\ast$.
Let $k$ be such that $(u^\ast,v^\ast)\in U_k\mprod V_k$ and consider round $\ell> k$ 
when $\forall$ picks $(c^\ast,v^\ast)$. There are three cases:
If $c^\ast$ is $\Rh$-irreflexive and there is $u\in U_{\ell-1}$ with $f_{\ell-1}(u,v^\ast)=c^\ast$,
then $f(u,v^\ast)\ne f(u^\ast,v^\ast)$, and so $u\ne u^\ast$.
If $c^\ast$ is $\Rh$-reflexive and there are $u,u'\in U_{\ell-1}$ with $u\ne u'$ and $f_{\ell-1}(u,v^\ast)=f_{\ell-1}(u',v^\ast)=c^\ast$, then either $u\ne u^\ast$ or $u'\ne u^\ast$.
Otherwise, there is $u^+\in U_\ell-U_{\ell-1}$ with $f(u^+,v^\ast)=c^\ast$. As $u^\ast\in U_{\ell-1}$,
it follows that $u^+\ne u^\ast$.
%
\end{proof}

%


Now we can complete the proof of Lemma~\ref{l:clusterpm}.

Item (i): Consider the game $\mathbb{G}(\C)$.
As $\C$ is infinite and the constructed networks in each round of a play in the game have finite domains,
$\exists$ can always respond according to the rules, and so she has a winning strategy in $\mathbb{G}(\C)$. So by Claim~\ref{c:game}, $\C$ is the p-morphic image of a product of two countable difference frames. As $\C$ contains infinitely many points that are $R_j$-connected, for both $j=\hh$ and
$j=\vv$, both components in the product preimage must be infinite.

Item (ii):
Suppose first that $\C$ is a finite infinity \cluster. Then $\C$ contains at least one $\ \rr$ point $a$.
So in every round of a play in the game $\mathbb{G}(\C)$, $\exists$ can always 
use a full $a$-sequence as the sequence of points needed in her response, giving her
a winning strategy in $\mathbb{G}(\C)$. So by Claim~\ref{c:game}, there exists an 
onto p-morphism $f:(U,\drel{U})\mprod(V,\drel{V})\to\C$, for some countable sets $U$ and $V$. 
We claim that
\begin{align}
\label{xtwoy}
& \mbox{if $\C$ contains a $\ \ri$ point then $|U|\geq 2\cdot |V|$,}\\
\label{ytwox}
& \mbox{if $\C$ contains a $\ \ir$ point then $|V|\geq 2\cdot |U|$, and}\\
\label{xsamey}
& \mbox{if $\C$ contains a $\ \ii$ point then $|U|= |V|$.}
\end{align}
Indeed, for \eqref{xtwoy}, let $a$ be a $\ \ri$ point in $\C$. Then for every
$v\in V$, there exist $u_v,u_v'\in U$, $u_v\ne u_v'$ such that
$f(u_v,v)=f(u_v',v)=a$. Also, if $v_1\ne v_2$ then $u_{v_1}$, $u_{v_1}'$, $u_{v_2}$, $u_{v_2}'$
must be four distint points, and so \eqref{xtwoy} follows. The proof of \eqref{ytwox} is similar.
For \eqref{xsamey}, let $b$ be a $\ \ii$ point in $\C$. Then for every
$v\in V$, there exist $u_v\in U$ such that
$f(u_v,v)=b$. Also, if $v_1\ne v_2$ then $u_{v_1}\ne u_{v_2}$ must hold, and so $|U|\geq |V|$
follows. We can show $|V|\geq |U|$ similarly, and so we obtain \eqref{xsamey}.
Now if $\C$ is an infinity \cluster, then the infinity of both $U$ and $V$ follows from
\eqref{xtwoy}--\eqref{xsamey}.

Suppose that $\C$ is an $n$-element (\bchvstrict) \cluster, and take any $n$-element sets $U$ and 
$V$. Let $f: U\mprod V\to\C$ be any function such that the $n\mprod n$-matrix 
$\bigl(f(u,v)\bigr)_{(u,v)\in U\mprod V}$ is a Latin square over the elements of $\C$
(that is, each element of $\C$ occurs exactly once in each row and exactly once in each column).
It is straightforward to check that such an $f$ is a p-morphism from $(U,\drel{U})\mprod(V,\drel{V})$
onto $\C$.

Suppose that $\C$ is an $n$-element (\bchstrict) \cluster, and take any $n$-element sets $U$ and 
$V^-$. Let $f:U\mprod V^-\to\C$ be any function such that the $n\mprod n$-matrix 
$\bigl(f(u,v)\bigr)_{(u,v)\in U\mprod V^-}$ is a Latin square over the elements of $\C$.
It is straightforward to check that such an $f$ is a p-morphism from $(U,\drel{U})\mprod(V^-,\urel{V^-})$
onto $\C$. Now take any set $V$ with $|V|\geq 2n$. By \eqref{difftosfivepm} and \eqref{pmprod}, we obtain that $\C$ is the p-morphic image of $(U,\drel{U})\mprod(V,\drel{V})$.
The proof for (\bcvstrict) \clusters{} is similar.

The following claim will also be used in \S\ref{nonfinax}:

\begin{lclaim}\label{c:neededininf}
If $\C$ is a {\rm (\bchvsw)} \cluster{} containing $n$ $\ \ri$ points and $m$ $\ \rr$ points,
then for all sets $U'$, $V$ with $|U'|=|V|\geq n+2m$ there exists a p-morphism from $(U',\urel{U'})\mprod(V,\drel{V})$ onto $\C$.
\end{lclaim}

\begin{proof}
Take any sets $U'$, $V$ such that $|U'|=|V|=N$ for some $N\geq n+2m$.
Let $S$ be an $N$-element set that contains all the $\ \ri$ points of $\C$, and at least two distinct `copies'
of each $\ \rr$ point in $\C$. Then let $f:U'\mprod V\to\C$ be any function such that the $N\mprod N$-matrix 
$\bigl(f(u,v)\bigr)_{(u,v)\in U'\mprod V}$ is a Latin square over the elements of $S$.
It is straightforward to check that 
such an $f$ is a p-morphism from $(U',\urel{U'})\mprod(V,\drel{V})$ onto $\C$.
\end{proof}

Now suppose that $\C$ is a (\bchvsw) \cluster{} containing $n$ $\ \ri$ points and $m$ $\ \rr$ points,
and take any sets $U,V$ with $|U|\geq 2\cdot |V|$ and $V\geq n+2m$. By \eqref{difftosfivepm}, \eqref{pmprod}, 
and Claim~\ref{c:neededininf}, we obtain that $\C$ is the p-morphic image of $(U,\drel{U})\mprod(V,\drel{V})$.
The proof for (\bcvhsw) \clusters{} is similar.

Suppose that $\C$ is a (\bceqsw) \cluster{} containing $n$ $\ \ii$ points and $m$ $\ \rr$ points,
and take any sets $U$, $V$ such that $|U|=|V|=N$ for some $N\geq n+2m$.
Let $S$ be an $N$-element set that contains all the $\ \ii$ points of $\C$, and at least two distinct `copies'
of each $\ \rr$ point in $\C$. Then let $f:U\mprod V\to\C$ be any function such that the $N\mprod N$-matrix 
$\bigl(f(u,v)\bigr)_{(u,v)\in U\mprod V}$ is a Latin square over the elements of $S$.
It is straightforward to check that such an $f$ is a p-morphism from $(U,\drel{U})\mprod(V,\drel{V})$
onto $\C$. 

Finally, suppose that $\C$ is an $n$-element (\bcfree) \cluster, and take any $n$-element sets $U^-$ and 
$V^-$. Let $f:U^-\mprod V^-\to\C$ be any function such that the $n\mprod n$-matrix 
$\bigl(f(u,v)\bigr)_{(u,v)\in U^-\mprod V^-}$ is a Latin square over the elements of $\C$.
It is straightforward to check that such an $f$ is a p-morphism from $(U^-,\urel{U^-})\mprod(V^-,\urel{V^-})$
onto $\C$. Now take any sets $U$, $V$ with $|U|\geq 2n$ and $|U|\geq 2n$. 
 By \eqref{difftosfivepm} and \eqref{pmprod}, we obtain that $\C$ is the p-morphic image of $(U,\drel{U})\mprod(V,\drel{V})$.
\end{proof}


\section{Non-finite axiomatisability}\label{nonfinax}

In this section we prove Theorems~\ref{co:nonfinax} and \ref{t:sqnonfinax},
using the proof pattern described in \S\ref{proofmethodnonfin}. 
We will also use a result of
\cite[Cor.~2.5]{Kurucz10}, saying that
if $\Cc$ is closed under ultraproducts and point-generated subframes, then 
\begin{multline}
\mbox{for every finite frame $\F$,\ \ $\F$ is a frame for $\Log\Cc$}\\
\label{finframe}
\mbox{iff\quad $\F$ is the p-morphic image of some frame in $\Cc$.}
\end{multline}
%

In order to prove Theorem~\ref{co:nonfinax}, we show the following more general statement, which
also generalises some results of \cite{KudinovSS12}:

\begin{theorem}\label{t:nonfinax}
Let $L$ be any bimodal logic such that
\begin{itemize}
\item
$L$ contains $\K\mprod\Diff$, and
\item
for every $k\in\nNp$ there are $U$, $V$, such that $|V|\geq k$, $|U|\geq 2\cdot|V|$ and
$(U,\urel{U})\mprod(V,\drel{V})$ is a frame for $L$.
\end{itemize}
Then $L$ is not axiomatisable using finitely many propositional variables.
\end{theorem}

\begin{proof}
For every $k\in\nNp$, $k\geq 2$, take the \grids\ $\F_k$ and $\G_k$ depicted in Fig.~\ref{f:nonfin1}.

\begin{figure}[ht]
\begin{center}
\setlength{\unitlength}{.03cm}
\begin{picture}(105,130)

\multiput(-10,0)(0,110){2}{\line(1,0){120}}
\multiput(-10,0)(60,0){3}{\line(0,1){110}}

\multiput(30,37)(0,8){3}{\circle*{1}}
\multiput(70,37)(0,8){3}{\circle*{1}}
\put(20,15){\ri}
\put(20,65){\ri}
\put(60,15){\ri}
\put(60,65){\ri}
\put(60,80){\ri}

\put(45,120){$\F_k$}
\put(-5,97){$\C_1(k)$}
\put(77,97){$\C_2(k)$}

\put(9,13){\begin{turn}{90}$\overbrace{\hspace*{1.8cm}}$\end{turn}}
\put(0,36){${}^k$}

\put(76,12){\begin{turn}{90}$\underbrace{\hspace*{2.3cm}}$\end{turn}}
\put(88,44){${}^{k+1}$}
\end{picture}
\hspace*{4cm}
\begin{picture}(105,130)
\multiput(-10,0)(0,110){2}{\line(1,0){120}}
\multiput(-10,0)(60,0){3}{\line(0,1){110}}

\multiput(30,37)(0,8){3}{\circle*{1}}
\multiput(70,37)(0,8){3}{\circle*{1}}
\put(20,15){\ri}
\put(20,65){\ri}
\put(20,80){\rr}
\put(60,15){\ri}
\put(60,65){\ri}
\put(60,80){\rr}

\put(45,120){$\G_k$}
\put(-5,97){$\C'_1(k)$}
\put(77,97){$\C'_2(k)$}

\put(11,13){\begin{turn}{90}$\overbrace{\hspace*{1.8cm}}$\end{turn}}
\put(-6,36){${}^{k-2}$}

\put(76,13){\begin{turn}{90}$\underbrace{\hspace*{1.8cm}}$\end{turn}}
\put(88,36){${}^{k-2}$}
\end{picture}
\end{center}
\caption{The \grids\ $\F_k$ and $\G_k$.}\label{f:nonfin1}
\end{figure}


\begin{tlemma}\label{l:nonfinax}
\begin{itemize}
\item[{\rm (i)}]
$\F_k$ is not a frame for $\K\mprod\Diff$.

\item[{\rm (ii)}]
$\G_k$ is a p-morphic image of $(U,\urel{U})\mprod(V,\drel{V})$, whenever $|V|\geq k$ and $|U|\geq 2\cdot|V|$.

\item[{\rm (iii)}]
If $k,m\in\nN$ and $k\geq 2^{m+1}$, then for every $m$-generated model $\M$ over $\F_k$ there is some model $\N$ over $\G_k$ that is a p-morphic image of $\M$.
\end{itemize}
%
\end{tlemma}

\begin{proof}
(i):
By definition, $\K\mprod\Diff=\Log\Cc$, where
\[
\Cc=\bigl\{\F_\hh\mprod\F_\vv : \mbox{$\F_\vv$ is a pseudo-equivalence frame}\bigr\}.
\]
Using \eqref{pgenprod} and the fact that the ultraproduct construction also commutes with
the modal product construction, 
it is not hard to see that $\Cc$ is closed 
under point-generated subframes and ultraproducts.
Therefore, by \eqref{finframe},
it is enough to show that $\F_k=(W,\Rh,\Rv)$ is not the p-morphic image of any
$(W_\hh,\Qh)\mprod(W_\vv,\Qv)$, where $\Qv$ is a pseudo-equivalence relation.
Suppose to the contrary that there is an onto p-morphism $f:(W_\hh,\Qh)\mprod(W_\vv,\Qv)\to\F_k$.
Take any point $a$ in the $k$-element \cluster{} $\C_1(k)$, and 
 any point $b$ in the $k+1$-element \cluster{} $\C_2(k)$. As $a\Rh b$,
 there are $x_0,x_1\in W_\hh$, $y_0\in W_\vv$ such that $x_0\Qh x_1$,
$f(x_0,y_0)=a$ and $f(x_1,y_0)=b$. 
As there are $k$ other points in $\C_2(k)$, each of them is $\Rv$-related to $b$, there exist $y_1,\dots,y_{k}\in W_\vv$
such that $y_0\Qv y_i$ for all $0<i\leq k$ and $y_i\ne y_j$ for all $i\ne j\leq k$. As $\Qv$ is a pseudo-equivalence relation, it follows that $y_i\Qv y_j$ for all $i\ne j\leq k$.
Then $f(x_0,y_i)\Rv f(x_0,y_j)$ must hold, for all $i\ne j\leq k$. As 
every point in $\C_1(k)$ is $\Rv$-irreflexive, this is not possible by the pigeonhole principle.

(ii):
Take any sets $U$, $V$, with $|V|\geq k$ and $|U|\geq 2\cdot |V|$, and choose two disjoint subsets 
$U_1$ and $U_2$ of $U$ such that $|U_1|=|U_2|= |V|$.
Observe that each of the two \clusters{} $\C'_i(k)$ in $\G_k$ is a (\bchvsw) \cluster, containing
$k-2$ $\ \ri$ points and one $\ \rr$ point
(cf.\ Fig.~\ref{f:nonfin1} and Table~\ref{tt:finclusters}).
So by Claim~\ref{c:neededininf}, 
there exist onto \mbox{p-morphisms} $h_i: (U_i,\urel{U_i})\mprod(V,\drel{V})\to\C'_i(k)$, for $i=1,2$.
Let $U'=U_1\cup U_2$, and 
define a function $h$ from $U'\mprod V$ to $\G_k$ by taking, for all $u\in U'$, $v\in V$,
\[
h(u,v)=\left\{
\begin{array}{ll}
h_1(u,v) & \mbox{if $u\in U_1$},\\
h_2(u,v), & \mbox{if $u\in U_2$}.
 \end{array}
 \right.
 \]
Then it is easy to check that $h$ is a p-morphism from $(U',\urel{U'})\mprod(V,\drel{V})$ onto
$\G_k$. As $(U',\urel{U'})$ is a p-morphic image of $(U,\urel{U})$, it follows from \eqref{pmprod}
that $\G_k$ is a p-morphic image of $(U,\urel{U})\mprod(V,\drel{V})$.

(iii):
Let $\M$ be a model over $\F_k$ such that if $\M(p)\ne\emptyset$ for some propositional variable $p$
then $p=p_i$ for some $i<m$. We define two equivalence relations $\sim_1$ and $\sim_2$ on 
$\C_1(k)$ and on $\C_2(k)$, respectively, by taking, for all $a,a'$ in $\C_1(k)$ and $b,b'$ in $\C_2(k)$,
\begin{align*}
a\sim_1 a' & \qquad\mbox{iff}\qquad a\in\M(p_i)\Leftrightarrow a'\in\M(p_i),\quad\mbox{for all $i<m$},\\
b\sim_2 b' & \qquad\mbox{iff}\qquad b\in\M(p_i)\Leftrightarrow b'\in\M(p_i),\quad\mbox{for all $i<m$}.
\end{align*}
As $k\geq 2^{m+1}$, by the generalised pigeonhole principle, 
there is a $\sim_1$-class containing at least two points $a,a'$, and
there is a $\sim_2$-class containing at least three points $b,b',b''$.
Now define a function $h$ from $\F_k$ onto $\G_k$ by
\begin{itemize}
\item
mapping $a$ and $a'$ to the 
$\ \rr$ point in $\C'_1(k)$, 
\item
mapping the remaining $k-2$ points in $\C_1(k)$ to the $k-2$ distinct 
$\ \ri$ points in $\C'_1(k)$,
\item
mapping $b,b'$ and $b''$ to the 
$\ \rr$ point in $\C'_2(k)$, 
\item
mapping the remaining $k-2$ points in $\C_2(k)$ to the $k-2$ distinct 
$\ \ri$ points in $\C'_2(k)$.

\end{itemize}
It is easy to check that $h$ is a p-morphism from $\F_k$ onto $\G_k$. Now define a model $\N$ over
$\G_k$ by taking,
for any propositional variable $p$, $\N(p)=\{c : h(a)=c\mbox{ for some $a\in\M(p)$}\}$.
By the above, $h$ is a p-morphism from $\M$ onto $\N$.
\end{proof}

Now the proof of Theorem~\ref{t:nonfinax} can be completed as follows.
Suppose to the contrary that $\Sigma$ axiomatises $L$ and $\Sigma$ contains only 
$m$ propositional variables, for some $m\in\nN$.
Let $k\geq 2^{m+1}$ and let $\M$ be an arbitrary model over $\F_k$. Let
 $\M_m$ be another model over $\F_k$ that is the same
as $\M$ on propositional variables occurring in $\Sigma$, and
$\emptyset$ otherwise. Then $\M_m$ is clearly $m$-generated
and $\M_m\models\Sigma$ iff $\M\models\Sigma$.
Also, by Lemma~\ref{l:nonfinax}~(iii) there is a model $\N$ over $\G_k$ that is a p-morphic image 
of $\M_m$. 
As there are $U$, $V$, such that $|V|\geq k$, $|U|\geq 4\cdot|V|$ and
$(U,\urel{U})\mprod(V,\drel{V})$ is a frame for $L$,
by Lemma~\ref{l:nonfinax}~(ii) $\G_k$ is a frame for $L$. Thus, $\N\models L$, and so $\M_m\models L$.
As $\Sigma\subseteq L$, we obtain $\M_m\models\Sigma$, and so $\M\models\Sigma$.
As this holds for any model $\M$ over $\F_k$, $\F_k$ is a frame for $\Sigma$. 
Therefore, $\Log\{\F_k\}$ is a bimodal logic containing $\Sigma$, and so we have  that
$\F_k$ is a frame for $L$. As $L$ contains $\K\mprod\Diff$, this implies
that $\F_k$ is a frame for $\K\mprod\Diff$, contradicting Lemma~\ref{l:nonfinax}~(i).
\end{proof}


\paragraph{Proof of Theorem~\ref{t:sqnonfinax}}
For every $k\in\nNp$, take the \grids{} $\G_k$ and $\Hh_k$ from Figs.~\ref{f:nonfin1} and 
\ref{f:nonfin2}, respectively.

\begin{figure}[ht]
\begin{center}
\setlength{\unitlength}{.025cm}
\begin{picture}(105,85)
\multiput(-10,0)(0,85){2}{\line(1,0){120}}
\multiput(-10,0)(60,0){3}{\line(0,1){85}}

\multiput(30,37)(0,8){3}{\circle*{1}}
\multiput(72,37)(0,8){3}{\circle*{1}}
\put(20,15){\ri}
\put(20,65){\ri}
\put(62,15){\ri}
\put(62,65){\ri}


\put(9,13){\begin{turn}{90}$\overbrace{\hspace*{1.5cm}}$\end{turn}}
\put(0,36){${}^k$}

\put(80,13){\begin{turn}{90}$\underbrace{\hspace*{1.5cm}}$\end{turn}}
\put(93,36){${}^{k}$}
\end{picture}
\end{center}
\caption{The \grid{} $\Hh_k$.}\label{f:nonfin2}
\end{figure}

\begin{tlemma}\label{l:sqnonfinax}
\begin{itemize}
\item[{\rm (i)}]
$\Hh_k$ is not a frame for $\dsqxd$.

\item[{\rm (ii)}]
$\G_k$ is a p-morphic image of $(\omega,\drel{\omega})\mprod(\omega,\drel{\omega})$.

\item[{\rm (iii)}]
If $k,m\in\nN$ and $k>2^{m}$, then for every $m$-generated model $\M$ over $\Hh_k$ there is some model $\N$ over $\G_k$ that is a p-morphic image of $\M$.
\end{itemize}
\end{tlemma}

\begin{proof}
(i):
By definition, $\dsqxd=\Log\{$square products of difference frames$\}$.
Using \eqref{pgenprod} and the fact that the ultraproduct construction also commutes with
the modal product construction, 
it is not hard to see that the class of all square products of difference frames is closed 
under  point-generated subframes and ultraproducts.
Therefore, by \eqref{finframe},
it is enough to show that $\Hh_k=(W,\Rh,\Rv)$ is not the p-morphic image of a square product $(U,\drel{U})\mprod (V,\drel{V})$ for any sets $U,V$ with $|U|=|V|>0$.
Suppose indirectly that it is.
As every point in $\Hh_k$ is $\Rv$-irreflexive, 
$|V|=k$ must hold. On the other hand, as $\Rh$ is the universal relation in $\Hh_k$, we must have 
$|U|\geq 2k$, contradicting $|U|=|V|>0$.

Item (ii) follows from Lemma~\ref{l:nonfinax}~(ii), \eqref{difftosfivepm} and \eqref{pmprod}.
The proof of item (iii)
is similar to that of Lemma~\ref{l:nonfinax}~(iii).
\end{proof}

Now the proof of Theorem~\ref{t:sqnonfinax} can be completed similarly to that of Theorem~\ref{t:nonfinax},
using Lemma~\ref{l:sqnonfinax} in place of Lemma~\ref{l:nonfinax}.



\section{Infinite canonical axiomatisation for $\dxd$}\label{dxd}

In this section we prove Theorem~\ref{t:axdxd}
using the proof pattern described in \S\ref{proofmethod} (for the class $\Cc$ of all products of difference frames).
So we will define a recursive
set $\Sigma_{\dxd}$ of Sahlqvist formulas, and prove that the following hold:
\begin{enumerate}
\item
All formulas in $\Sigma_{\dxd}$ are valid in every product of difference frames.
\item
For every countable rooted frame $\F$ that is not the p-morphic image of some product of 
difference frames, there is some $\phi_\F\in\Sigma_{\dxd}$ such that $\phi_\F$ is not valid in $\F$.
\end{enumerate}
To begin with, if $\F$ is a countable rooted frame such that $\F\not\models\comm$, then 
$\F\not\models\commf$, and so we let
$\phi_\F=\commf\in\Sigma_{\dxd}$.
So from now on we assume that $\F\models\comm$, and so $\F$
is a \grid{} by Lemma~\ref{l:grid}.
We call a countable \grid{} $\F$ \emph{bad} if it is not the p-morphic image of a product of difference frames.

In \S\ref{gbgrids} below we discuss two kinds of `finitary reasons' for a countable \grid{} being bad,
and prove that these are the only such reasons.
Then in \S\ref{sformulas} we provide the Sahlqvist formulas in $\Sigma_{\dxd}$ `eliminating' these reasons.


\subsection{Bad \grids}\label{gbgrids}

The first reason for a countable \grid{} $\F$ being bad is when $\F$ contains a finite impossible \cluster.
This reason will be `eliminated' by a Sahlqvist formula in \S\ref{badc}, where it is also shown
that this is indeed a reason for $\F$ being bad (see Corollary~\ref{co:noimpossible}).

So suppose that
$\F=(W,\Rh,\Rv)$ is a countable rooted frame for $\comm$ that is represented as a \grid{} as $(X,Y,g)$,
and 
$\F$ contains no impossible \clusters.
%
%
We may assume that $X$ and $Y$ are disjoint, and consider the elements of $X\cup Y$ as distinct variables. We define a set $\con{\F}$ of `constraints' such that each constraint in $\con{\F}$ is
one of the forms $(z=n)$, $(z\geq k)$, or $(z\geq\lambda z')$, 
for some $z,z'\in X\cup Y$, $n\in\nNp\cup\{\infy\}$, $k\in\nNp$, and $\lambda=1,2$.
%
%
For all $x\in X$ and $y\in Y$,
\begin{equation}\label{conset}
\mbox{let $\con{\F}$ contain }\ \left\{
\begin{array}{ll}
\mbox{$(x=\infy)$ and $(y=\infy)$}, & \mbox{if $\F^{xy}$ is infinite},\\[5pt]
\mbox{the constraints from Table~\ref{tt:finclusters}}, & \mbox{if $\F^{xy}$ is finite}.
\end{array}
\right.
\end{equation}
We assume that $(z\geq 1)\in \con{\F}$ for every $z\in X\cup Y$.
A \emph{solution of\/} $\con{\F}$ 
is a function
\[
\xi:(X\cup Y)\to\bigl(\nNp\cup\{\infy\}\bigr)
\]
validating all constraints in $\con{\F}$. 
In other words, we are trying to solve a special kind of integer programming problem: $\con{\F}$ is a 
(possibly infinite) set of linear equations and inequalities (where all coefficients are positive integers or $\infy$), and we are looking for integer plus possibly (countably) infinite solutions of it.
By Lemmas~\ref{l:fits}~(ii) and \ref{l:clusterpm}, it is easy to see the following:

\begin{claim}\label{c:solpm}
If $\F$ is a  countable \grid{} that contains no impossible \clusters{} and
$\xi$ is a solution of $\con{\F}$, then
there is an onto p-morphism
\mbox{$h_\xi:(U,\drel{U}) \mprod(V,\drel{V})\to\F$} for some sets $U,V$ with
$|U|=\sum_{x\in X}\xi(x)$ and $|V|=\sum_{y\in Y}\xi(y)$.
\end{claim}

(In \S\ref{sformulas} we will show that the converse of Claim~\ref{c:solpm} also holds in the sense that
whenever a countable \grid{} $\F$ is a p-morphic image of a product of difference frames,
then $\F$ contains no impossible \clusters, and $\con{\F}$ has a solution; see Corollary~\ref{co:hassolution}.)



%



%



In order to characterise those countable $\F$ for which $\con{\F}$ has no solution, 
we first introduce some notions dealing with the one-variable constraints in $\con{\F}$. For every $z\in X\cup Y$, we let
\begin{align}
\label{maxz}
\maxz{z} & =  \left\{
\begin{array}{ll}
\infy, & \mbox{if $(z=n)\notin\con{\F}$ for any $n$,}\\[3pt]
\min\bigl\{ n : (z=n)\in\con{\F}\bigr\}, & \mbox{otherwise},
\end{array}
\right.\\[5pt]
\label{minz}
\minz{z} & =  \sup\bigl\{ k : \mbox{either }(z= k)\in\con{\F}\mbox{ or }(z\geq k)\in\con{\F}\bigr\}.
\end{align}
Next, in order to deal with the two-variable constraints, we define a (finite or countably infinite) 
edge-labelled digraph 
$\gf=(X\cup Y, E_\F)$ by taking, for any $z,z'\in X\cup Y$,
%
%
%
\begin{equation}\label{gedges}
(z\to^\lambda z')\in E_\F\qquad\mbox{iff}\qquad (z\geq\lambda z')\in\con{\F}.
\end{equation}
Observe that 
(i) all edges either go from some $x\in X$ to some $y\in Y$, or from some $y\in Y$ to some $x\in X$,
(ii) edge-labels $\lambda$ can only be $1$ or $2$, and 
(iii) if $(z\to^1 z')\in E_\F$ for some $z,z'$ then $(z'\to^1 z)\in E_\F$ as well.
%
For some $m\in\nN$,
we call a path
$z_0\to^{\lambda_1} z_{1}\to^{\lambda_{2}}\dots z_{m-1}\to^{\lambda_m} z_m$ 
in $\gf$ \emph{bad}, if
$\maxz{z_0}<  \lambda_1\cdot$ $\dots$ $\cdot\lambda_m \cdot \minz{z_m}$. 
(Observe that when $m=0$ then $z_0$ is a bad path of length $0$ whenever $\maxz{z_0}<  \minz{z_0}$.
Note that a bad path is not necessarily simple: it may contain the same edge more than once.)
Figs.~\ref{f:badgrids2} and \ref{f:badgrids1}
show two examples of \grids{} that are bad because their graphs contain some bad paths.

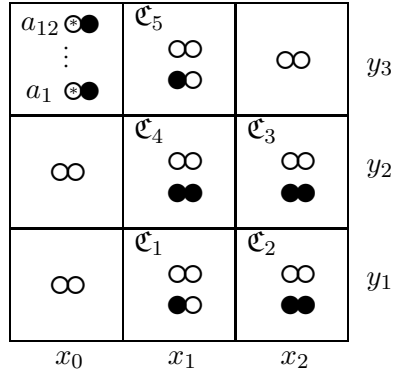
\begin{figure}[ht]
\begin{center}
\setlength{\unitlength}{.03cm}
\begin{picture}(250,165)(-15,-10)
\put(20,-10){$x_0$}
\put(70,-10){$x_1$}
\put(120,-10){$x_2$}
\put(158,25){$y_1$}
\put(158,75){$y_2$}

\multiput(0,0)(0,50){4}{\line(1,0){150}}
\multiput(0,0)(50,0){4}{\line(0,1){150}}

\multiput(25,122)(0,4){3}{\circle*{.5}}
\put(24,108){\xxi}
\put(7,108){$a_1$}
\put(24,138){\xxi}
\put(5,138){$a_{12}$}
\put(158,120){$y_3$}

\put(70,13){\ir}
\put(70,27){\rr}
\put(120,13){\ii}
\put(120,27){\rr}
\put(120,63){\ii}
\put(120,77){\rr}
\put(70,63){\ii}
\put(70,77){\rr}
\put(70,113){\ir}
\put(70,127){\rr}

\put(18,22){\rr}
\put(18,72){\rr}
\put(118,122){\rr}

\put(55,40){$\C_1$}  
\put(105,40){$\C_2$} 
\put(105,90){$\C_3$}  
\put(55,90){$\C_4$}  
\put(55,140){$\C_5$}
\end{picture}

\bigskip
\noindent
$P$:\quad
$
y_3\to^2 x_1\to^1 y_2\to^1 x_2\to^1 y_1\to^2 x_1\to^1 y_2\to^1 x_2\to^1 y_1\to^2 x_1
$

\medskip
with $\minz{x_1}=3$ and $\maxz{y_3}=12<24=2\cdot 1\cdot 1\cdot 1\cdot 2\cdot 1\cdot 1\cdot 1\cdot 2\cdot\minz{x_1}$.
\end{center}
\caption{A bad path that is not simple.}\label{f:badgrids2}
\end{figure}

In \S\ref{badpath} we will show that if a \grid{} $\F$ is such that it does not contain impossible \clusters, 
 but $\gf$ contains a bad path,
then there is a Sahlqvist formula `eliminating' this reason (and $\F$ is indeed bad).
Here we show that we have found all reasons for 
$\con{\F}$ not having a solution:

\begin{lemma}\label{l:bad}
Let $\F=(X,Y,g)$ be a countable \grid{} such that
\begin{enumerate}
\item $\F$ contains no impossible \clusters, and
\item there is no bad path in $\gf$.
\end{enumerate}
Then $\con{\F}$ has a solution.
\end{lemma}

\begin{proof}
Suppose $\F$ contains no impossible \clusters, and there is no bad path in $\gf$.
We will define a `minimal' solution $\smin$ such that it takes the same value on variables belonging to the same strongly connected component of $\gf$. 
To begin with,
for every strongly connected component $\scc$ in $\gf$, we let
(with a slight abuse of notation),
\begin{align}
\nonumber
\maxS{\scc} & =  \min\bigl\{ \maxz{z} : z\in\scc\bigr\},\\[5pt]
\label{minS}
\minS{\scc} & =  \left\{
\begin{array}{ll}
\infy, & \mbox{there is some $\to^2$ edge within $\scc$,}\\[3pt]
\sup\bigl\{ \minz{z}: z\in\scc\bigr\}, & \mbox{otherwise}.
\end{array}
\right.
\end{align}
Next, we define an
acyclic digraph $\gfp$ as follows
($\gfp$ is what is called the \emph{condensation} of $\gf=(X\cup Y,E_\F)$):
its nodes are the strongly connected components of $\gf$, 
%
%
and we define the edges by taking
\begin{align*}
\scc\Rightarrow \scc'\qquad & \mbox{iff}\qquad  \mbox{there exist $z$ in $\scc$, $z'$ in $\scc'$ with 
$(z\to^2 z')\in E_\F$}\\
& \mbox{iff}\qquad  \mbox{there exist $z$ in $\scc$, $z'$ in $\scc'$ with 
$(z\geq 2z')\in\con{\F}$.}
\end{align*}
%
%
For $n\in\nN$, we call a path
$\scc_0\Rightarrow \scc_{1}\Rightarrow\dots \scc_{n-1}\Rightarrow \scc_n$
in $\gfp$ \emph{bad}, if
$\maxS{\scc_0} < 2^n\cdot \minS{\scc_n}$. 

\begin{lclaim}\label{c:r4}
There is no bad path in $\gfp$.
\end{lclaim}

\begin{proof}
Suppose indirectly that $\scc_0\Rightarrow \scc_{1}\Rightarrow\dots \scc_{n-1}\Rightarrow \scc_n$ is a bad path in $\gfp$, that is, $\maxS{\scc_0} < 2^n\cdot \minS{\scc_n}$. 
 Then there exist $m\in\nN$, $z_0\in \scc_0$, $z_n\in \scc_n$ and a path $P$ of the form
\mbox{$z_0\to^{\lambda_1} \dots \to^{\lambda_m} z_n$} in $\gf$
such that 
$\maxS{\scc_0}=\maxz{z_0}$
and $2^n\leq\lambda_1\cdot$ $\dots$ $\cdot \lambda_m$ whenever $m>0$.
Now there are several cases:
\begin{itemize}
\item[(a)]
There is $z\in \scc_n$ such that $\minS{\scc_n}=\minz{z}$.
Then take $P$ and  continue it with any path from $z_n$ to $z$. The resulting path in $\gf$ is bad,
a contradiction.

\item[(b)]
$\minS{\scc_n}=\infy$ and there is a $\to^2$ edge within $\scc_n$.
Then take any path $Q$ from $z_n$ to $z_n$ containing this $\to^2$ edge.
Suppose $Q$ is of the form $z_n\to^{\mu_1}\dots\to^{\mu_i} z_n$. 
Then 
\mbox{$\mu_1\cdot$ $\dots$ $\cdot \mu_i\geq 2$,} and so there is $r\in\nN$ such that 
$\maxz{z_0}<2^n\cdot(\mu_1\cdot$ $\dots$ $\cdot \mu_i)^r\cdot \minz{z_n}$.
Then the path in $\gf$ obtained by starting with $P$ and then repeating $Q$ $r$ times is bad,
a contradiction.

\item[(c)]
$\minS{\scc_n}=\infy$, there is no $\to^2$ edge within $\scc_n$,
 but for every $i\in\nN$
 there is some $w_i\in\scc_n$ with $\minz{w_i}\geq i$.
 Then choose $i$ such that $i\cdot 2^n>\maxz{z_0}$.
 Then the path in $\gf$ obtained by starting with $P$ and then continuing with any path from $z_n$ to $w_i$  is bad, a contradiction again,
\end{itemize}
proving Claim~\ref{c:r4}.
\end{proof}

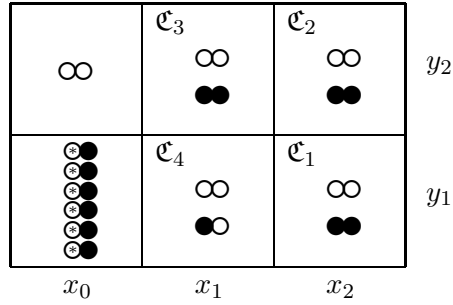
\begin{figure}[ht]
\begin{center}
\setlength{\unitlength}{.035cm}
\begin{picture}(160,100)
\put(20,-10){$x_0$}
\put(70,-10){$x_1$}
\put(120,-10){$x_2$}
\put(158,25){$y_1$}
\put(158,75){$y_2$}

\multiput(0,0)(0,50){3}{\line(1,0){150}}
\multiput(0,0)(50,0){4}{\line(0,1){100}}

\multiput(20,4)(0,7.5){6}{\xxi}
\put(70,13){\ir}
\put(70,27){\rr}
\put(120,13){\ii}
\put(120,27){\rr}
\put(120,63){\ii}
\put(120,77){\rr}
\put(70,63){\ii}
\put(70,77){\rr}

\put(18,72){\rr}

\put(55,40){$\C_4$}
\put(105,40){$\C_1$}
\put(105,90){$\C_2$}
\put(55,90){$\C_3$}
\end{picture}

\vspace*{.8cm}
$P$:\quad
$y_1\to^2 x_1\to^1 y_2\to^1 x_2\to^1 y_1$\qquad

\medskip
with $\minz{y_1}=6$ and $\maxz{y_1}=6<12=2\cdot 1\cdot 1\cdot 1\cdot \minz{y_1}$

\end{center}
\caption{A bad path within a strongly connected component.}\label{f:badgrids1}
\end{figure}

%
%

%
%

Next, 
%
for every node $\scc$ in $\gfp$, let
%
\[
\rank(\scc)=\text{sup}\{n : \mbox{there is a path in $\gfp$ of length $n$ starting at }\scc\}.
 \]
We define a function $\numin$ from the nodes of $\gfp$ to $\nNp\cup\{\infy\}$
by induction on their $\rank$ by taking, for every strongly connected component $\scc$,
%
\begin{equation}\label{nulabel}
\numin(\scc)=\left\{
\begin{array}{ll}
\text{sup}\bigl(\{2\numin(\scc') : \scc\Rightarrow \scc'\}\cup\{\minS{\scc}\}\bigr),\! & \mbox{if $\rank(\scc)\in\nN$},\\
\infy, & \mbox{if $\rank(\scc)=\infy$}
\end{array}
\right.
\end{equation}
(see Examples~\ref{e:numin} and \ref{e:sqbad} below).

%
\begin{lclaim}\label{c:lbsolution}
For all strongly connected components $\scc,\scc'$ in $\gf$, 
all $n\in\nNp\cup\{\infy\}$, $k\in\nNp$, and $\lambda\in\{1,2\}$, we have the following:
\begin{itemize}
\item[{\rm (i)}]
If $(z=n)\in\con{\F}$ for some $z\in\scc$, then $\numin(\scc)=n$.
\item[{\rm (ii)}]
If $(z\geq k)\in\con{\F}$ for some $z\in\scc$, then $\numin(\scc)\geq k$.
\item[{\rm (iii)}]
If $(z\geq\lambda z')\in\con{\F}$ for some $z\in\scc$, $z'\in\scc'$, then 
$\numin(\scc)\geq\lambda\cdot\numin(\scc')$.
\end{itemize}
\end{lclaim}
\begin{proof}
%
%

(i): 
If $n=\infy$ then $\minz{z}=\infy$, and so $\numin(\scc)=\infy$.
So suppose that $n\in\nNp$. Then
\begin{equation}\label{alln}
\maxS{\scc}\leq\maxz{z}\leq n\leq\minz{z}\leq\minS{\scc}.
\end{equation}
If any of the inequalities $\leq$ in \eqref{alln} were $<$, then $\scc$ would be a bad path of length $0$ in
$\gfp$, contradicting Claim~\ref{c:r4}. So we have $\minS{\scc}=\maxS{\scc}=n$.
We also have that $\rank(\scc)\in\nN$. (Otherwise,
 there would exist a bad path of length $>n$ in $\gfp$ starting
at $\scc$.)
If $\rank(\scc)=0$ then $\numin(\scc)=\minS{\scc}$ by definition, and so we have $\numin(\scc)=n$.
Now suppose that $\rank(\scc)>0$. By definition, $\numin(\scc)\geq\minS{\scc}$ always holds.
So suppose indirectly that $\numin(\scc)>\minS{\scc}$. 
We will construct a bad path in $\gfp$, contradicting Claim~\ref{c:r4}.
To begin with,
there is $\scc_1$ such that $\scc\Rightarrow \scc_1$
and $2\cdot\numin(\scc_1)>\minS{\scc}=\maxS{\scc}$.
(Either because $\scc_1$ is such that $\numin(\scc)=2\numin(\scc_1)$ or
because $\numin(\scc)=\infy$.)
As $\numin(\scc_1)\geq\minS{\scc_1}$, there are two cases: 
(a) $\numin(\scc_1)=\minS{\scc_1}$. 
Then $2\cdot\minS{\scc_1}>\minS{\scc}=\maxS{\scc}$, and so $\scc\Rightarrow \scc_1$ is a bad path in $\gfp$.
(For example, this is the case when $\scc_1$ is a final node in $\gfp$.)
(b) $\numin(\scc_1)>\minS{\scc_1}$.
Then there is $\scc_2$ such that $\scc_1\Rightarrow \scc_2$ and
$2^2\cdot\numin(\scc_2)> \maxS{\scc}$.
(Either because $\scc_2$ is such that $\numin(\scc_1)=2\numin(\scc_2)$ or
because $\numin(\scc_1)=\infy$.)
Again, there are two cases: 
(b.1) $\numin(\scc_2)=\minS{\scc_2}$.
Then 
$\scc\Rightarrow \scc_1\Rightarrow \scc_2$ is a bad path in $\gfp$.
(b.2) $\numin(\scc_2)>\minS{\scc_2}$.
Then again, there are two cases. And so on, sooner or later we reach
a final node in $\gfp$, where we only have case (a). 

(ii):
If $\numin(\scc)=\infy$ then the statement holds. If  $\numin(\scc)\in\nN$ then
$\rank(\scc)\in\nN$, and so
$\numin(\scc)\geq\minS{\scc}\geq\minz{z}\geq k$.

%

(iii):
If $\lambda=1$ then $\scc=\scc'$, and so $\numin(\scc)=\numin(\scc')$, as required.
If $\lambda=2$ and $\scc=\scc'$, then $\minS{\scc}=\minS{\scc'}=\infy$, and so 
$\numin(\scc)=\numin(\scc')=\infy$. Thus $\numin(\scc)\geq 2\cdot\numin(\scc')$ holds.
If $\lambda=2$ and $\scc\ne \scc'$, then $\scc\Rightarrow \scc'$. If $\numin(\scc)=\infy$, then $\numin(\scc)\geq 2\cdot\numin(\scc')$ holds.
If $\numin(\scc)\in\nN$ then $\rank(\scc)\in\nN$, and so again $\numin(\scc)\geq 2\cdot\numin(\scc')$ holds, as required.
\end{proof}

Now for every $\scc$ in $\gfp$ and every $z$ in $\scc$, we define
\begin{equation}\label{lbdef}
\smin(z)=\numin(\scc).
\end{equation}
By Claim~\ref{c:lbsolution},
$\smin$ is a solution of $\con{\F}$, proving Lemma~\ref{l:bad}.
\end{proof}

Now by Claim~\ref{c:solpm} and Lemma~\ref{l:bad} we obtain:
\begin{corollary}\label{co:bad}
If a countable \grid{} $\F=(X,Y,g)$ is bad (that is, $\F$ is not the p-morphic image of a product of difference frames), then at least one of the following two reasons holds:
\begin{enumerate}
\item either $\F$ contains a finite impossible \cluster,
\item or there is a bad path in $\gf$.
\end{enumerate}
\end{corollary}

\begin{example}\label{e:numin}
{\rm
Take the \grid{} $\F$ in Fig.~\ref{f:bothways}.
We compute $\numin$. To begin with, we have the following strongly connected components in $\gf$:
$\scc_1=\{y_1\}$,
$\scc_2=\{x_1\}$,
$\scc_3=\{y_3,x_2,y_2\}$,
$\scc_4=\{y_4\}$,
$\scc_5=\{x_4,y_5,x_5,x_6\}$,
$\scc_6=\{x_3\}$.
Then the edges in $\gfp$ are $\scc_2\Rightarrow\scc_3$ and $\scc_4\Rightarrow\scc_5$.
Therefore, we have:
\begin{align*}
& \numin(\scc_1)=\minS{\scc_1}=\minz{y_1}=14,\\
& \numin(\scc_6)=\minS{\scc_6}=\minz{x_3}=8,\\
& \numin(\scc_3)=\minS{\scc_3}=\max\bigl\{\minz{y_3},\minz{x_2},\minz{y_2} \bigr\}=
\max\{ 7,3,3\}=7,\\
& \numin(\scc_2)=\max\bigl\{2\cdot\numin(\scc_3),\minS{\scc_2}\bigr\}=
\max\{2\cdot 7,\minz{x_1}\}=\max\{2\cdot 7, 14\}=14,\\
& \numin(\scc_5)=\minS{\scc_5}=\max\bigl\{\minz{x_4},\minz{y_5},\minz{x_5},\minz{x_6} \bigr\}=\max\{ 3,3,3,3\}=3,\\
& \numin(\scc_4)=\max\bigl\{2\cdot\numin(\scc_5),\minS{\scc_4}\bigr\}=
\max\{2\cdot 3,\minz{y_4}\}=\max\{2\cdot 3, 8\}=8.
\end{align*}
}
\end{example}


\subsection{Sahlqvist formulas}\label{sformulas}

In \S\ref{badc} and \S\ref{badpath} below,
we will eliminate each of the two kinds of reasons in 
Corollary~\ref{co:bad} for a countable \grid{} $\F$ being bad,
using a Sahlqvist formula $\phi_{\F}$ in each case.


\subsubsection{Eliminating impossible \clusters}\label{badc}

Recall from Table~\ref{tt:finclusters} that a finite \cluster{} is \emph{impossible\/}, if it is one of the types (\bcno{1})--(\bcno{4}). 
We define formulas for the cases of (\bcno{1}), (\bcno{3}) and (\bcno{4});
the case of (\bcno{2}) is similar and left to reader. 
So let $\C=(C,\Rh,\Rv)$ be a \cluster\ consisting of $n=k+\ell$ points for some $k,\ell\in\nNp$, out of which $a_1,\dots,a_k$ are $\Rv$-irreflexive, $a_1$ is $\Rh$-reflexive, $b_1,\dots,b_\ell$ are $\Rh$-irreflexive; see Fig.~\ref{f:no}.
(It does not matter whether any of $a_2,\dots,a_k$ are $\Rh$-reflexive or -irreflexive, or
 whether any of $b_1,\dots,b_\ell$ are $\Rv$-reflexive or -irreflexive.)

\begin{figure}[ht]
\begin{center}
\setlength{\unitlength}{.03cm}
\begin{picture}(50,125)(0,5)
\multiput(0,0)(0,125){2}{\line(1,0){60}}
\multiput(0,0)(60,0){2}{\line(0,1){125}}
\put(30,10){\ri}
\put(15,10){$a_1$}
\put(30,25){\xxi}
\put(15,25){$a_2$}
\multiput(37,42)(0,4){3}{\circle*{.5}}
\put(30,60){\xxi}
\put(15,60){$a_k$}
\put(30,75){\ixx}
\put(15,75){$b_1$}
\multiput(37,92)(0,4){3}{\circle*{.5}}
\put(30,110){\ixx}
\put(15,110){$b_\ell$}
\end{picture}
\end{center}
\caption{An impossible \cluster{} of type (\bcno{1}), (\bcno{3}) or (\bcno{4}).}\label{f:no}
\end{figure}

We introduce fresh propositional variables $\avar{i}$ for $i=1,\dots k$, and  $\bvar{j}$ for $j=1,\dots,\ell$,  
and define
\begin{align}
\label{atype}
& \atp{i} :\quad \neg\avar{i}\land\Bv\avar{i}\land\bigwedge_{j=1}^{\ell}\bvar{j},\quad\mbox{for all $i=1,\dots,k$},\\
\label{btype}
& \btp{j} :\quad \neg\bvar{j}\land\Bh\bvar{j}\land\bigwedge_{i=1}^{k}\avar{i},\quad\mbox{for all $j=1,\dots,\ell$},\\
\nonumber
& \alphaC:\quad\atp{1}
\land \bigwedge_{i=1}^k\Dh\bigl(\atp{i}\land\bigwedge_{\substack{s=1\\ s\ne i}}^k\Dv\atp{s}\bigr)
\land\bigwedge_{j=1}^\ell\Dh\bigl(\btp{j}\land\bigwedge_{s=1}^k\Dv\atp{s}\bigr)\,\land\\
\nonumber
& \hspace*{2.7cm}
\bigwedge_{i=2}^k\Dv\bigl(\atp{i}\land\bigwedge_{t=1}^\ell\Dh\btp{t}\bigr)
\land \bigwedge_{j=1}^\ell\Dv\bigl(\btp{j}\land\bigwedge_{\substack{t=1\\ t\ne j}}^\ell\Dh\btp{t}\bigr),\\
\nonumber
& \nof:\quad\alphaC\to\Dh^+\Dv^+\bigl(\bigwedge_{i=1}^k\avar{i}\land\bigwedge_{j=1}^{\ell}\bvar{j}\bigr).
\end{align}
It is straightforward to check that $\nof$ is a Sahlqvist formula.

\begin{lemma}\label{l:impossibledec}
It is decidable whether a bimodal formula is of the form $\nof$ for some impossible \cluster{} $\C$.
\end{lemma}

\begin{proof}
Observe that $\nof$ only depends on the numbers $k,\ell$ and the type of $\C$.
\end{proof}


\begin{lemma}\label{l:badclincl}
$\nof$ is not valid in any \grid{} that contains the \cluster{} $\C$.
\end{lemma}

\begin{proof}
Suppose $\F=(W,\Rh,\Rv)$ is a \grid{} containing $\C$. We define a model $\M$ on $\F$ by taking
\begin{align*}
\M(\avar{i}) & = \{w\in W: a_i\Rv w\},\quad\mbox{for $i=1,\dots,k$},\\
\M(\bvar{j}) & = \{w\in W: b_j\Rh w\},\quad\mbox{for $j=1,\dots,\ell$}.
\end{align*}
It is straightforward to check that $\M,a_1\models\alphaC$. On the other hand, 
if $w\in W$ is such
that $\M,w\models\bigwedge_{i=1}^k\avar{i}\land\bigwedge_{j=1}^{\ell}\bvar{j}$, then 
$w$ must be in $\C$ by the definition of $\M$ and \grids. As all the $a_i$ are $\Rv$-irreflexive
and all the $b_j$ are $\Rh$-irreflexive, $w$ should be different from all of them, 
a contradiction.
\end{proof}

\begin{lemma}\label{l:badclinprod}
$\nof$ is valid in every product of difference frames.
\end{lemma}

\begin{proof}
Let $\M$ be a model over a product $(U,\drel{U})\mprod(V,\drel{V})$ of difference frames, and suppose that
$\M,(u_0,v_1)\models\alphaC$. By \eqref{atype}--\eqref{btype},
there are distinct points $u_0,\dots,u_n$ in $U$ and distinct points $v_1,\dots,v_n$ in $V$ such that
\begin{align}
\label{sees1}
& \M,(u_i,v_1)\models\atp{i}\land\bigwedge_{\substack{s=1\\ s\ne i}}^k\Dv\atp{s},\quad\mbox{for $i=1,\dots,k$},\\
\label{sees2}
& \M,(u_{k+j},v_1)\models\btp{j}\land\bigwedge_{s=1}^k\Dv\atp{s},\quad\mbox{for $j=1,\dots,\ell$},\\
\label{sees3}
& \M,(u_0,v_i)\models\atp{i}\land\bigwedge_{t=1}^\ell\Dh\btp{j},\quad\mbox{for $i=1,\dots,k$},\\
\label{sees4}
& \M,(u_0,v_{k+j})\models\btp{j}\land\bigwedge_{\substack{t=1\\ t\ne j}}^\ell\Dh\btp{t},\quad\mbox{for $j=1,\dots,\ell$}
\end{align}
(see Fig.~\ref{f:forall1}).
\begin{figure}[ht]
\begin{center}
\setlength{\unitlength}{.07cm}
\begin{picture}(135,128)(-12,-9)
\thicklines
\put(20,4){$\atp{1}$}
\put(20,9){$u_0$}

\put(35,3){$\atp{1}$}
\put(36,-5){$\uparrow$}
\put(32,-9){$\Dv\atp{i}$}
\put(35,9){$u_1$}

\put(50,3){$\atp{2}$}
\put(51,-5){$\uparrow$}
\put(47,-9){$\Dv\atp{i}$}
\put(50,9){$u_2$}

\put(78,3){$\atp{k}$}
\put(79,-5){$\uparrow$}
\put(75,-9){$\Dv\atp{i}$}
\put(78,9){$u_k$}

\put(93,2){$\btp{1}$}
\put(94,-5){$\uparrow$}
\put(90,-9){$\Dv\atp{i}$}
\put(93,9){$u_{k+1}$}

\put(123,2){$\btp{\ell}$}
\put(124,-5){$\uparrow$}
\put(120,-9){$\Dv\atp{i}$}
\put(123,9){$u_{k+\ell}$}

\put(22,13){\line(0,1){22}}
\multiput(22,42)(0,3){3}{\circle*{.5}}
\put(22,55){\line(0,1){25}}
\multiput(22,87)(0,3){3}{\circle*{.5}}
\multiput(37,87)(0,3){3}{\circle*{.5}}
\multiput(52,87)(0,3){3}{\circle*{.5}}
\multiput(80,87)(0,3){3}{\circle*{.5}}
\multiput(95,87)(0,3){3}{\circle*{.5}}
\multiput(125,87)(0,3){3}{\circle*{.5}}
\put(22,100){\vector(0,1){10}}
\put(20,113){$V$}

\put(20,15){\line(1,0){35}}
\multiput(62,15)(3,0){3}{\circle*{.5}}
\multiput(62,30)(3,0){3}{\circle*{.5}}
\multiput(62,60)(3,0){3}{\circle*{.5}}
\multiput(62,75)(3,0){3}{\circle*{.5}}
\multiput(62,105)(3,0){3}{\circle*{.5}}
\put(75,15){\line(1,0){25}}
\multiput(107,15)(3,0){3}{\circle*{.5}}
\multiput(107,30)(3,0){3}{\circle*{.5}}
\multiput(107,60)(3,0){3}{\circle*{.5}}
\multiput(107,75)(3,0){3}{\circle*{.5}}
\multiput(107,105)(3,0){3}{\circle*{.5}}
\put(120,15){\vector(1,0){10}}
\put(133,13){$U$}
\put(14,14.5){$v_1$}

\put(8,29){$\atp{2}$}
\put(0,29){$\to$}
\put(-13,29){$\Dh\btp{j}$}
\put(14,29.5){$v_2$}

\put(8,59){$\atp{k}$}
\put(0,59){$\to$}
\put(-13,59){$\Dh\btp{j}$}
\put(14,59.5){$v_k$}

\put(4,74){$\btp{1}$}
\put(-3,74){$\to$}
\put(-16,74){$\Dh\btp{j}$}
\put(10,74.5){$v_{k+1}$}

\put(4,104){$\btp{\ell}$}
\put(-3,104){$\to$}
\put(-16,104){$\Dh\btp{j}$}
\put(10,104.5){$v_{k+\ell}$}

\thinlines
\multiput(37,13)(15,0){2}{\line(0,1){22}}
\multiput(37,42)(0,3){3}{\circle*{.5}}
\multiput(52,42)(0,3){3}{\circle*{.5}}
\multiput(80,13)(15,0){2}{\line(0,1){22}}
\multiput(80,42)(0,3){3}{\circle*{.5}}
\multiput(95,42)(0,3){3}{\circle*{.5}}
\put(125,13){\line(0,1){22}}
\multiput(37,55)(15,0){2}{\line(0,1){25}}
\multiput(80,55)(15,0){2}{\line(0,1){25}}
\put(125,55){\line(0,1){22}}
\multiput(125,42)(0,3){3}{\circle*{.5}}
\multiput(37,100)(15,0){2}{\line(0,1){5}}
\multiput(80,100)(15,0){2}{\line(0,1){5}}
\put(125,100){\line(0,1){5}}

\put(20,30){\line(1,0){35}}
\multiput(20,60)(0,15){2}{\line(1,0){35}}
\put(20,105){\line(1,0){35}}

\put(75,30){\line(1,0){25}}
\multiput(75,60)(0,15){2}{\line(1,0){25}}
\put(75,105){\line(1,0){25}}

\put(120,30){\line(1,0){5}}
\multiput(120,60)(0,15){2}{\line(1,0){5}}
\put(120,105){\line(1,0){5}}
\end{picture}
\end{center}
\caption{Satisfying $\alphaC$ in a product frame $(U,\drel{U})\mprod(V,\drel{V})$.}\label{f:forall1}
\end{figure}
We say that a pair $(u,v)\in U\mprod V$ is of $a$-\emph{type\/} (or of $b$-\emph{type\/}) if
$\M,(u,v)\models\atp{i}$ for some $i=1,\dots,k$ 
(or $\M,(u,v)\models\btp{j}$ for some $j=1,\dots,\ell$). 
Take the subset $Z$  of $U\mprod V$ consisting of the pairs $(u_i,v_j)$ for  $i=0,\dots,n$ and  $j=1,\dots,n$.
We claim that
\begin{equation}\label{phole1}
\mbox{there exists a pair in $Z$ that is neither $a$-type nor $b$-type.}
\end{equation}
Indeed, suppose the contrary, that is, every pair in $Z$ is either $a$-type or $b$-type.
For every $0\leq i\leq n$, there can be $\leq k$ many $a$-type pairs among 
$(u_i,v_1),\dots,(u_i,v_n)$. So there have to be $\geq \ell$ many $b$-type pairs among them.
So altogether in $Z$ there are $\geq (n+1)\cdot\ell$ many $b$-type pairs.
Thus, by the generalised pigeonhole principle, there exists $1\leq s\leq n$ such that there are $>\ell$ many $b$-type points among $(u_0,v_s),\dots,(u_n,v_s)$.
But for every $1\leq j\leq n$, there can be $\leq \ell$ many $b$-type pairs among 
$(u_0,v_j),\dots,(u_n,v_j)$, a contradiction, proving \eqref{phole1}.

So suppose $(u,v)\in Z$ is neither $a$-type nor $b$-type. 
By \eqref{sees1}--\eqref{sees2},  for every $1\leq i\leq k$ there is some $z_i\in V$ such that
$z_i\ne v$ and $\M,(u,z_i)\models\atp{i}$, and so $\M,(u,v)\models\avar{i}$ by \eqref{atype}.
Similarly, by \eqref{sees3}--\eqref{sees4},  for every $1\leq j\leq \ell$ there is some $w_j\in U$ such that
$w_j\ne u$ and $\M,(w_j,v)\models\btp{j}$, and so $\M,(u,v)\models\bvar{j}$ by \eqref{btype}.
Therefore,
$
\M,(u,v)\models\bigwedge_{i=1}^k\avar{i}\land\bigwedge_{j=1}^{\ell}\bvar{j},
$
as required.
\end{proof}

As a consequence of Lemmas~\ref{l:badclincl} and \ref{l:badclinprod} we also obtain:

\begin{corollary}\label{co:noimpossible}
For every countable \grid{} $\F$, if $\F$ is a p-morphic image of a product of difference frames,
then $\F$ contains no impossible \clusters.
\end{corollary}


\subsubsection{Eliminating bad paths}\label{badpath}

Let $\F=(W,\Rh,\Rv)$ be a countable rooted frame for $\comm$ that is represented 
as a \grid{} as $(X,Y,g)$.
Suppose that $\F$ contains no impossible \clusters, but $\gf$ contains a bad path $P$ of the form
\[
P:\quad
z_0\to^{\lambda_1} z_{1}\to^{\lambda_{2}}\dots z_{m-1}\to^{\lambda_m} z_m
\]
such that $m\in\nN$ and 
$\maxz{z_0}<  \lambda_1\cdot$ $\dots$ $\cdot\lambda_m \cdot \minz{z_m}$. 
Throughout this subsection, we assume that $z_0,z_m\in Y$, and define
a Sahlqvist formula $\badpathf{P}$ for this case. The other three
cases are similar and left to the reader.

The antecedent of $\badpathf{P}$ will consists of two conjuncts:
$\firstf{P}$ (expressing the value of $\maxz{z_0}$), and
$\Dh^+\Dv^+\lastf{P}$ (expressing that the value of $\lambda_1\cdot$ $\dots$ $\cdot\lambda_m \cdot \minz{z_m}$ is sufficiently large). 

We begin with defining $\firstf{P}$.
As $\maxz{z_0}<  \lambda_1\cdot$ $\dots$ $\cdot\lambda_m \cdot \minz{z_m}$, we must have that
$\maxz{z_0}=\noc{P}$ for some $\noc{P}\in\nNp$, and so $(z_0=\noc{P})\in\con{\F}$ by
\eqref{maxz}.
As $z_0\in Y$, 
%
\begin{equation}
\label{xp}
\mbox{there is some $x_P\in X$ such that $\F^{x_Pz_0}$ consists of }
\mbox{$\noc{P}$ many $\ \xxi$ points $a_1,\dots,a_{\noc{P}}$.}
\end{equation}
%
So, we introduce fresh propositional variables $\avarr$ and $\avar{i}$, for $i=1,\dots,\noc{P}$,  and
define the formula
%
\[
\firstf{P} :\quad \Bv^+\avarr\land \bigwedge_{i=1}^{\noc{P}}\Dv^+(\neg\avar{i}\land\Bh\avar{i}).
\]

%
%

In order to define $\lastf{P}$, 
we first describe the path $P$ with a formula $\pathf{P}$. To this end,
we say that a \cluster{} $\C$ \emph{corresponds to an edge\/} $z\to^\lambda z'$ in $\gf$,
if $\C$ is (isomorphic to) $\F^{zz'}$ whenever $z\in X$, $z'\in Y$, and
$\C$ is (isomorphic to) $\F^{z'z}$ whenever $z'\in X$, $z\in Y$.
Observe that only switch \clusters{} can correspond to some edge in $\gf$. In particular, for every $x\in X$ and every $y\in Y$, we have the following:
\begin{itemize}
\item $x\to^1 y$ is an edge in $\gf$ iff $y\to^1 x$ is an edge in $\gf$ iff $\F^{xy}$ is a type (\bceqsw) \cluster.
\item  $x\to^2 y$ is an edge in $\gf$ iff $\F^{xy}$ is a type (\bchvsw) \cluster.
\item  $y\to^2 x$ is an edge in $\gf$ iff $\F^{xy}$ is a type (\bcvhsw) \cluster.
\end{itemize}
%
%
%
%
If $m>0$ then
let $\C_1,\dots\C_m$ be the sequence of \clusters{} corresponding to the edges in $P$  (that is, $\C_j$ corresponds to $z_{j-1}\to^{\lambda_j} z_j$).
Observe that for each $j=1,\dots,m$, 
\begin{itemize}
\item
if $\C_j$ is of type (\bcvhsw) then there is some $\ \ir$ point $c_j$ in $\C_j$; 
\item
if $\C_j$ is of type (\bchvsw) then there is some $\ \ri$ point $c_j$ in $\C_j$; and
\item
if $\C_j$ is of type (\bceqsw) then there is some $\ \ii$ point $c_j$ in $\C_j$.
\end{itemize}
So, for $j=1,\dots,m$, we introduce fresh propositional variables $\Cvar{j}$, and define
formulas
\begin{equation}\label{clusterfdef}
\ctp{j} :\quad\left\{
\begin{array}{ll}
\neg\Cvar{j}\land\Bh\Cvar{j}, &\mbox{if $\C_j$ is of type (\bcvhsw)},\\[3pt]
\neg\Cvar{j}\land\Bv\Cvar{j}, &\mbox{if $\C_j$ is of type (\bchvsw)},\\[3pt]
\neg\Cvar{j}\land\Bh\Cvar{j}\land\Bv\Cvar{j},\quad &\mbox{if $\C_j$ is of type (\bceqsw)}.
\end{array}
\right.
\end{equation}
We also introduce a fresh propositional variable $\bvarr$, and define the formulas $\betaf_0$, $\betaf_1$, $\dots$, $\betaf_m=\pathf{P}$ inductively as follows.
Let $\betaf_0=\neg\avarr\land\Bh\bvarr$ (where $\avarr$ is the same variable as in $\firstf{P}$),
and for $j=1,\dots,m$, let
%
\begin{equation}\label{pathfdef}
\betaf_j:\ \ \left\{
\begin{array}{ll}
\Dh(\ctp{j}\land\betaf_{j-1}), & \mbox{if $z_{j-1}\in X$ and $\C_{j}$ is (\bceqsw)},\\[5pt]
\Dh\bigl(\ctp{j}\land\betaf_{j-1}\land \Dh(\ctp{j}\land\betaf_{j-1})\bigr), & \mbox{if $z_{j-1}\in X$ and $\C_{j}$ is (\bchvsw)},\\[5pt]
\Dv(\ctp{j}\land\betaf_{j-1}), & \mbox{if $z_{j-1}\in Y$ and $\C_{j}$ is (\bceqsw)},\\[5pt]
\Dv\bigl(\ctp{j}\land\betaf_{j-1}\land\Dv(\ctp{j}\land\betaf_{j-1})\bigr), & \mbox{if $z_{j-1}\in Y$ and $\C_{j}$ is (\bcvhsw)}.
\end{array}
\right.
\end{equation}

Now we are in a position to define
$\lastf{P}$, expressing that the value of \mbox{$\lambda_1\cdot$ $\dots$ $\cdot\lambda_m \cdot \minz{z_m}$} for the endpoint $z_m$ of $P$
is sufficiently large.
Let $\nok{P}\in\nNp$ be such that 
$\maxz{z_0}<  \mbox{$\lambda_1\cdot$ $\dots$ $\cdot\lambda_m \cdot \nok{P}$}$ and $\nok{P}\leq\minz{z_m}$.
We have two cases, depending on why $\minz{z_m}$ is `too large':
 \begin{enumerate}
 \item 
 
 There is $x_P'\in X$ such that 
 $\sizev(\F^{x_P'z_m})\geq \nok{P}$;
 \item
or there is $x_P'\in X$ such that $\F^{x_P'z_m}$ is an infinity \cluster
 \end{enumerate}
(see \eqref{minz}, \eqref{conset}, and Table~\ref{tt:finclusters}).
 We define a Sahlqvist formula $\lastf{P}$ for each of these two cases. 
 

\smallskip
\noindent
{\bf Case 1.}
%
Then there are $\Rv$-reflexive points $b_1^\circ,\dots,b_{\nor{P}}^\circ$ and
$\Rv$-irreflexive points $b_1^\bullet,\dots,b_{\noi{P}}^\bullet$ in $\F^{x_P'z_m}$ such that
$2\nor{P}+\noi{P}\geq \nok{P}$.
%
We introduce fresh propositional variables $\bvar{j}^\circ$ for $j=1,\dots,\nor{P}$, and
$\bvar{s}^\bullet$ for $s=1,\dots,\noi{P}$,
and define the formulas 
%
%
\begin{align}
\label{bvarnos2}
& \btp{j}^\circ :\quad \bvar{j}^\circ\land\bigwedge_{\substack{t=1\\ t\ne j}}^{\nor{P}}\neg\bvar{t}^\circ
\land\bigwedge_{t=1}^{\noi{P}}\neg\bvar{t}^\bullet,\quad\mbox{for all $j=1,\dots,\nor{P}$},\\[5pt]
\label{bvarnos22}
& \btp{s}^\bullet :\quad \bvar{s}^\bullet\land\bigwedge_{\substack{t=1\\ t\ne s}}^{\noi{P}}\neg\bvar{t}^\bullet
\land\bigwedge_{t=1}^{\nor{P}}\neg\bvar{t}^\circ,\quad\mbox{for all $s=1,\dots,\noi{P}$}.
\end{align}
Then we let
\[
\lastf{P}:\quad
\bigwedge_{j=1}^{\nor{P}}\Dv^+\bigl(\btp{j}^\circ\land\pathf{P}\land\Dv(\btp{j}^\circ\land\pathf{P})\bigr)\land
\bigwedge_{s=1}^{\noi{P}}\Dv^+\bigl(\btp{s}^\bullet\land\pathf{P}\bigr).
\]

%
%
%


\smallskip
\noindent
{\bf Case 2.}
Now we cannot use that we have enough different points in $\F^{x_P'z_m}$ like in Case 1, but instead we need to `generate' them.
There are two cases: 
Either 
(a)  $\F^{x_P'z_m}$ contains some $\ \ir$ point $c$ and some $\ \xxi$ point $d$
(this is when $\F^{x_P'z_m}$ is of type (\bcinf{1}),  (\bcinf{2}) or (\bcinf{4}));
or
(b) $\F^{x_P'z_m}$ contains some $\ \ixx$ point $c$ and some $\ \ri$ point $d$
(this is when $\F^{x_P'z_m}$ is of type (\bcinf{1}), (\bcinf{3}) or (\bcinf{4})).
In both cases, instead of the $\bvar{j}^\circ$ and $\bvar{j}^\bullet$ variables, we introduce fresh 
propositional variables $\cvarr$ and $\dvarr$, and define the formulas
\[
\cctp :  \quad \neg\cvarr\land\Bh\cvarr,\qquad\mbox{and}\qquad
\ddtp :  \quad \neg\dvarr\land\Bv\dvarr.
\]
Then we define the formulas
 $\deltaf_1,\dots,\deltaf_{\nok{P}}=\lastf{P}$ inductively as follows.
Let 
\[
\deltaf_1:\quad   \left\{
\begin{array}{ll}
\cctp\land\Dh(\ddtp\land\pathf{P}), & \mbox{in case (a)},\\[5pt]
\ddtp\land\Dv(\cctp\land\pathf{P}), & \mbox{in case (b)},
\end{array}
\right.
\]
and for $j=2,\dots,\nok{P}$, let
\[
\deltaf_j:\quad   \left\{
\begin{array}{ll}
\cctp\land\Dh\bigl(\ddtp\land\pathf{P}\land\Dv(\deltaf_{j-1}\land\Dv\deltaf_{j-1})\bigr), & \mbox{in case (a)},\\[5pt]
\ddtp\land\Dv\bigl(\cctp\land\pathf{P}\land\Dh(\deltaf_{j-1}\land\Dh\deltaf_{j-1})\bigr),  & \mbox{in case (b)}.
\end{array}
\right.
\]
%
%
%
%

Finally,
we define $\badpathf{P}$ by taking
\[
\badpathf{P}:\quad
\bigl(\firstf{P}\land\Dh^+\Dv^+\lastf{P}\bigr)\to
\Dv^+\Bigl(\bvarr\land\bigwedge_{i=1}^{\noc{P}}\avar{i}\Bigr).
\]
%
It is straightforward to check that $\badpathf{P}$ is a Sahlqvist formula.


\begin{lemma}\label{l:badpathdec}
It is decidable whether a bimodal formula is of the form $\badpathf{P}$ for some \grid{} $\F$ and bad path $P$ in $\gf$.
\end{lemma}

\begin{proof}
Observe that $\badpathf{P}$  only depends on the numbers $\noc{P}$, $\nor{P}$, $\noi{P}$,  $\nok{P}$, and the types in the sequence of \clusters{} corresponding to the edges in $P$.
\end{proof}


\begin{lemma}\label{l:badclinclgen}
Suppose $\F=(X,Y,g)$ is a \grid{} that contains no impossible \clusters.
If $P$ is a bad path in $\gf$, then $\badpathf{P}$ is not valid in~$\F$.
\end{lemma}

\begin{proof}
We use the notation introduced in the definition of $\badpathf{P}$.
We define a model $\M$ on $\F$ by taking
\begin{align*}
\M(\avarr) & =  \{w\in W: a_1\Rv^+ w\},\\
\M(\avar{i}) & =  \{w\in W: a_i\Rv w\},\quad\mbox{for $i=1,\dots,\noc{P}$},\\
\M(\bvarr) & =  \left\{
\begin{array}{ll}
\{w\in W: c_1\Rh w\}, &\mbox{if $m>0$},\\[3pt]
\displaystyle
\bigcup_{j=1}^{\nor{P}}\{w\in W: b_j^\circ\Rh w\}\cup\bigcup_{s=1}^{\noi{P}}\{w\in W: b_j^\bullet\Rh w\},\!\!
&\mbox{in Case 1, if $m=0$},\\[3pt]
\{w\in W: c\Rh w\},   &\mbox{in Case 2, if $m=0$},
\end{array}
\right.\\[5pt]
\M(\Cvar{j}) & =  \left\{
\begin{array}{ll}
\{w\in W: c_j\Rh w\}, &\mbox{if $\C_j$ is (\bcvhsw)},\\[3pt]
\{w\in W: c_j\Rv w\},  &\mbox{if $\C_j$ is (\bchvsw)},\\[3pt]
\{w\in W: c\Rh w\},  &\mbox{if $\C_j$ is (\bceqsw),\quad for $j=1,\dots,m$},
\end{array}
\right.
\end{align*}
then in Case 1, take
\begin{align*}
& \M(\bvar{j}^\circ) = \{b_j^\circ\},\quad\mbox{for $j=1,\dots,\nor{P}$},\\
& \M(\bvar{s}^\bullet) = \{b_s^\bullet\},\quad\mbox{for $s=1,\dots,\noi{P}$},
\end{align*}
and in Case 2, take
\[
\M(\cvarr) =  \{w\in W: c\Rh w\},\quad\mbox{and}\quad
\M(\dvarr)  =  \{w\in W: d\Rv w\}.
\]
It is straightforward to check that $\M,a_1\models\firstf{P}$. 
Further, it is easy to see that
%
in Case~1, $\M,b_j^\circ\models\pathf{P}$ for all $j=1,\dots,\nor{P}$ and
$\M,b_s^\bullet\models\pathf{P}$ for all $s=1,\dots,\noi{P}$,
and in Case~2, $\M,c\models\pathf{P}$.
%
%
Using these, it is not hard to check that $\lastf{P}$ is satisfied in $\M$ in both cases.

On the other hand, suppose $w\in W$ is such that $a_1\Rv^+ w$ and
$\M,w\models\bvarr$. We claim that $a_1\Rh^+ w$ 
%
%
follows. Indeed, if $m=0$ then this is because we have some $w'$ in $\F^{x_P'z_0}$ with $w'\Rh^+a_1$ and $w'\Rh w$, 
and if $m>0$ then because of $a_1\Rh c_1$ and $c_1\Rh w$.
Thus, $w$ must be in $\F^{x_Pz_0}$ by the definition of  \grids. 
If $\M,w\models\bigwedge_{i=1}^{\noc{P}}\avar{i}$ held as well, then $w$ should be different from all the
$a_i$, contradicting $\F^{x_Pz_0}=\{a_1,\dots,a_{\noc{P}}\}$.
\end{proof}

\begin{lemma}\label{l:badclinprodgen}
$\badpathf{P}$ is valid in every product of difference frames.
\end{lemma}

\begin{proof}
Again, we use the notation introduced in the definition of $\badpathf{P}$.
Let $\M$ be a model over a product $(U,\drel{U})\mprod(V,\drel{V})$ of difference frames, and suppose that
$\M,(u,v)\models\firstf{P}$ for some $u,v$.
Then there are distinct points $v_1,\dots,v_{\noc{P}}$ in $V$ such that 
\begin{equation}\label{strictas}
\M,(u,v_i)\models\Bv^+\avarr\land\neg\avar{i}\land\Bv\avar{i}\quad\mbox{for all $i=1,\dots,\noc{P}$.}
\end{equation}
\begin{lclaim}\label{c:kpoints}
If $\lastf{P}$ is satisfied in $\M$, then
there exist points  $u_1,\dots,u_{\nok{P}}\in U$ and distinct points  $w_1,\dots,w_{\nok{P}}\in V$ such that
$\M,(u_j,w_j)\models\pathf{P}$, for all $j=1,\dots,\nok{P}$.
\end{lclaim}

%
%

\begin{proof}
In Case~1 this easily follows from $2\nor{P}+\noi{P}\geq \nok{P}$ and  \eqref{bvarnos2}--\eqref{bvarnos22}. 
In Case~2(a):
We show by induction that, for all $j=1,\dots,\nok{P}$, if $\M,(a,b)\models\deltaf_j$ for some $(a,b)$, 
then there are distinct points $u_1,\dots,u_{j}$ in $U$ and
distinct points $w_1,\dots,w_j$ in $V$ such that 
\begin{itemize}
\item
$w_j=b$, 
\item
$u_s\ne a$ for any $s$ with $1\leq s\leq j$,
\item
$\M,(u_s,w_s)\models\ddtp\land\pathf{P}$ for all $1\leq s\leq j$, and
\item
$\M,(u_s,w_{s-1})\models\cctp$ for all $2\leq s\leq j$.
\end{itemize}
As $\lastf{P}=\deltaf_{\nok{P}}$, Claim~\ref{c:kpoints} will follow. To begin with,
the $j=1$ case is obvious. So suppose inductively that the statement holds for some $j-1$,
and suppose that $\M,(a,b)\models\deltaf_j$. Then there are $a'\in U$, $b_1,b_2\in V$ such that
$a'\ne a$, $b,b_1,b_2$ are all distinct, $\M,(a',b)\models\ddtp\land\pathf{P}$, and
$\M,(a',b_i)\models\deltaf_{j-1}$ for $i=1,2$.
By the IH, for each $i=1,2$, 
there are distinct points $u_1^i,\dots,u_{j-1}^i$ in $U$ and
distinct points $w_1^i,\dots,w_{j-1}^i$ in $V$ such that $w_{j-1}^i=b_i$, $u_s^i\ne a'$ for any $s$, 
$\M,(u_s^i,w_s^i)\models\ddtp\land\pathf{P}$ for all $1\leq s\leq j-1$, and
$\M,(u_s^i,w_{s-1}^i)\models\cctp$ for all $2\leq s\leq j-1$. Thus, for each $i=1,2$, 
\begin{align}
\label{cdchain1}
& \mbox{$\M,(u_s^i,w_s^i)\models\neg\dvarr\land\Bv\dvarr$ for all $1\leq s\leq j-1$, and}\\
\label{cdchain2}
& \mbox{$\M,(u_s^i,w_{s-1}^i)\models\neg\cvarr\land\Bh\cvarr$ for all $2\leq s\leq j-1$.}
\end{align}
As $w_{j-1}^1=b_1\ne b_2=w_{j-1}^2$, \eqref{cdchain1} and \eqref{cdchain2} imply that
all the $u_1^1,\dots,u_{j-1}^1$, $u_1^2,\dots,u_{j-1}^2$ are distinct.
Thus, either $a\notin\{u_1^1,\dots,u_{j-1}^1\}$ or $a\notin\{u_1^2,\dots,u_{j-1}^2\}$.
Let $j$ be such that $a\notin\{u_1^j,\dots,u_{j-1}^j\}$. Then the points $u_s=u_s^j$, $w_s=w_s^j$ for
$s=1,\dots,j-1$, $u_j=a'$, $w_j=b$ are as required. 

Case 2(b) is similar.
\end{proof}

\begin{lclaim}\label{c:pathconsok}
$\displaystyle\M,(u,v)\models\Dv^+\Bigl(\bvarr\land\bigwedge_{i=1}^{\noc{P}}\avar{i}\Bigr)$.
\end{lclaim}

\begin{proof}
We define sets $\wset{z_j}$ and $\hset{z_j}$ inductively, for $j=m,\dots,0$, such that the following hold,
for every $j\leq m$:
\begin{itemize}
\item[(i)]
$\wset{z_j}\subseteq U\mprod V$;
\item[(ii)]
$\M,(a,b)\models\betaf_j$ for every $(a,b)\in\wset{z_j}$;
\item[(iii)]
if $j<m$ then $\M,(a,b)\models\ctp{j+1}$ for every $(a,b)\in\wset{z_j}$;
\item[(iv)]
$\hset{z_j}=\left\{
\begin{array}{ll}
\{ a\in U : (a,b)\in\wset{z_j}\mbox{ for some $b$}\}, & \mbox{ if $z_j\in X$},\\[3pt]
\{ b\in V : (a,b)\in\wset{z_j}\mbox{ for some $a$}\}, & \mbox{ if $z_j\in Y$};
\end{array}
\right.
$

\item[(v)]
$|\hset{z_m}|=\nok{P}$, and if $j<m$ then $|\hset{z_{j}}|=\lambda_{j+1}\cdot |\hset{z_{j+1}}|$.
\end{itemize}
To begin with, we take the points from Claim~\ref{c:kpoints} and let
\begin{align*}
& \wset{z_m}=\bigl\{(u_1,w_1),\dots,(u_{\nok{P}},w_{\nok{P}})\bigr\},\\
& \hset{z_m}=\{w_1,\dots,w_{\nok{P}}\}.
\end{align*}
 Now suppose inductively that (i)--(v) hold for some $j$.
  There are several cases. Suppose first that $z_{j-1}\in X$, and
  $\hset{z_j}=\{b_1,\dots,b_{s_j}\}$ for some $s_j$.
  Take some $a_1,\dots,a_{s_j}\in U$ such that $(a_i,b_i)\in\wset{z_j}$ for all $i=1,\dots,s_j$.
  \begin{itemize}
  \item 
  If $\lambda_j=1$ (that is, $\C_j$ is of type (\bceqsw)),
 then $\betaf_j=\Dh(\ctp{j}\land\betaf_{j-1})$ by \eqref{pathfdef}.
By (ii), there are $a'_1,\dots,a'_{s_j}\in U$ such that $\M,(a'_i,b_i)\models\ctp{j}\land\betaf_{j-1}$,
for all  $i=1,\dots,s_j$. As $\ctp{j}=\neg\Cvar{j}\land\Bh\Cvar{j}\land\Bv\Cvar{j}$ by \eqref{clusterfdef},
all the $a'_i$ are distinct. We let $\hset{z_{j-1}}=\{a'_1,\dots,a'_{s_j}\}$.

  \item 
  If $\lambda_j=2$ (that is, $\C_j$ is of type (\bchvsw)),
 then $\betaf_j=\Dh\bigl(\ctp{j}\land\betaf_{j-1}\land\Dh(\ctp{j}\land\betaf_{j-1})\bigr)$ by \eqref{pathfdef}.
By (ii), there are $a'_1,\dots,a'_{2s_j}\in U$ such that $\M,(a'_i,b_i)\models\ctp{j}\land\betaf_{j-1}$
and $\M,(a'_{s_j+i},b_i)\models\ctp{j}\land\betaf_{j-1}$,
for all  $i=1,\dots,s_j$. As $\ctp{j}=\neg\Cvar{j}\land\Bv\Cvar{j}$ by \eqref{clusterfdef},
all the $2s_j$ many $a'_i$ are distinct. We let $\hset{z_{j-1}}=\{a'_1,\dots,a'_{2s_j}\}$.
\end{itemize}
The cases when $z_{j-1}\in Y$ are similar.

So we have (i)--(v) for all $j=0,\dots,m$, and so $|\hset{z_0}|=\nok{P}\cdot\lambda_m\cdot$ $\dots$ $\cdot\lambda_1$.
As $P$ is a bad path, $|\hset{z_0}|>\noc{P}$, and so by the pigeonhole principle, there is $w\in\hset{z_0}$ such that $w\ne v_i$ for any $1\leq i\leq \noc{P}$.
Therefore, $\M,(u,w)\models\bigwedge_{i=1}^{\noc{P}}\avar{i}$ by \eqref{strictas}.
We claim that $\M,(u,w)\models\bvarr$ also holds. 
Indeed, 
take any $a\in U$ such that $(a,w)\in\wset{z_0}$. As $\betaf_0=\neg\avarr\land\Bh\bvarr$, 
by (ii) we have $\M,(a,w)\models\neg\avarr\land\Bh\bvarr$.
Therefore, $a\ne u$ by \eqref{strictas}, and so $\M,(u,w)\models\bvarr$ follows, as required.
\end{proof}

As by Claim~\ref{c:pathconsok} the consequent of $\badpathf{P}$ holds at $(u,v)$ in $\M$, the proof of 
Lemma~\ref{l:badclinprodgen} is completed.
\end{proof}

As a consequence of Corollary~\ref{co:noimpossible} and Lemmas~\ref{l:bad}, \ref{l:badclinclgen} and \ref{l:badclinprodgen} we also obtain the following `converse' of Lemma~\ref{l:clusterpm}:

\begin{corollary}\label{co:hassolution}
For every countable \grid{} $\F$, if $\F$ is a p-morphic image of a product of difference frames,
then $\F$ contains no impossible \clusters, and $\con{\F}$ has a solution.
\end{corollary}


\section{Infinite canonical axiomatisation for $\dsqxd$}\label{dxdsq}

In this section we prove Theorem~\ref{t:sqax}
using the proof pattern described in \S\ref{proofmethod} (for the class $\Cc$ of all square products of difference frames).
So we will define a recursive
set $\Sigma_{\dsqxd}$ of generalised Sahlqvist formulas containing $\Sigma_{\dxd}$,
 and prove that the following hold:
\begin{enumerate}
\item
All formulas in $\Sigma_{\dsqxd}$ are valid in every square product of difference frames.
\item
For every countable rooted frame $\F$ that is not the p-morphic image of some square product of 
difference frames, there is some $\phi_\F\in\Sigma_{\dsqxd}$ such that $\phi_\F$ is not valid in $\F$.
\end{enumerate}
To begin with, if $\F$ is a countable rooted frame such that it is not the p-morphic image of a product of difference frames at all, then there is some $\phi_\F\in\Sigma_{\dxd}$ such that $\phi_\F$ is not valid in $\F$. So from now on we assume that $\F$ is the p-morphic image of a product of difference frames.
In particular, $\F\models\comm$, and so $\F$ is a \grid{} by Lemma~\ref{l:grid}.
We call a countable \grid{} $\F$ \emph{square-bad} if it is the p-morphic image of a product of difference frames, but 
it is not the p-morphic image of a square product of difference frames.

In \S\ref{sec:sqbad} below we classify square-bad \grids{} into several categories. Then in
\S\ref{sformulassq} we define the generalised Sahlqvist formulas in $\Sigma_{\dsqxd}$,
for each such category.
Finally, in \S\ref{dxdsqsahl} we discuss Corollary~\ref{co:sqS}, that is, why $\dsqxd$ is in fact Sahlqvist axiomatisable.


\subsection{Good \grids{} that are not p-morphic images of squares}\label{sec:sqbad}

Throughout, we suppose that $\F=(W,\Rh,\Rv)$ is square-bad and
represented as a \grid{} as $(X,Y,g)$. 
By Corollary~\ref{co:hassolution}, $\F$ contains no impossible \clusters,  and the set $\con{\F}$ of constraints (as defined in \eqref{conset}) does have a solution. 
By Claim~\ref{c:solpm},
we have 
$\sum_{y\in Y}\xi(y)\ne\sum_{x\in X}\xi(x)$ 
for any solution $\xi$ of $\con{\F}$
(see Figs.~\ref{f:sqbad} and \ref{f:bothways} for some examples).
In particular,
\begin{equation}\label{noinf}
\mbox{there is no solution $\xi$ of $\con{\F}$ such that }\sum_{x\in X}\xi(x)=\sum_{y\in Y}\xi(y)=\infy.
\end{equation}
\begin{figure}
\begin{center}
\setlength{\unitlength}{.035cm}
\begin{picture}(100,100)
\multiput(0,0)(0,50){3}{\line(1,0){100}}
\multiput(0,0)(50,0){3}{\line(0,1){100}}

\multiput(70,3)(0,7.5){6}{\ii}

\put(-12,95){$\F$}
\put(20,77){\ir}
\put(20,65){\rr}

\put(70,65){\rr}
\put(70,77){\rr}

\put(20,27){\ir}
\put(20,15){\rr}

\put(20,-10){$x_1$}
\put(70,-10){$x_2$}
\put(105,20){$y_1$}
\put(105,70){$y_2$}

\end{picture}

\[
\con{\F}=\bigl\{(x_1\geq 3),\,(x_2\geq 4),\,(x_2=6),\,
(y_1=6),\,(y_1\geq 4),\,
(y_2\geq 4),\,(y_1\geq 2x_1),\,(y_2\geq 2x_1)\bigr\}.
\]
%
So, say, $\xi^\ast(x_1)=3$, $\xi^\ast(x_2)=\xi^\ast(y_1)=\xi^\ast(y_2)=6$ is a solution,
but $\xi(y_1)+\xi(y_2)\geq 12> 9=\xi(x_1)+\xi(x_2)$, for every solution $\xi$.
\end{center}
\caption{An example of a square-bad \grid{} $\F$.}\label{f:sqbad}
\end{figure}

%
%
%
%
%
%

%



We begin with introducing some notions that will help us to deal with `upper bound' constraints.
For any $n\in\nNp$ and any $z\in X\cup Y$, we call $z$ $n$-\emph{strict} if 
$(z=n)\in\con{\F}$.
We call $z$ \emph{strict} if it is $n$-strict for some $n\in\nNp$. 
Now recall the digraph $\gf$ from \eqref{gedges}.
We call a $z\in X\cup Y$ \emph{bounded} if there is a path in $\gf$ from some strict $z'$ to $z$,
and \emph{unbounded} otherwise. (In particular, if $z$ is strict then $z$ is bounded.)
Given a bounded $z$ and a path $P$ in $\gf$ of the form
$z_0\to^{\lambda_1} z_1\dots\to^{\lambda_m} z_m$ where $z_m=z$ and $z_0$ is $n$-strict for some $n\in\nNp$, we let
\[
\weight(P)=\left\{
\begin{array}{ll}
n, & \mbox{if $m=0$},\\[3pt]
\bigl\lfloor\frac{n}{\mbox{\scriptsize $\lambda_1\cdot$ $\dots$ $\cdot\lambda_m$}}\bigr\rfloor,
& \mbox{if $m>0$}.
\end{array}
\right.
\]
Then for every bounded $z\in X\cup Y$, we let
\[
\ub{\F}{z}=\min\{\weight(P) : \mbox{$P$ is a path in $\gf$ from some strict node to $z$}\}.
\]
Note that since $\con{\F}$ has a solution, $\ub{\F}{z}=n$ for any $n$-strict node $z$. Also,
it is easy to see that
%
\begin{align}
\label{ubmax}
& \mbox{$ \xi(z)\leq\ub{\F}{z}\in\nNp$, for all bounded $z\in X\cup Y$ and all solutions $\xi$ of $\con{\F}$;}\\
\label{ubmin}
& \mbox{there is a solution $\xi$ of $\con{\F}$ such that $\xi(z)=\infy$ for all unbounded $z\in X\cup Y$.}
\end{align}
%
For any bounded $z$, we choose a simple path $\Pub{\F}{z}$ from $z$ to a strict node such
that $\weight\bigl(\Pub{\F}{z}\bigr)=\ub{\F}{z}$, and if $z$ is strict then $\Pub{\F}{z}$ consists of just $z$.



%
%
%
%
%




\begin{lemma}\label{l:badsqgrid}
Suppose $\F=(X,Y,g)$ is a \grid{} such that $\con{\F}$ is defined and has a solution, but
$\sum_{y\in Y}\xi(y)\ne\sum_{x\in X}\xi(x)$ for any solution $\xi$ of $\con{\F}$.
Then one of the following cases holds:
%
\begin{itemize}
\item[{\rm (i)}]
$X\cup Y$ is finite, and every $z\in X\cup Y$ is bounded.

\item[{\rm (ii)}]
$X\cup Y$ is finite, and
either (a) every $x\in X$ is bounded, there is some unbounded $y^\star\in Y$,
and $\sum_{x\in X }\xi(x) <\sum_{y\in Y}\xi(y)$ for every solution $\xi$ of $\con{\F}$;
or (b) 
every $y\in Y$ is bounded, there is some unbounded $x^\star\in X$,
and $\sum_{y\in Y }\xi(y) <\sum_{x\in X}\xi(x)$ for every solution $\xi$ of $\con{\F}$.

\item[{\rm (iii)}]
Either (a)
$X$ is finite, every $x\in X$ is bounded, $Y$ is infinite, and 
there is a finite subgrid $\F^-=(X,Y^-,g^-)$ of $\F$ such that
$\sum_{x\in X }\xi(x) <\sum_{y\in Y^-}\xi(y)$ for every solution $\xi$ of $\con{\F^-}$;
or (b) 
$Y$ is finite, every $y\in Y$ is bounded, $X$ is infinite, and 
there is a finite subgrid $\F^-=(X^-,Y,g^-)$ of $\F$ such that
$\sum_{y\in Y }\xi(y) <\sum_{x\in X^-}\xi(x)$ for every solution $\xi$ of $\con{\F^-}$.
\end{itemize}
\end{lemma}

\begin{proof}
Observe that since $\con{\F}$ has a solution, 
at least one of $X$ and $Y$ must be finite by \eqref{noinf}. 
So suppose, say, that $X$ is finite.  Suppose also that there is some
unbounded $x\in X$. Then by \eqref{ubmin} and 
 \eqref{noinf} we obtain that $Y$ is finite and every $y\in Y$ is bounded. So at least one of $X$ or $Y$ must be such that it is finite and all its members
are bounded. 
Suppose, say, that $X$ is finite and every $x\in X$ is bounded. There are three cases:

If $Y$ is finite and every $y\in Y$ is bounded, then we have Case~(i). 

If $Y$ is finite and there is some unbounded $y^\star\in Y$, then (as every $x\in X$ is bounded)
the only constraints about $y^\star$ in $\con{\F}$ are of the form $(y^\ast\geq k)$ or $(y^\ast\geq \lambda x)$,
for some $k\in\nNp$, $\lambda=1,2$ and $x\in X$.
So for any solution $\xi$ of $\con{\F}$, if we keep $\xi(z)$ for all $z\ne y^\ast$, and increase $\xi(y^\ast)$ arbitrarily, we obtain another solution. Therefore,
we cannot have
that $\sum_{x\in X}\xi(x)\geq \sum_{y\in Y}\xi(y)$ for any solution $\xi$,
and so we have Case~(ii). 

Finally, suppose that $Y$ is infinite. Then let
\[
Y'=\{y\in Y : \mbox{$y$ occurs in $\Pub{\F}{x}$ for some $x\in X$}\}.
\]
Clearly, $Y'$ is finite. Now take any finite $Y^-\supseteq Y'$ 
such that $|Y^-|>\sum_{x\in X}\ub{\F}{x}$, and
let $\F^-=\bigl(X,Y^-,g|_{X\mprod Y^-}\bigr)$. By \eqref{ubmax}, 
for every solution $\xi$ of $\con{\F^-}$, we have
\[
\sum_{x\in X}\xi(x)\leq\sum_{x\in X}\ub{\F^-}{x}=\sum_{x\in X}\ub{\F}{x}<|Y^-|\leq\sum_{y\in Y^-}\xi(y),
\]
and so we have Case~(iii). 
\end{proof}




\subsection{Generalised Sahlqvist formulas}\label{sformulassq}

In this subsection, we will eliminate the reasons 
for a countable \grid{} $\F$ being square-bad. 
For each of the cases described in Lemma~\ref{l:badsqgrid}, we use a different generalised Sahlqvist formula $\phi^\textit{sq}_\F$.
Throughout this subsection, we assume that
$\F=(W,\Rh,\Rv)$ is 
represented as a \grid{} as $(X,Y,g)$,
$\con{\F}$ is defined and has a solution, but
$\sum_{y\in Y}\xi(y)\ne\sum_{x\in X}\xi(x)$ for any solution $\xi$ of $\con{\F}$.
In \S\ref{s:casefin} and \S\ref{s:caseinf} we will discuss the cases when $X\cup Y$ is finite and infinite,  respectively.


\subsubsection{$X\cup Y$ is finite}\label{s:casefin}

If $X\cup Y$ is finite then,
by Lemma~\ref{l:badsqgrid}, at least one of $X$ and $Y$ is such that all its members are bounded.
Suppose, say, that every $x\in X$ is bounded, and let $\yub$ be the set of bounded members in $Y$.
 (The case when every $y\in Y$ is bounded is similar.)

We will define a formula $\solf$ that is satisfiable in $\F$ (Lemma~\ref{l:solfsat}),
 and `forces' a solution of $\con{\F}$
when satisfied in a product of difference frames (Lemma~\ref{l:solgen}).
The formula $\solf$ will consist of three conjuncts:
$\antefone$ and $\antefthree$ will describe the respective upper and lower bound constraints on any possible solution, while $\anteftwo$ will describe the interactions among switch \clusters{} in $\F$.
We also want $\solf$ to be a generalised
Sahlqvist antecedent, and so it is a problem that the digraph $\gf=(X\cup Y,E_\F)$ might contain cycles.
The following claim says that we can always take a suitable acyclic subgraph of it:

\begin{claim}\label{c:ubgraph}
If $\F$ is finite and every $x\in X$ is bounded, then
there is an acyclic subgraph
 $\hfub=(X\cup Y,\efub)$ of $\gf=(X\cup Y,E_\F)$ such that the following hold:
 \begin{itemize}
 \item[{\rm (i)}]
 Every initial node in $\hfub$ is either strict or belongs to $Y-\yub$.
  \item[{\rm (ii)}]
 $\efub$ contains all the $\to^2$-edges in $E_\F$.
 \item[{\rm (iii)}]
For every edge $z\to^1z'$ in $E_\F$,
 there is an undirected path
 %
 %
  between $z$ and $z'$ in $\hfub$ such that all edges in the path are $\to^1$ edges.
 \end{itemize}
\end{claim}

\begin{proof}
Observe that by \eqref{noinf}, \eqref{nulabel} and \eqref{minS},
no strongly connected component $\scc$ in $\gf$ contain any $\to^2$ edge.
So all cycles in $\gf$ consists of $\to^1$ edges only.
Observe also that if $y\in Y-\yub$, then (as all $x\in X$ are bounded) there are no edges in $E_\F$ of the form $x\to^\lambda y$ for any $x\in X$.
So either $y$ is an isolated node in $\gf$,
or there is an edge $y\to^2 x$ in $E_\F$ for some (maybe several) $x\in X$.
In any case, 
the strongly connected component $y$ belongs to consists of just $y$ alone.

We give an algorithm for how to construct $\efub$ from $E_\F$.
For every strongly connected component $\scc=(S,E_\scc)$ containing only bounded nodes, we define 
step-by-step a subset $E_\scc^-$ of $E_\scc$ such that $(S,E_\scc^-)$ is acyclic.
First, we choose a node $z_\scc$ in $\scc$ as follows. If there is a strict node in $\scc$, then 
let $z_\scc$ be any of the strict nodes in $\scc$. If there is no strict node in $\scc$,
then let $z_\scc$ be any node in $\scc$ such that there is a $\to^2$ edge in $\gf$ starting at some bounded node and ending at $z_\scc$. Let $X_0=\{z_\scc\}$ and $E_0=\emptyset$.
In the inductive step,
take some $z\in S-X_n$, and consider any path $P$ within $\scc$ from some node in $X_n$ to $z$ such that no other node in $P$ is in $X_n$.
Let $E_{n+1}$ consist of the edges in $E_n$ plus the edges in $P$, and let $X_{n+1}$ be obtained from $X_n$ by adding all the nodes in $P$. Clearly,
if $(X_n,E_n)$ is acyclic, then $(X_{n+1},E_{n+1})$ is acyclic as well. We do this until
$X_i=\scc$ for some $i$, and let $E_\scc^-=E_i$.

Now let $\efub$ consist of the edges in $E_\scc^-$ for each $\scc$, plus all 
the $\to^2$-edges in $E_\F$. It is easy to check that $\hfub$ is as required.
\end{proof}


\paragraph{The formula $\antefone$}
We will describe the `bounded' part of $\hfub$, while also
keeping track of the connections with the unbounded nodes in $Y-\yub$.

%
%
%
%
To begin with,
we need to describe that the rows and columns of the grid-structure are pairwise disjoint.
So for every $x\in X$ and every $y\in Y$, we introduce respective fresh propositional variables $\xvar$ and
$\yvar$, and define the formulas
\begin{equation}
\label{xybar}
\xvarbar: \quad\Bv^+\xvar\land\bigwedge_{\substack{x'\in X\\ x'\ne x}}\neg\xvar'
\qquad\mbox{and}\qquad
\yvarbar: \quad\Bh^+\yvar\land\bigwedge_{\substack{y'\in Y\\ y'\ne y}}\neg\yvar'.
\end{equation} 

Next, we need to describe strict nodes in $X\cup\yub$.
Observe that, for every $n\in\nNp$ and every $n$-strict $x\in X$, there exist some $y'\in Y$ 
and distinct $\Rh$-irreflexive points $a_1^x,\dots,a_n^x$ in $\F^{xy'}$.
Similarly, for every $n$-strict $y\in\yub$, there exist some $x'\in X$ 
and distinct $\Rv$-irreflexive points $a_1^y,\dots,a_n^y$ in $\F^{x'y}$.
Thus,
for every $n\in\nNp$ and every $n$-strict $z\in X\cup\yub$, 
we introduce fresh propositional variables $\hvar{z}{i}$ for $i=1,\dots,n$.

We also need to describe the switch \clusters{} in $\hfub$. To simplify notation, for all $z,z'\in X\cup Y$,
we will write
\[
\mbox{$z\to z'$\quad iff\quad $z\to^\lambda z'$ is an edge in $\hfub$ for some $\lambda$,}
\]
and  let 
\[
\swall=\{(x,y)\in X\mprod Y : 
\mbox{ either $x\to y$ or $y\to x$}\}.
\]
Observe that for every $(x,y)\in\swall$, $\F^{xy}$ is a switch \cluster. Therefore, 
$\F^{xy}$ contains a point $c^{xy}$ such that 
(i) $c^{xy}$ is $\ \ri$ when $\F^{xy}$ is of type (\bchvsw) (that is, $x\to^2 y$ is an edge in $\hfub$);
(ii) $c^{xy}$ is $\ \ir$ when $\F^{xy}$ is of type (\bcvhsw) (that is, $y\to^2 x$ is an edge in $\hfub$);
and (iii) $c^{xy}$ is $\ \ii$ when $\F^{xy}$ is of type (\bceqsw) (that is, either $x\to^1y$ or $y\to^1 x$ is an edge in $\hfub$).
Thus, for every $(x,y)\in\swall$, we introduce a fresh propositional variable $\cvar{x}{y}$, and 
define the formula
\begin{equation}\label{cformula}
\ctpp{x}{y}:\ \left\{
\begin{array}{ll}
\displaystyle
\xvarbar\land\yvarbar\land
\neg\cvar{x}{y}\land\Bv\cvar{x}{y}\,\land 
\Dh(\xvarbar\land\yvarbar\land\neg\cvar{x}{y}\land\Bv\cvar{x}{y}),
& \mbox{if $x\to^2 y$ is an edge in $\hfub$},\\[5pt]
\displaystyle
\xvarbar\land\yvarbar\land
\neg\cvar{x}{y}\land\Bh\cvar{x}{y}\,\land \Dv(\xvarbar\land\yvarbar\land\neg\cvar{x}{y}\land\Bh\cvar{x}{y}),
& \mbox{if $y\to^2 x$ is an edge in $\hfub$},\\[5pt]
\displaystyle
\xvarbar\land\yvarbar\land
\neg\cvar{x}{y}\land\Bh\cvar{x}{y}\land\Bv\cvar{x}{y},
& \mbox{if $x\to^1 y$ or $y\to^1 x$}\\
& \hspace*{1.6cm}\mbox{is an edge in $\hfub$}.
\end{array}
\right.
\end{equation}

Let $\hfubm$ be the induced subgraph of $\hfub$ on its bounded nodes, that is, 
on node set $X\cup\yub$.
Starting at each strict node as root, we unravel $\hfubm$ into a forest 
(a disjoint union of directed rooted trees)
 $\tfub$, where each branch of each tree is continued until it reaches either a strict node different from the root or, if there is no such on the branch, a final node in $\hfubm$.
 So for each node $q$ in $\tfub$ there is a unique $z\in X\cup\yub$ such that $q$ is a(n unravelled) copy of $z$. (Each $z\in X\cup\yub$ might have many different copies.)
 For every node $q$ in $\tfub$, we let $\treec{q}$ denote the set of its children in $\tfub$.
 
For every node $q$ in $\tfub$, now we define a formula $\treef{q}$ by induction on the structure of $\tfub$
starting at its leaves:

\begin{itemize}
\item
If $q$ is not a root in $\tfub$, then there is a unique $q^\ast$ with $q\in\treec{q^\ast}$.
There are two cases:

If $q$ is a copy of $x\in X$ and $q^\ast$ is a copy of $y^\ast\in\yub$, then let
%
\[
\treef{q}:\quad\ctpp{x}{y^\ast}\land \bigwedge_{q'\in\treec{q}}\Dh\treef{q'}
\land\bigwedge_{\substack{y'\in\yub,\;y'\ne y^\ast\\ y'\to x}}\Dv\neg\cvar{x}{y'}
\land\bigwedge_{\substack{y'\notin\yub\\ y'\to x}}
\Dv\ctpp{x}{y'}.
\]

If $q$ is a copy of $y\in\yub$ and $q^\ast$ is a copy of $x^\ast\in X$, then let
\[
\treef{q}:\quad\ctpp{x^\ast}{y}\land \bigwedge_{q'\in\treec{q}}\Dv\treef{q'}
\land\bigwedge_{\substack{x'\ne x^\ast\\ x'\to y}}\Dh\neg\cvar{x'}{y}.
\]

%

\item
If $q$ is a root in $\tfub$, then $q$ is a copy of some $n$-strict $z\in X\cup\yub$ for some $n\in\nNp$.
Again, there are two cases:

If $z$ is some $x\in X$, then for each $i=1,\dots,n$, let
%
\[
\treefi{i}{q}:\quad\xvarbar\land\neg\hvar{x}{i}\land\Bh\hvar{x}{i}\land
\bigwedge_{q'\in\treec{q}}\Dv\treef{q'}\land
\bigwedge_{\substack{y'\in\yub\\ y'\to x}}\Dv\neg\cvar{x}{y'}
\land\bigwedge_{\substack{y'\notin\yub\\ y'\to x}}
\Dv\ctpp{x}{y'},
\]
%
and then let
\[
\treef{q}:\quad\bigwedge_{i=1}^{n}\Dh^+\treefi{i}{q}.
\]

If $z$ is some $y\in\yub$, then for each $i=1,\dots,n$, let
\[
\treefi{i}{q}:\quad
\yvarbar\land\neg\hvar{y}{i}\land\Bv\hvar{y}{i}\land
\bigwedge_{q'\in\treec{q}}\Dh\treef{q'}\land
\bigwedge_{\substack{x'\\ x'\to y}}\Dh\neg\cvar{x'}{y},
\]
and then let
\[
\treef{q}:\quad\bigwedge_{i=1}^{n}\Dv^+\treefi{i}{q}.
\]
%
%
%
%
\end{itemize}
Finally,
let $\antefone$ be the conjunction of $\Dh^+\Dv^+\treef{q}$ for all roots $q$ in the forest
$\tfub$ (see Example~\ref{e:sqbad} below).

%
%
%

\paragraph{The `interaction' formula $\anteftwo$} 
%
%
We use the variables introduced for the formula $\antefone$.
For every $x\in X$ and every $y\in\yub$, we define the respective formulas
\begin{align}
\label{swfx}
& \swf{x}:\left\{
\begin{array}{ll}
\displaystyle
\Bh^+\Bv^+\Bigl(\bigwedge_{i=1}^{n}\hvar{x}{i}\to\bigwedge_{\substack{y'\\ x\to y'\ \mbox{\scriptsize or}\ y'\to x}}\Bv^+\cvar{x}{y'}\Bigr), & \mbox{if $x$ is $n$-strict},\\[5pt]
\displaystyle
\Bh^+\Bv^+\bigwedge_{\substack{y',y''\\\ y'\to x\to y''}}\Bigl(\Bv^+\cvar{x}{y'}\to\Bv^+\cvar{x}{y''}\Bigr),
& \mbox{if $x$ is not strict},\\
\end{array}
\right.
\end{align}

\begin{align}
\label{swfy}
& \swf{y}:\left\{
\begin{array}{ll}
\displaystyle
\Bh^+\Bv^+\Bigl(\bigwedge_{i=1}^{n}\hvar{y}{i}\to\bigwedge_{\substack{x'\\ y\to x'\ \mbox{\scriptsize or}\ x'\to y}}\Bh^+\cvar{x'}{y}\Bigr), & \mbox{if $y$ is $n$-strict},\\
\displaystyle
\Bh^+\Bv^+\bigwedge_{\substack{x',x''\\\ x'\to y\to x''}}\Bigl(\Bh^+\cvar{x'}{y}\to\Bh^+\cvar{x''}{y}\Bigr),
& \mbox{if $y$ is not strict},\\
\end{array}
\right.
\end{align}
and let $\anteftwo$ be the conjunction of $\swf{z}$, for all $z\in X\cup\yub$ (see Example~\ref{e:sqbad} below).

%

\paragraph{The formula $\antefthree$} 
We use the $\cvar{x}{y}$ variables introduced for  $\antefone$, and will also introduce some fresh variables.

Observe that by \eqref{minz} and \eqref{noinf}, 
for every 
$z\in X\cup Y$, we have $\minz{z}\in\nNp$ and either \mbox{$\bigl(z= \minz{z}\bigr)$}$\in\con{\F}$ or
$\bigl(z\geq \minz{z}\bigr)\in\con{\F}$. So if $x\in X$, then there is $y_x\in Y$ such that 
$\sizeh\bigl(\F^{xy_x}\bigr)=\minz{x}$, and so there are $\Rh$-reflexive points 
$b_1^\circ(x)$, $\dots$, $b_{\nor{x}}^\circ(x)$ and
$\Rh$-irreflexive points $b_1^\bullet(x)$, $\dots$, $b_{\noi{x}}^\bullet(x)$ in $\F^{xy_x}$ such that
$2\nor{x}+\noi{x} =\minz{x}$.
Similarly, 
if $y\in Y$, then there is $x_y\in X$ such that $\sizev\bigl(\F^{x_yy}\bigr)=\minz{y}$,
and so there are $\Rv$-reflexive points $b_1^\circ(y),\dots,b_{\nor{y}}^\circ(y)$ and
$\Rv$-irreflexive points $b_1^\bullet(y),\dots,b_{\noi{y}}^\bullet(y)$ in $\F^{x_yy}$ such that
$2\nor{y}+\noi{y} = \minz{y}$.

Now recall the function $\smin$ from \eqref{lbdef}. As $\smin$ is a solution of $\con{\F}$ by
Claim~\ref{c:lbsolution}, it follows from \eqref{noinf} that  $\smin(z)\in\nNp$, for every $z\in X\cup Y$.
We define
\begin{align*}
 \xlb & =\bigl\{x\in X : \mbox{$x$ is non-strict and } \minz{x}=\smin(x)\bigr\},\\
 \ylb & =\bigl\{y\in Y : \mbox{$y$ is non-strict and }\minz{y}=\smin(y)\bigr\}.
\end{align*}
%
%
For every $z\in\xlb\cup\ylb\cup(Y-\yub)$,
we introduce fresh propositional variables $\bvar{j}^\circ(z)$ for $j=1,\dots,\nor{z}$, and
$\bvar{s}^\bullet(z)$ for $s=1,\dots,\noi{z}$,
and define the formulas
\begin{align}
\label{bvarnosp2}
\btp{j}^\circ(z) :\quad  & \zvarbar\land\bvar{j}^\circ(z)\land\bigwedge_{\substack{t=1\\ t\ne j}}^{\nor{z}}\neg\bvar{t}^\circ(z)
\land\bigwedge_{t=1}^{\noi{z}}\neg\bvar{t}^\bullet(z),\quad\mbox{for all $j=1,\dots,\nor{z}$},\\[5pt]
\label{bvarnosp22}
\btp{s}^\bullet(z) :\quad & \zvarbar\land\bvar{s}^\bullet(z)\land\bigwedge_{\substack{t=1\\ t\ne s}}^{\noi{z}}\neg\bvar{t}^\bullet(z)
\land\bigwedge_{t=1}^{\nor{z}}\neg\bvar{t}^\circ(z),\quad\mbox{for all $s=1,\dots,\noi{z}$}.
\end{align}
For every $x\in\xlb$, we define
\begin{multline*}
\lbf{x}:\ 
\bigwedge_{j=1}^{\nor{x}}\Dh^+\Bigl[\btp{j}^\circ(x)\land\bigwedge_{\substack{y'\in\yub\\ y'\to x}}\Dv^+\neg\cvar{x}{y'}\land
\Dh\Bigl(\btp{j}^\circ(x)\land\bigwedge_{\substack{y'\in\yub\\ y'\to x}}\Dv^+\neg\cvar{x}{y'}\Bigr)\Bigr]\\[5pt]
\land\bigwedge_{s=1}^{\noi{x}}\Dh^+\Bigl(\btp{s}^\bullet(x)\land\bigwedge_{\substack{y'\in\yub\\ y'\to x}}\Dv^+\neg\cvar{x}{y'}\Bigr),
\end{multline*}
and for every $y\in\ylb\cup(Y-\yub)$, we define
\begin{multline*}
\lbf{y}:\
\bigwedge_{j=1}^{\nor{y}}\Dv^+\Bigl[\btp{j}^\circ(y)\land\bigwedge_{\substack{x'\\ x'\to y}}\Dh^+\neg\cvar{x'}{y}\land
\Dv\Bigl(\btp{j}^\circ(y)\land\bigwedge_{\substack{x'\\ x'\to y}}\Dh^+\neg\cvar{x'}{y}\Bigr)\Bigr]\\[5pt]
\land\bigwedge_{s=1}^{\noi{y}}\Dv^+\Bigl(\btp{s}^\bullet(y)\land\bigwedge_{\substack{x'\\ x'\to y}}\Dh^+\neg\cvar{x'}{y}\Bigr).
\end{multline*}
Let $\antefthree$ be the conjunction of $\Dh^+\Dv^+\lbf{z}$, for all $z\in\xlb\cup\ylb\cup(Y-\yub)$
(see Example~\ref{e:sqbad} below).

\paragraph{The formula $\solf$}
We let
\[
\solf:\quad\antefone\land\anteftwo\land\antefthree.
\]

\begin{lemma}\label{l:solutiondec}
It is decidable whether a bimodal formula is of the form $\solf$ for some finite \grid{} $\F=(X,Y,g)$ 
for which every $X\cup Y$ is bounded,
$\con{\F}$ is defined and has a solution, but $\sum_{y\in Y}\xi(y)\ne\sum_{x\in X}\xi(x)$ for any solution $\xi$ of $\con{\F}$.

It is also decidable whether a bimodal formula is of the form $\solf$ for some finite \grid{} $\F=(X,Y,g)$ 
for which every $x\in X$ is bounded, there is some unbounded $y^\star\in Y$,
$\con{\F}$ is defined and has a solution, but 
$\sum_{x\in X }\xi(x)$ $<\sum_{y\in Y}\xi(y)$ for every solution $\xi$ of $\con{\F}$.
\end{lemma}

\begin{proof}
It is not hard to check that $\solf$ only depends on
\begin{itemize}
\item
the finite acyclic digraph $\hfub$ and the types of \clusters{} corresponding to its edges,
\item
the values $\minz{z}\in\nNp$ for all nodes $z$ in $\hfub$, and
\item
which nodes in $\hfub$ are strict.
\end{itemize}
An inspection of the proof of Claim~\ref{c:ubgraph} shows that it is decidable whether $\hfub$ is
obtained from some finite edge-labelled digraph $\mathcal{G}$ with designated strict nodes and
$\minz{z}$ values.
 And it is clearly decidable whether
such a $\mathcal{G}$ can be obtained as $\gf$ for some finite \grid{} $\F=(X,Y,g)$ as described.
\end{proof}

  
%

\begin{lemma}\label{l:gensahl}
$\solf$ is a generalised Sahlqvist antecedent.
\end{lemma}

\begin{proof}
It is straightforward to check that $\solf$ is a potential generalised Sahlqvist antecedent.
We show that  the dependency digraph $\mathcal{D}(\solf)$ of $\solf$ is acyclic. To this end, observe that
the 
nodes of $\mathcal{D}(\solf)$ are among the $\cvar{x}{y}$ variables occurring in $\anteftwo$, and
we have the following edges $\Rrightarrow$ in $\mathcal{D}(\solf)$:
%
\begin{align}
\label{xedge}
& \mbox{$\cvar{x}{y'}\Rrightarrow\cvar{x}{y''}$,\quad if $x\in X$, $y',y''\in Y$,  $y'\to x$ and $x\to y''$;}\\
\label{yedge}
& \mbox{$\cvar{x'}{y}\Rrightarrow\cvar{x''}{y}$,\quad if $y\in \yub$, $x',x''\in X$,  $x'\to y$ and $y\to x''$.}
\end{align}
We claim that if $Q$ is a path of length $>0$ in $\mathcal{D}(\solf)$
from some $\cvar{x}{y}$ to some $\cvar{x'}{y'}$, 
then either there is a path of length $>0$ in $\hfub$ from $x$ to $x'$,
or there is a path of length $>0$ in $\hfub$ from $y$ to $y'$. Indeed, we show this
 by induction on the length $\ell$ of $Q$. If $\ell=1$ then this follows from 
 \eqref{xedge}--\eqref{yedge}. So suppose $\ell>0$ and $Q$ is $Q^-$ followed by an edge
 of the form, say, $\cvar{x'}{y''}\Rrightarrow\cvar{x'}{y'}$ for some $y''\in Y$. (The other case
 is similar.) 
 By the IH, there are two cases: (i) either there is a path of length $>0$ in $\hfub$ from $x$ to $x'$, in which case we are done, (ii) or there is a path of length $>0$ in $\hfub$ from $y$ to $y''$. 
 As we also have $y''\to x'\to y'$ by \eqref{xedge}, we have a path of length $>0$ in $\hfub$ from $y$ to $y'$, as required. 

Now suppose indirectly that there is a cycle in the dependency digraph of $\solf$.
Choose an arbitrary edge in this cycle of the form, say, $\cvar{x}{y}\Rrightarrow\cvar{x'}{y}$ for some $x,x'\in
X$, $y\in Y$. (The other case
 is similar.) Then $x\to y\to x'$ holds by \eqref{yedge}. Also, either there is a path of length $>0$ in $\hfub$ from $x'$ to $x$,
 or there is a path of length $>0$ in $\hfub$ from $y$ to $y$. In both cases, we have a cycle
 in $\hfub$, contradicting that it is acyclic by Claim~\ref{c:ubgraph}.
 \end{proof}

\begin{lemma}\label{l:solfsat}
If a \grid{} $\F_1=(X_1,Y_1,g_1)$ contains $\F$ as a subgrid,
then $\solf$ is satisfiable in $\F_1$.
\end{lemma}

\begin{proof}
We use the notation introduced in the definition of $\solf$.
We define a model $\M$ on $\F_1$ by taking
\begin{align*}
\M(\xvar) & =  \bigcup_{y\in Y_1}\F_1^{xy},\quad\mbox{for $x\in X$},\\[5pt]
\M(\yvar) & = \bigcup_{x\in X_1}\F_1^{xy},\quad\mbox{for $y\in Y$},\\[5pt]
\M(\cvar{x}{y}) & =  \{w : w\ne c^{xy}\},\quad\mbox{for $(x,y)\in\swall$},\\[5pt]
\M(\hvar{z}{i}) & = \{ w : w\ne a^z_i\},\quad\mbox{for $n\in\nNp$, $i=1,\dots,n$, and $n$-strict $z\in X\cup \yub$},\\[3pt]
\M\bigl(\bvar{j}^\circ(z)\bigr) & = \bigl\{b_j^\circ(z)\bigr\},\quad\mbox{for $z\in\xlb\cup\ylb\cup(Y-\yub)$, $j=1,\dots,\nor{z}$,}\\[3pt]
\M\bigl(\bvar{s}^\bullet(z)\bigr) & = \bigl\{b_s^\bullet(z)\bigr\},\quad\mbox{for $z\in\xlb\cup\ylb\cup(Y-\yub)$, $s=1,\dots,\noi{z}$.}
\end{align*}
It is not hard to check that $\solf$ is satisfied in $\M$. 
\end{proof}

\begin{lemma}\label{l:solgen}
Suppose $\M,(u,v)\models\solf$ for some point $(u,v)$ in a model $\M$ over a product frame $(U,\drel{U})\mprod(V,\drel{V})$. Then for every $x\in X$ there is a set $\hset{x}\subseteq U$, and for every $y\in Y$ there is a set $\hset{y}\subseteq V$
such that the following hold, for every $z\in X\cup Y$:
%
\begin{itemize}
\item[{\rm (i)}]
$\hset{z}\cap\hset{z'}=\emptyset$ whenever $z\ne z'$; $z,z'\in X$ or $z,z'\in Y$.

\item[{\rm (ii)}] 
If $z\in X\cup\yub$ then we can `identify' points outside $\hset{z}$ with a positive formula. In particular:

 %
  %
 %
%
If $z=x\in X$ then for all $a\in U-\hset{x}$,
\begin{align*}
& \M,(a,v)\models\Dv^+\bigwedge_{i=1}^{n}\hvar{x}{i},\quad \mbox{whenever  $x$ is $n$-strict, and}\\[5pt]
& \M,(a, v)\models\bigwedge_{\substack{y\in\yub\\ y\to x}}\Bv^+\cvar{x}{y},\quad \mbox{whenever $x$ is non-strict.}
\end{align*}

If $z=y\in\yub$ then for all $b\in V-\hset{y}$,
\begin{align*}
& \M,(u,b)\models\Dh^+\bigwedge_{i=1}^{n}\hvar{y}{i},\quad \mbox{whenever  $y$ is $n$-strict, and}\\[5pt]
& \M,(u,b)\models\bigwedge_{\substack{x\\ x\to y}}\Bh^+\cvar{x}{y},\quad \mbox{whenever $y$ is non-strict.}
\end{align*}
%

\item[{\rm (iii)}]
$\xi$ is a solution of $\con{\F}$, where $\xi(z)=|\hset{z}|$, for $z\in X\cup Y$.
\end{itemize}
\end{lemma}

\begin{proof}
The argument uses a series of claims.
To begin with,
for every node $q$ in $\tfub$ we will define, inductively on the height of $q$ in $\tfub$, a set $\wset{q}\subseteq U\mprod V$ 
such that, for every $q$, and every $(a,b)\in\wset{q}$, 
\begin{align}
\label{genihx}
& \M,(a,b)\models\xvarbar\land\bigwedge_{q'\in\treec{q}}\Dv\treef{q'},\quad\mbox{if $q$ is a copy of some $x\in X$,}\\[5pt]
\label{genihy}
& \M,(a,b)\models\yvarbar\land\bigwedge_{q'\in\treec{q}}\Dh\treef{q'},\quad\mbox{if $q$ is a copy of some $y\in\yub$.}
\end{align}
%
%
\begin{itemize}
\item
If $q$ is a root and a copy of some $n$-strict $z\in X\cup\yub$, then $\Dh^+\Dv^+\treef{q}$ is a conjunct of $\antefone$, and so
there is $(a,b)$ with $\M,(a,b)\models\treef{q}$. 

So if $z=x\in X$, then there are distinct $a_1,\dots,a_n\in U$  such that
$\M,(a_i,b)\models\treefi{i}{q}$ for $i=1,\dots,n$.
Let $\wset{q}=\{(a_1,b),\dots,(a_n,b)\}$. Then \eqref{genihx} holds for  every $(a,b)\in\wset{q}$.
Observe that
%
\begin{multline}
\label{gentreefi}
\mbox{if $q$ is a root and it is a copy of some $n$-strict $x\in X$, then there is $b$ such that}\\
\mbox{for every $1\leq i\leq n$ there is $(a_i,b)\in\wset{q}$ with $\M,(a_i,b)\models\treefi{i}{q}$.}
\end{multline}

Also,
%
\begin{multline}
\label{thereisi}
\mbox{if $q$ is a root and it is a copy of some $n$-strict $x\in X$, then}\\
\mbox{for every $(a,b)\in\wset{q}$ there exist $1\leq i\leq n$ such that $\M,(a,b)\models\treefi{i}{q}$.}
\end{multline}

Similarly,  if $z=y\in\yub$, then 
there are distinct $b_1,\dots,b_n\in V$ such that
$\M,(a,b_i)\models\treefi{i}{u}$ for $i=1,\dots,n$.
Let $\wset{q}=\{(a,b_1),\dots,(a,b_n)\}$. Then \eqref{genihy} holds for  every $(a,b)\in\wset{q}$.

\item
If $q\in\treec{q^\ast}$, $q$ is a copy of some $x\in X$ and $q^\ast$ is a copy of some $y^\ast\in\yub$
then, by \eqref{genihy} of the IH, we have $\M,(a,b)\models\Dh\treef{q}$ for every $(a,b)\in\wset{q^\ast}$. 
Thus, for every $(a,b)\in\wset{q^\ast}$, there is some $a'\in U$ with $\M,(a',b)\models\treef{q}$.
Let $\wset{q}=\{(a',b) : (a,b)\in\wset{q^\ast}\}$. Then \eqref{genihx} holds for  every $(a,b)\in\wset{q}$.
Observe that 
\begin{align}
\nonumber
& \mbox{if $q\in\treec{q^\ast}$ and $q^\ast$ is a copy of some $y^\ast\in\yub$, then}\\
\label{gennext}
& \hspace*{.5cm}\mbox{for every $(a,b)\in\wset{q^\ast}$ there is $a'$ with $(a',b)\in\wset{q}$,}\\
\label{genprev}
& \hspace*{.5cm}\mbox{and for every $(a',b)\in\wset{q}$ there is $a$ with $(a,b)\in\wset{q^\ast}$.}
\end{align}

Similarly,
if $q\in\treec{q^\ast}$, $q$ is a copy of some $y\in\yub$ and $q^\ast$ is a copy of some $x^\ast\in X$
then, by \eqref{genihx} of the IH, we have $\M,(a,b)\models\Dv\treef{q}$ for every $(a,b)\in\wset{q^\ast}$. 
Thus, for every $(a,b)\in\wset{q^\ast}$, there is some $b'\in V$ with $\M,(a,b')\models\treef{q}$.
Let $\wset{q}=\{(a,b') : (a,b)\in\wset{q^\ast}\}$. Then \eqref{genihy} holds for  every $(a,b)\in\wset{q}$.

Observe that 
%
\begin{equation}
\label{gentreef}
\mbox{if $q$ is not a root, then  $\M,(a,b)\models\treef{q}$ for every $(a,b)\in\wset{q}$.}
\end{equation}
%

\end{itemize}
Next, for every node $q$ in $\tfub$ that is a copy of some $x\in X$, we let
\begin{equation}\label{hsetdef}
\hset{q}=\{ a\in U : (a,b)\in\wset{q}\mbox{ for some $b$}\},\mbox{ and}
\end{equation}
for every node $q$ that is a copy of some $y\in\yub$, we let
\begin{equation}\label{vsetdef}
\hset{q}=\{ b\in V : (a,b)\in\wset{q}\mbox{ for some $a$}\}.
\end{equation}
Observe that 
%
\begin{equation}
\label{strictwset} 
\mbox{if $q$ is a root in $\tfub$ and it is a copy of some $n$-strict $z\in X\cup Y$, then $|\hset{q}|=n$.}
\end{equation}
%

\begin{lclaim}\label{c:out}
For every node $q$ in $\tfub$, the following hold:
\begin{itemize}
\item[{\rm (i)}]
 If $q$ is a root and a copy of some $n$-strict $x\in X$,
 then $\M,(a^\ast,v)\models\Dv^+\bigwedge_{i=1}^n\hvar{x}{i}$ 
 for all $a^\ast\in U-\hset{q}$.

 If $q$ is a root and a copy of some $n$-strict $y\in\yub$, then
 $\M,(u,b^\ast)\models\Dh^+\bigwedge_{i=1}^n\hvar{y}{i}$
 for all $b^\ast\in V-\hset{q}$.

 \item[{\rm (ii)}]
If $q$ is a copy of some $x\in X$, $q\in\treec{q'}$,  and $q'$ is a copy of some $y'\in\yub$, then
$\M,(a^\ast, v)\models\Bv^+\cvar{x}{y'}$
for all $a^\ast\in U-\hset{q}$.

If $q$ is a copy of some $y\in\yub$, $q\in\treec{q'}$,  and $q'$ is a copy of some $x'\in X$, then
$\M,(u,b^\ast)\models\Bh^+\cvar{x'}{y}$
for all $b^\ast\in V-\hset{q}$.
\end{itemize}
%
\end{lclaim}

\begin{proof}
We prove the claim by induction on the height of $q$ in $\tfub$.

(i): Suppose $q$ is a root and a copy of some $n$-strict $x\in X$, and take some $a^\ast\in U-\hset{q}$.
By \eqref{gentreefi}, there is $b$ such that for every $1\leq i\leq n$ there is $(a_i,b)\in\wset{q}$ with
$\M,(a_i,b)\models\treefi{i}{q}$. As $\Bh\hvar{x}{i}$ is a conjunct of $\treefi{i}{q}$,
we have $\M,(a_i,b)\models\Bh\hvar{x}{i}$ for every $1\leq i\leq n$.
As by \eqref{hsetdef} $a_i\in\hset{q}$ for every $1\leq i\leq n$, we have that $a^\ast\ne a_i$ for any $i$,
and so $\M,(a^\ast,b)\models\bigwedge_{i=1}^n\hvar{x}{i}$.
The case when  $q$ is a copy of some $n$-strict $y\in\yub$ is similar.

(ii):
Suppose $q\in\treec{q'}$, $q$ is a copy of some $x\in X$, $q'$ is a copy of some $y'\in\yub$.
(The case when $q\in\treec{q'}$, $q$ is a copy of some $y\in\yub$ and $q'$ is a copy of some $x'\in X$ is similar.)
Suppose inductively that we have (i)--(ii) for $q'$, and take some $a^\ast\in U-\hset{q}$. 
We claim that 
\begin{equation}\label{ceverywhere}
\M,(a^\ast,b^\ast)\models\cvar{x}{y'},\quad\mbox{for every $b^\ast\in V$,}
\end{equation}
implying $\M,(a^\ast, v)\models\Bv^+\cvar{x}{y'}$, as required.
Indeed, take some $b^\ast\in V$. There are two cases, either $b^\ast\in\hset{q'}$ or $b^\ast\notin\hset{q'}$.
If $b^\ast\in\hset{q'}$ then there is some $a$ such that $(a,b^\ast)\in\wset{q'}$ by \eqref{vsetdef}.
Thus, there is $a'$ such that $(a',b^\ast)\in\wset{q}$ by \eqref{gennext}, and so 
$\M,(a',b^\ast)\models\treef{q}$ by \eqref{gentreef}.  As $\ctpp{x}{y'}$ is a conjunct of $\treef{q}$, we also have
$\M,(a',b^\ast)\models\ctpp{x}{y'}$. 
As $q\in\treec{q'}$, we have $y'\to x$. Thus,
$\Bh\cvar{x}{y'}$ is a conjunct of $\ctpp{x}{y'}$, and so 
$\M,(a',b^\ast)\models \Bh\cvar{x}{y'}$ as well. As $a^\ast\notin\hset{q}$ but $a'\in\hset{q}$ by  \eqref{hsetdef}, it follows that $a'\ne a^\ast$, and so \eqref{ceverywhere} holds.

If $b^\ast\notin\hset{q'}$ then suppose first that $q'$ is a root, that is, $y'$ is $n$-strict for some $n\in\nNp$. By item (i) of the IH, $\M,(u,b^\ast)\models\Dh^+\bigwedge_{i=1}^{n}\hvar{y'}{i}$.
As $y'\to x$ holds,
$\M,(u,b^\ast)\models\Bh^+\cvar{x}{y'}$ follows by \eqref{swfy},
 and so \eqref{ceverywhere} holds.
Finally, suppose that $q'$ is not a root.
 Let $q''$ be such that $q'\in\treec{q''}$, and suppose that $q''$ is a copy of some $x''\in X$.
Then $\M,(u,b^\ast)\models\Bh^+\cvar{x''}{y'}$ by item (ii) of the IH.
As $x''\to y'\to x$ holds, 
$\M,(u,b^\ast)\models\Bh^+\cvar{x}{y'}$ follows by \eqref{swfy}.
Thus, \eqref{ceverywhere} holds in this case as well.
\end{proof}

\begin{lclaim}\label{c:countok}
For all nodes $q,q'$ in $\tfub$, the following hold:
 \begin{itemize}
\item[{\rm (i)}]
If $q'\to^2 q$ is an edge in $\tfub$ then $|\hset{q'}|\geq 2\cdot |\hset{q}|$.

 \item[{\rm (ii)}]
If $q'\to^1 q$ is an edge in $\tfub$ then $|\hset{q}|= |\hset{q'}|$.
\end{itemize}
%
\end{lclaim}

\begin{proof}
Suppose $q'\to^\lambda q$ is an edge in $\tfub$ for some $\lambda$,
$q$ is a copy of some $x\in X$, $q'$ is a copy of some $y'\in\yub$.
(The case when $q$ is a copy of some $y\in\yub$ and $q'$ is a copy of some $x\in X$ is similar.)
By \eqref{hsetdef}, for every $a\in\hset{q}$ there is $b_a$ with $(a,b_a)\in\wset{q}$.
Then $\M,(a,b_a)\models\treef{q}$ by \eqref{gentreef}. As $\ctpp{x}{y'}$ is a conjunct of $\treef{q}$, we also have
$\M,(a,b_a)\models\ctpp{x}{y'}$. 
By  \eqref{vsetdef}, we have
\begin{equation}\label{firstin}
\mbox{$b_a\in\hset{q'}$ for every $a\in\hset{q}$.}
\end{equation}

(i):
As $q'\to^2 q$ is an edge in $\tfub$, $y'\to^2 x$ is an edge in $\hfub$. Thus, $\ctpp{x}{y'}$ implies
$\neg\cvar{x}{y'}\land\Bh\cvar{x}{y'}\land\Dv(\neg\cvar{x}{y'}\land\Bh\cvar{x}{y'})$ (see \eqref{cformula}), and so $\M,(a,b_a)\models\neg\cvar{x}{y'}\land\Bh\cvar{x}{y'}\land\Dv(\neg\cvar{x}{y'}\land\Bh\cvar{x}{y'})$. So for every $a\in\hset{q}$, there is also a $b'_a\ne b_a$
with $\M,(a,b'_a)\models\neg\cvar{x}{y'}\land\Bh\cvar{x}{y'}$.
Therefore
%
\begin{equation}
\label{allfour}
\mbox{if $a_1,a_2\in\hset{q}$ and $a_1\ne a_2$, then $b_{a_1}$, $b'_{a_1}$, $b_{a_2}$, and
$b'_{a_2}$ are four distinct points.}
\end{equation}
%
We claim that
 \begin{equation}\label{ezisin}
 b'_a\in\hset{q'}\ \mbox{for every $a\in\hset{q}$}.
 \end{equation}
 Indeed,
 suppose indirectly that $b'_a\notin\hset{q'}$ for some $a\in\hset{q}$. 
 There are two cases: If $q'$ is a root and $y'$ is $n$-strict for some $n\in\nNp$, then 
 $\M,(u,b'_a)\models\Dh^+\bigwedge_{i=1}^n\hvar{y'}{i}$
 by Claim~\ref{c:out}~(i). As $y'\to x$ holds, 
$\M,(u,b_a')\models\Bh^+\cvar{x}{y'}$ follows by \eqref{swfy},
contradicting $\M,(a,b'_a)\models\neg\cvar{x}{y'}$.
If $q'\in\treec{q''}$ for some $q''$ and $q''$ is a copy of some $x''\in X$, then by Claim~\ref{c:out}~(ii)
 we have that $\M,(u,b'_a)\models\Bh^+\cvar{x''}{y'}$. As $x''\to y'\to x$ holds,
 $\M,(u,b'_a)\models\Bh^+\cvar{x}{y'}$ follows by \eqref{swfy}, a contradiction again, proving \eqref{ezisin}.
 Now $|\hset{q'}|\geq 2\cdot |\hset{q}|$ follows by \eqref{firstin},
\eqref{allfour} and \eqref{ezisin}.
 
 (ii):
 As $q'\to^1 q$ is an edge in $\tfub$, $y'\to^1 x$ is an edge in $\hfub$.
So $\neg\cvar{x}{y'}\land\Bh\cvar{x}{y'}$ is a conjunct of $\ctpp{x}{y'}$,
and so $\M,(a,b_a)\models\neg\cvar{x}{y'}\land\Bh\cvar{x}{y'}$ as well.
So if $a_1\ne a_2\in\hset{q}$ then $b_{a_1}\ne b_{a_2}$ must hold, and so
 $|\hset{q}|\leq |\hset{q'}|$ by \eqref{firstin}.
 On the other hand, 
 by \eqref{vsetdef} and \eqref{hsetdef}, for every $b\in\hset{q'}$ there is $a_b\in\hset{q}$ such that 
 $(a_b,b)\in\wset{q}$, and so $\M,(a_b,b)\models\treef{q}$ by \eqref{gentreef}.  
 As $\ctpp{x}{y'}$ is a conjunct of $\treef{q}$, we also have
$\M,(a_b,b)\models\ctpp{x}{y'}$. 
As $\neg\cvar{x}{y'}\land\Bv\cvar{x}{y'}$ is a conjunct of $\ctpp{x}{y'}$, we have
 $\M,(a_b,b)\models\neg\cvar{x}{y'}\land\Bv\cvar{x}{y'}$ as well.
 So if $b_1\ne b_2\in\hset{q'}$ then $a_{b_1}\ne a_{b_2}$ must hold, and so
 $|\hset{q'}|\leq |\hset{q}|$.
\end{proof}

\begin{lclaim}\label{c:psets}
For all nodes $q_1,q_2$ in $\tfub$, 
if $q_1$ and $q_2$ are both copies of the same $z\in X\cup\yub$, then $\hset{q_1}= \hset{q_2}$.
\end{lclaim}

\begin{proof}
Clearly,
it is enough to show that $\hset{q_1}\subseteq \hset{q_2}$.
Suppose that $q_1\ne q_2$ are both copies of the same $x\in X$, and there is some 
$a\in\hset{q_1}-\hset{q_2}$. (The case when $q_1\ne q_2$ are both copies of the same $y\in\yub$ is similar.)
As $a\in\hset{q_1}$, there is $b$ with $(a,b)\in\wset{q_1}$ by \eqref{hsetdef}.
It cannot be that both $q_1$ and $q_2$ are roots in $\tfub$, so there are three cases:

If $q_1$ is a root and $q_2$ is not a root. Then suppose $x$ is $n$-strict for some $n\in\nNp$, 
$q_2\in\treec{q_2'}$, and $q_2'$ is a copy of some $y''\in\yub$.
  So by \eqref{thereisi} there exist $1\leq i\leq n$ such that $\M,(a,b)\models\treefi{i}{q_1}$.
As $y''\to x$ holds, $\Dv\neg\cvar{x}{y''}$ is a conjunct of $\treefi{i}{q_1}$, and so
$\M,(a,b)\models\Dv\neg\cvar{x}{y''}$. On the other hand, as $a\notin\hset{q_2}$, by Claim~\ref{c:out}~(ii)
 we have $(a,v)\models\Bv^+\cvar{x}{y''}$,
a contradiction.

If $q_2$ is a root and $q_1$ is not a root. Again, suppose $x$ is $n$-strict for some $n\in\nNp$, 
$q_1\in\treec{q_1'}$, and $q_1'$ is a copy of some $y'\in\yub$. 
We have  $\M,(a,b)\models\treef{q_1}$ by \eqref{gentreef}.
As $\ctpp{x}{y'}$ is a conjunct of $\treef{q_1}$, we have $\M,(a,b)\models\ctpp{x}{y'}$.
As $\neg\cvar{x}{y'}$  is a conjunct of $\ctpp{x}{y'}$,  we have $\M,(a,b)\models\neg\cvar{x}{y'}$.
On the other hand, as $a\notin\hset{q_2}$, by Claim~\ref{c:out}~(i) we have 
$\M,(a,v)\models\Dv^+\bigwedge_{i=1}^n\hvar{x}{i}$.
As $y'\to x$ holds,
$\M,(a,v)\models\Bv^+\cvar{x}{y'}$ by \eqref{swfx},
a contradiction.

If neither $q_1$ nor $q_2$ is a root. Then suppose $q_1\in\treec{q_1'}$, $q_2\in\treec{q_2'}$,
$q_1'$ is a copy of $y'\in\yub$, and $q_2'$ is a copy of $y''\in\yub$. 
We have $\M,(a,b)\models\treef{q_1}$ by \eqref{gentreef}.
As $\ctpp{x}{y'}$ is a conjunct of $\treef{q_1}$, we have $\M,(a,b)\models\ctpp{x}{y'}$.
As $\neg\cvar{x}{y'}$  is a conjunct of $\ctpp{x}{y'}$,  we have $\M,(a,b)\models\neg\cvar{x}{y'}$.
On the other hand, as $a\notin\hset{q_2}$, by Claim~\ref{c:out}~(ii)
we have $\M,(a,v)\models\Bv^+\cvar{x}{y''}$. If $y'=y''$, this is a contradiction.
If $y'\ne y''$ then $\Dv\neg\cvar{x}{y''}$ is a conjunct of $\treef{q_1}$, and so 
$\M,(a,b)\models\Dv\neg\cvar{x}{y''}$,
a contradiction again.
\end{proof}

Next, for every $z\in X\cup Y$, we will define $\hset{z}$. There are two cases: 
\begin{itemize}
\item
If $z\in X\cup\yub$, then let
\begin{equation}\label{boundedsetdef}
\mbox{$\hset{z}=\hset{q}$ for some (any) copy $q$ of $z$.}
\end{equation}
This is well-defined, as
$\tfub$ contains some copy of every $z\in X\cup\yub$,
and the definition does not depend on the choice of the particular copy by Claim~\ref{c:psets}.

\item
If $y\in Y-\yub$ then
$\Dh^+\Dv^+\lbf{y}$ is a conjunct of $\antefthree$, and so there exist $a\in U$ and distinct $b_1,\dots,b_{2\nor{y}},b_1',\dots,b_{\noi{y}}'\in V$ such that 
\begin{align}
\label{lbgen1}
& \mbox{$\M,(a,b_j)\models\btp{j}^\circ(y)$ and $\M,(a,b_{\nor{y}+j})\models\btp{j}^\circ(y)$, for all $1\leq j\leq \nor{y}$,}\\
\label{lbgen2}
& \mbox{$\M,(a,b_s')\models\btp{s}^\bullet(y)$ for all $1\leq s\leq \noi{y}$.}
\end{align}
 We let
\begin{equation}\label{lbsetdef}
\hsetlb{y}=\{b_1,\dots,b_{2\nor{y}},b_1',\dots,b_{\noi{y}}'\}.
\end{equation}
As $2\nor{y}+\noi{y}=\smin(y)$, we have that
\begin{equation}\label{lbsetok}
|\hsetlb{y}|=\smin(y).
\end{equation}

Next, for each $x\in X$ such that $y\to^2 x$ is an edge in $\gf$, we will define a
set $\hsetx{x}{y}$ such that
\begin{equation}\label{hsetxok}
|\hsetx{x}{y}|\geq 2\cdot |\hset{x}|.
\end{equation}
To this end, first we claim that
\begin{multline}
\label{xyconjunct}
\mbox{there is a copy $q$ of $x$ in $\tfub$ such that}\\
\mbox{$\M,(a,b)\models\Dv\ctpp{x}{y}$,
for every $(a,b)\in\wset{q}$.}
\end{multline}
Indeed, there are two cases.
If $x$ is $n$-strict for some $n\in\nNp$, then choose a copy $q$ of $x$ that is a root in $\tfub$.
Then by \eqref{gentreefi} and \eqref{thereisi},
there are distinct $a_1,\dots,a_k\in U$ and $b\in V$ such that
$(a_i,b)\in\wset{q}$, and so $\M,(a_i,b)\models\treefi{i}{q}$, for all $1\leq i\leq n$.
As $\Dv\ctpp{x}{y}$
is a conjunct of $\treefi{i}{q}$ for every $1\leq i\leq n$, \eqref{xyconjunct} follows.
If $x$ is not strict, then choose a copy $q$ of $x$ that is not a root in $\tfub$.
By \eqref{gentreef}, $\M,(a,b)\models\treef{q}$ for  every $(a,b)\in\wset{q}$.
As 
$\Dv\ctpp{x}{y}$
is a conjunct of $\treef{q}$, again we have \eqref{xyconjunct}.

As $y\to^2 x$ is an edge in $\gf$, $\ctpp{x}{y}$ implies
$\Dv\bigl(\yvarbar\land\neg\cvar{x}{y}\land\Bh\cvar{x}{y}\land\Dv(\yvarbar\land\neg\cvar{x}{y}\land\Bh\cvar{x}{y})\bigr)$
(see \eqref{cformula}).
So by \eqref{xyconjunct}, \eqref{hsetdef} and \eqref{boundedsetdef}, for  every $a\in\hset{x}$
there are $b_{a}\ne b_{a}'\in V$ such
that
\begin{equation}\label{xyset}
\mbox{$\M,(a,b_{a})\models\yvarbar\land\neg\cvar{x}{y}\land\Bh\cvar{x}{y}$\  and\ 
$\M,(a,b'_{a})\models\yvarbar\land\neg\cvar{x}{y}\land\Bh\cvar{x}{y}$.}
\end{equation}
 We let
\begin{equation}\label{xysetdef}
\hsetx{x}{y}=\{b_{a},b'_{a} : a\in\hset{x}\}.
\end{equation}
Clearly, if $a_1,a_2\in\hset{x}$ and $a_1\ne a_2$, then $b_{a_1}$, $b_{a_1}'$, $b_{a_2}$, and $b_{a_2}'$ are four distinct points, and so \eqref{hsetxok} follows, as required.

Finally, let
%
\[
\hset{y}=\hsetlb{y}\cup\bigcup_{\substack{x\\ y\to x}}\hsetx{x}{y}.
\]
%
Thus, by \eqref{lbsetok} and \eqref{hsetxok}, respectively, we obtain that 
\begin{align}
\label{hsetokub1}
& |\hset{y}|\geq\smin(y),\ \mbox{ and}\\[3pt]
\label{hsetokub2}
& |\hset{y}|\geq 2\cdot |\hset{x}|,\
\mbox{ for every $x\in X$ such that $y\to^2x$ is an edge in $\gf$.}
\end{align}
\end{itemize}

\begin{lclaim}\label{c:zok}
For all $z,z'\in X\cup Y$, the following hold:
 \begin{itemize}
  \item[{\rm (i)}]
 If $z$ is $n$-strict for some $n\in\nNp$, then $|\hset{z}|=n$.

 \item[{\rm (ii)}]
If $z'\to^2 z$ is an edge in $\gf$, then $|\hset{z'}|\geq 2\cdot |\hset{z}|$.

 \item[{\rm (iii)}]
 If $z$ and $z'$ are in the same strongly connected component of $\gf$, then
 $|\hset{z}|= |\hset{z'}|$.
 
   \item[{\rm (iv)}]
$|\hset{z}|\geq\smin(z)$. 
\end{itemize}
\end{lclaim}

\begin{proof}
Item (i) is by \eqref{strictwset}.

(ii):
If $z,z'\in X\cup\yub$, then
 there are nodes $q$ and $q'$ in $\tfub$
such that $q$ is a copy of $z$, $q'$ is a copy of $z'$, and $q'\to^2 q$ is an edge in $\tfub$.
So $|\hset{z'}|\geq 2\cdot |\hset{z}|$ follows by Claim~\ref{c:countok}~(i).
If $z'=y\in Y-\yub$ and $z=x\in X$, then $y\to^2 x$ is an edge in $\gf$. So 
$|\hset{y}|\geq 2\cdot |\hset{x}|$ follows by \eqref{hsetokub2}.

(iii):
Suppose $z\ne z'$. Then $z,z'\in X\cup\yub$, and by Claim~\ref{c:ubgraph}~(iii),
there is an undirected path $P$ in $\hfub$ between $z$ and $z'$  such that all edges in the path are $\to^1$ edges.
We can break $P$ up to a union of directed paths in $\hfub$ (each of which has copies in the unravelling $\tfub$), 
and then $|\hset{z}|=|\hset{z'}|$ follows by (possibly repeated applications of) Claim~\ref{c:countok}~(ii).

(iv):
It is enough to show that for every strongly connected component $\scc$ in $\gf$,
\begin{equation}\label{allzins}
|\hset{z}|\geq\numin(\scc),\quad\mbox{for every $z$ in $\scc$.}
\end{equation}
To this end, observe that for every $\scc$, we have $\numin(\scc)\in\nNp$
by \eqref{noinf}, \eqref{lbdef} and Claim~\ref{c:lbsolution}, and so
by \eqref{nulabel}, \eqref{minS} and \eqref{minz},
\begin{align}
\label{numininN}
& \numin(\scc)=\max\bigl(\{2\cdot\numin(\scc') : \scc\Rightarrow \scc'\}\cup\{\minS{\scc}\}\bigr)\in\nNp,\\
\nonumber
& \minS{\scc}=\max\{\minz{z} : z\in\scc\}\in\nNp,\\
\label{mininN}
& \minz{z}=\max\bigl\{ k : \mbox{either }(z= k)\in\con{\F}\mbox{ or }(z\geq k)\in\con{\F}\bigr\}\in\nNp,\
\mbox{for every $z\in\scc$}.
\end{align}
Therefore,
\[
\minS{\scc}=\max\bigl\{ k : \mbox{either }(z= k)\in\con{\F}\mbox{ or }(z\geq k)\in\con{\F}\mbox{ for some $z\in\scc$}\bigr\}\in\nNp.
\]
Let $z^\ast\in\scc$ be such that $\minS{\scc}=\minz{z^\ast}$.
%

We prove \eqref{allzins} by induction on $\rank(\scc)$. Suppose $\rank(\scc)=0$.
Then $\smin(z^\ast)=\numin(\scc)=\minS{\scc}=\minz{z^\ast}$.
There are two cases:
\begin{itemize}
\item[(a)]
$z^\ast$ is $n$-strict for some $n$.
As $\con{\F}$ has a solution, $n=\minz{z^\ast}$ must hold by \eqref{mininN}.
Thus,
$|\hset{z^\ast}|=\minz{z^\ast}$ by item (i), 
and so \eqref{allzins} follows by item~(iii). 

\item[(b)]
$z^\ast$ is non-strict.
%
If $z^\ast\in Y-\yub$ in $\scc$, then $\scc=\{z^\ast\}$ and so  \eqref{allzins} follows by \eqref{hsetokub1}.
If $z^\ast=x^\ast\in\xlb$, then 
$y\to x^\ast$ holds for some $y\in\yub$. Also,
$\Dh^+\Dv^+\lbf{x^\ast}$ is a conjunct of $\antefthree$, and so there exist $b\in V$ and distinct 
$a_1,\dots,a_{\minz{x^\ast}}\in U$ such that 
$\M,(a_i,b)\models\Dv^+\neg\cvar{x}{y}$ for every $1\leq i\leq\minz{x^\ast}$.
By Claim~\ref{c:out}~(ii), $a_i\in\hset{x^\ast}$ for every $1\leq i\leq\minz{x^\ast}$. Thus 
$|\hset{x^\ast}|\geq\minz{x^\ast}=\numin(\scc)$,
and so \eqref{allzins} follows by item (iii).
(The case when $z^\ast\in (\ylb\cap\yub)$ is similar.)
 \end{itemize}
Now take some $\scc$ with $\rank(\scc)>0$, and suppose inductively that  \eqref{allzins} holds for every $\scc'$ with $\rank(\scc')<\rank(\scc)$.
There are two cases: If $\numin(\scc)=\minS{\scc}$, then \eqref{allzins} can be shown as in (a)--(b) above,
using $z^\ast\in\scc$,
Otherwise, by \eqref{numininN} there is $\scc'$ such that $\numin(\scc)=2\cdot\numin(\scc')$ and
$\scc\Rightarrow\scc'$. Then there exist $z_1$ in $\scc$ and $z_2$ in $\scc'$ such that $z_1^\ast\to^2 z_2$
is an edge in $\gf$. Therefore, by item (ii) and the IH, we have
$|\hset{z_1}|\geq 2\cdot |\hset{z_2}|\geq 2\cdot\numin(\scc')=\numin(\scc)$, and so
 \eqref{allzins} follows by item (iii).
\end{proof}


Finally, we can complete the proof of Lemma~\ref{l:solgen}:

Item (i): We claim that
\begin{equation}\label{yvarok}
\mbox{for every $y\in Y$ and every $b\in\hset{y}$ there is $a$ such that
$\M,(a,b)\models\yvarbar$.}
\end{equation}
Indeed, there are three cases.
If $y\in\yub$, then $b\in\hset{q}$ for some copy $q$ of $y$ by \eqref{boundedsetdef}.
So there is $a$ such that $(a,b)\in\wset{q}$ by \eqref{vsetdef}, and so \eqref{yvarok} follows
by \eqref{genihy}.
If $y\in Y-\yub$ and $b\in\hsetx{x}{y}$ for some $x$ with $y\to x$, then  \eqref{yvarok} follows
by \eqref{xyset} and \eqref{xysetdef}.
If $y\in Y-\yub$ and $b\in\hsetlb{y}$, then  \eqref{yvarok}  follows from \eqref{lbgen1}--\eqref{lbsetdef}
and from the fact that $\yvarbar$ is a conjunct of each $\btp{j}^\circ(y)$ and $\btp{s}^\bullet(y)$
(see \eqref{bvarnosp2}--\eqref{bvarnosp22}).

Now suppose indirectly that $y\ne y'\in Y$ and there is some $b\in\hset{y}\cap\hset{y'}$. 
By \eqref{yvarok}, there are $a$ and $a'$ such that 
$\M,(a,b)\models\yvarbar$ and $\M,(a',b)\models\yvarbar'$, and so 
$\M,(a,b)\models\Bh^+\yvar$ and $\M, (a',b)\models\neg\yvar$ by \eqref{xybar}, a contradiction.
(The case of $x,x'\in X$ is similar, using the $\xvar$ variables.)

Item (ii) follows from Claim~\ref{c:out}.

Item (iii):
Constraints of the form $(z=n)\in\con{\F}$
hold by Claim~\ref{c:zok}~(i).
Constraints of the form $(z'\geq 2 z)\in\con{\F}$
hold by Claim~\ref{c:zok}~(ii).
Constraints of the form $(z'\geq z)\in\con{\F}$
hold by Claim~\ref{c:zok}~(iii).
 Finally, consider a constraint of the form $(z\geq k)\in\con{\F}$, for some $k\in\nNp$.
 As $|\hset{z}|\geq\smin(z)$ by Claim~\ref{c:zok}~(iv), we have $|\hset{z}|\geq k$ by \eqref{lbdef} and Claim~\ref{c:lbsolution}.
\end{proof}


\paragraph{The consequent $\consf$ of the generalised Sahlqvist implication}
We will use the positive formulas given in Lemma~\ref{l:solgen}~(ii).
%
For every $x\in X$, we let
\begin{equation}\label{lfinal}
\finalf{x}:\quad\left\{
\begin{array}{ll}
\displaystyle\Dv^+\bigwedge_{i=1}^{n}\hvar{x}{i}, & \mbox{if $x$ is $n$-strict for some $n\in\nNp$},\\[5pt]
\displaystyle
\bigwedge_{\substack{y'\in\yub\\ y'\to x}} \Bv^+\cvar{x}{y'}, & \mbox{if $x$ is non-strict}.
\end{array}
\right.
\end{equation}
Similarly, 
for every $y\in\yub$, we let
%
\[
\finalf{y}:\quad\left\{
\begin{array}{ll}
\displaystyle\Dh^+\bigwedge_{i=1}^{n}\hvar{y}{i}, & \mbox{if $y$ is $n$-strict for some $n\in\nNp$},\\[5pt]
\displaystyle
\bigwedge_{\substack{x'\\ x'\to y}} \Bh^+\cvar{x'}{y}, & \mbox{if $y$ is non-strict}.
\end{array}
\right.
\]
Then we let
\begin{align*}
\consf:\quad & \left\{
\begin{array}{ll}
\displaystyle
\Dh^+\bigwedge_{x\in X}\finalf{x}\lor\Dv^+\bigwedge_{y\in\yub}\finalf{y}, & \mbox{if $\yub=Y$},\\[5pt]
\displaystyle
\Dh^+\bigwedge_{x\in X}\finalf{x}, & \mbox{if $Y-\yub\ne\emptyset$.}
\end{array}
\right.\\[5pt]
\badsqf:\quad & \solf\to\consf
\end{align*}
(see Example~\ref{e:sqbad} below).
Using Lemma~\ref{l:gensahl},
it is straightforward to check that $\badsqf$ is a generalised Sahlqvist formula.
Also, by Lemma~\ref{l:solutiondec}, it is easy to see the following:

\begin{lemma}\label{l:badsqdec}
It is decidable whether a bimodal formula is of the form $\badsqf$ for some finite \grid{} $\F=(X,Y,g)$ 
for which every $z\in X\cup Y$ is bounded,
$\con{\F}$ is defined and has a solution, but $\sum_{y\in Y}\xi(y)\ne\sum_{x\in X}\xi(x)$ for any solution $\xi$ of $\con{\F}$.

It is also decidable whether a bimodal formula is of the form $\badsqf$ for some finite \grid{} $\F=(X,Y,g)$ 
for which every $x\in X$ is bounded, there is some unbounded $y^\star\in Y$,
$\con{\F}$ is defined and has a solution, but 
$\sum_{x\in X }\xi(x) <\sum_{y\in Y}\xi(y)$ for every solution $\xi$ of $\con{\F}$.
\end{lemma}

\begin{lemma}\label{l:finalsat}
$\badsqf$ is not valid in $\F$.
\end{lemma}

\begin{proof}
Let $\F_1=\F$ and take the model $\M$ on $\F_1$ from the proof of Lemma~\ref{l:solfsat} satisfying $\solf$.
It is easy to see that $\neg\consf$ is satisfied in $\M$ as well.
\end{proof}

\begin{lemma}\label{l:finalfinsq}
$\badsqf$ is valid in every square product of difference frames.
\end{lemma}

\begin{proof}
Suppose $\M$ is a model on $(U,\drel{U})\mprod(V,\drel{V})$ for some $U,V$ with $|U|=|V|>0$, and
$\M,(u,v)\models\solf$.
For every $z\in X\cup Y$, take the set $\hset{z}$ from Lemma~\ref{l:solgen}.
There are two cases:
\begin{itemize}
\item 
If $\yub=Y$, then 
Lemma~\ref{l:solgen}~(i) and (iii) imply that
$\sum_{x\in X}|\hset{x}|\ne \sum_{y\in Y}|\hset{y}|$.
As $|U|=|V|$, 
either there is $a\in U-\bigcup_{x\in X}\hset{x}$,
or there is $b\in V-\bigcup_{y\in Y}\hset{y}$.

\item
If $Y-\yub\ne\emptyset$, then 
Lemmas~\ref{l:badsqgrid} and \ref{l:solgen}~(i),(iii)
imply that
$\sum_{x\in X}|\hset{x}|< \sum_{y\in Y}|\hset{y}|$.
As $|U|=|V|$, there is $a\in U-\bigcup_{x\in X}\hset{x}$.
\end{itemize}
In both cases,
$\M,(u,v)\models\consf$ follows by Lemma~\ref{l:solgen}~(ii).
\end{proof}

Note
 that when $X\cup Y$ is finite and every $z\in X\cup Y$ is bounded (cf.\ case (i) in Lemma~\ref{l:badsqgrid}),
  then it can happen that 
 $\sum_{x\in X}\xi(x)\ne\sum_{y\in Y}\xi(y)$ for any solution of $\con{\F}$, but
  there are solutions $\xi_1$ and $\xi_2$ of $\con{\F}$
such that $\sum_{x\in X}\xi_1(x)<\sum_{y\in Y}\xi_1(y)$ and $\sum_{x\in X}\xi_2(x)>\sum_{y\in Y}\xi_2(y)$;
see Fig.~\ref{f:bothways} for an example.

\begin{example}\label{e:sqbad}
{\rm
Take the square-bad \grid{} $\F$ in Fig.~\ref{f:sqbad}.
We describe the formulas $\antefone$, $\anteftwo$, $\antefthree$, and $\consf$.

To begin with, we have $X=\{x_1,x_2\}$ and $\yub=\{y_1\}$ (so $\F$ belongs to case (ii)(a) in Lemma~\ref{l:badsqgrid}).
Also, $\gf=\hfub$ has an isolated node $x_2$ and two edges: $y_1\to^2 x_1$ and $y_2\to^2 x_1$. Thus,
the unravelling $\tfub$ of the bounded part $\hfubm$ of $\hfub$ has two roots, $y_1$ and $x_2$ (both
are $6$-strict), and one edge: $y_1\to^2 x_1$. Therefore, we have:
\begin{multline*}
\antefone:\quad
\Dh^+\Dv^+\Bigl(\bigwedge_{i=1}^{6}\Dv^+(\yvarbar_1\land\neg\hvar{y_1}{i}\land\Bv\hvar{y_1}{i})
\land\Dh\bigl(\ctpp{x_1}{y_1}\land\Dv\ctpp{x_1}{y_2}\bigr)\Bigr)\\
\land\Dh^+\Dv^+\Bigl(\bigwedge_{j=1}^{6}\Dh^+(\xvarbar_2\land\neg\hvar{x_2}{j}\land\Bh\hvar{x_2}{j})\Bigr),
\end{multline*}
where
\[
\ctpp{x_1}{y_j}:\quad
\xvarbar_1\land\yvarbar_j\land
\neg\cvar{x_1}{y_j}\land\Bh\cvar{x_1}{y_j}\land\Dv(\xvarbar_1\land\yvarbar_j\land\neg\cvar{x_1}{y_j}\land\Bh\cvar{x_1}{y_j}),\quad\mbox{for $j=1,2$.}
\]
We also have
\[
\anteftwo=\swf{y_1}:\quad\Bh^+\Bv^+\Bigl(\bigwedge_{i=1}^{6}\hvar{y_1}{i}\to\Bh^+\cvar{x_1}{y_1}\Bigr).
\]

Next, we compute the solution $\smin$ of $\con{\F}$. Note that all strongly connected components in $\gf$
are singletons, and so $\smin(z)=\numin\bigl(\{z\}\bigr)$ for all $z\in X\cup Y$. So we have:
%
\begin{align*}
& \smin(x_1)=\minz{x_1}=3,\\
& \smin(y_1)=\max\bigl\{2\cdot\smin(x_1),\minz{y_1}\bigr\}=\max\{2\cdot 3,6\}=6,\\
& \smin(y_2)=\max\bigl\{2\cdot\smin(x_1),\minz{y_2}\bigr\}=\max\{2\cdot 3,4\}=6,\\
& \smin(x_2)=\minz{x_2}=6. 
\end{align*}
Thus, $\xlb=\{x_1\}$, $\ylb=\emptyset$ and $\ylb\cup (Y-\yub)=\{y_2\}$.
We choose $\F^{x_1y_1}$ and $\F^{x_2y_2}$ to `witness'  that
$\minz{x_1}=3$ and $\minz{y_2}=4$, respectively, and so we have 
$\nor{x_1}=\noi{x_1}=1$, $\nor{y_2}=2$, $\noi{y_2}=0$, and
\begin{align*}
& \antefthree:\quad  \Dh^+\Dv^+  \Bigl[
\Dh^+\Bigl(\btp{1}^\circ(x_1)\land\Dv^+\neg\cvar{x_1}{y_1}\land\Dh\bigl(\btp{1}^\circ(x_1)\land\Dv^+\neg\cvar{x_1}{y_1}\bigr)  \Bigr)\\
& \hspace*{.7cm}\land\Dh^+\bigl(\btp{1}^\bullet(x_1)\land\Dv^+\neg\cvar{x_1}{y_1}\bigr)\Bigr]
 \land\Dh^+\Dv^+\Bigl[\Dv^+\bigl(\btp{1}^\circ(y_2)\land\Dv\btp{1}^\circ(y_2)\bigr)
\land\Dv^+\bigl(\btp{2}^\circ(y_2)\land\Dv\btp{2}^\circ(y_2)\bigr)\Bigr],
\end{align*}
where
\begin{align*}
& \btp{1}^\circ(x_1):\quad\xvarbar_1\land\bvar{1}^\circ(x_1)\land\neg\bvar{1}^\bullet(x_1),\hspace*{2cm}
\btp{1}^\bullet(x_1):\quad\xvarbar_1\land\bvar{1}^\bullet(x_1)\land\neg\bvar{1}^\circ(x_1),\\
& \btp{1}^\circ(y_2):\quad\yvarbar_2\land\bvar{1}^\circ(y_2)\land\neg\bvar{2}^\circ(y_2),\hspace*{2.1cm}
\btp{2}^\circ(y_2):\quad\yvarbar_2\land\bvar{2}^\circ(y_2)\land\neg\bvar{1}^\circ(y_2).
\end{align*}

Finally, we have:
\[
\consf:\quad
\Dh^+\Bigl(\Bv^+\cvar{x_1}{y_1}\land \Dv^+\bigwedge_{i=1}^{6}\hvar{x_2}{i}\Bigr).
\]
}
\end{example}

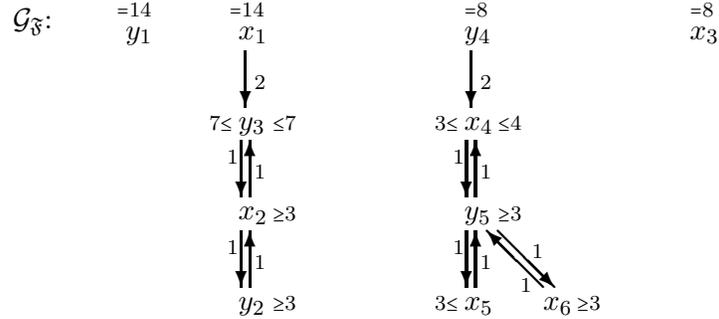
\begin{figure}
\begin{center}
\setlength{\unitlength}{.03cm}
\begin{picture}(310,260)
\multiput(0,0)(0,50){6}{\line(1,0){300}}
\multiput(0,0)(50,0){7}{\line(0,1){250}}

\put(12,8){\ii}
\put(12,32){\ii}
\multiput(16,18)(0,4){3}{\circle*{.5}}
\put(25,20){$14$}

\put(20,109){\ri}
\put(20,118){\rr}
\put(20,127){\rr}
\put(20,136){\rr}

\put(70,65){\ii}
\put(70,75){\rr}

\put(70,115){\ii}
\put(70,125){\rr}

\put(112,158){\ii}
\put(112,182){\ii}
\multiput(116,168)(0,4){3}{\circle*{.5}}
\put(125,170){$8$}

\put(170,165){\ir}
\put(170,175){\rr}

\put(170,215){\ii}
\put(170,225){\rr}

\put(220,215){\ii}
\put(220,225){\rr}

\put(270,215){\ii}
\put(270,225){\rr}

\multiput(70,20)(50,0){5}{\rr}
\put(20,70){\rr}
\multiput(120,70)(50,0){4}{\rr}
\multiput(120,120)(50,0){4}{\rr}
\multiput(20,170)(50,0){2}{\rr}
\multiput(220,170)(50,0){2}{\rr}
\multiput(20,220)(50,0){3}{\rr}

\put(20,-10){$x_1$}
\put(70,-10){$x_2$}
\put(120,-10){$x_3$}
\put(170,-10){$x_4$}
\put(220,-10){$x_5$}
\put(270,-10){$x_6$}
\put(305,20){$y_1$}
\put(305,70){$y_2$}
\put(305,120){$y_3$}
\put(305,170){$y_4$}
\put(305,220){$y_5$}

\put(-20,240){$\F$:}
\end{picture}

\begin{picture}(305,170)(0,40)
\put(0,165){$\gf$:}
\put(46,165){${}^{=14}$}
\put(50,160){$y_1$}
\put(96,165){${}^{=14}$}
\put(100,160){$x_1$}
\put(200,165){${}^{=8}$}
\put(200,160){$y_4$}
\put(300,165){${}^{=8}$}
\put(300,160){$x_3$}

\thicklines
\multiput(103,155)(100,0){2}{\vector(0,-1){25}}
\multiput(107,140)(100,0){2}{${}_2$}

\put(100,120){$y_3$}
\put(115,115){${}^{\leq 7}$}
\put(87,115){${}^{7 \leq}$}

\put(200,120){$x_4$}
\put(215,115){${}^{\leq 4}$}
\put(187,115){${}^{3\leq}$}

\multiput(105,90)(100,0){2}{\vector(0,1){25}}
\multiput(107,100)(100,0){2}{${}_1$}
\multiput(95,100)(100,0){2}{${}^1$}
\multiput(101,115)(100,0){2}{\vector(0,-1){25}}

\put(100,80){$x_2$}
\put(115,75){${}^{\geq 3}$}
\put(200,80){$y_5$}
\put(215,75){${}^{\geq 3}$}

\multiput(105,50)(100,0){2}{\vector(0,1){25}}
\multiput(101,75)(100,0){2}{\vector(0,-1){25}}
\multiput(107,60)(100,0){2}{${}_1$}
\multiput(95,60)(100,0){2}{${}^1$}

\put(100,40){$y_2$}
\put(115,35){${}^{\geq 3}$}
\put(200,40){$x_5$}
\put(187,35){${}^{3\leq}$}
\put(235,40){$x_6$}
\put(250,35){${}^{\geq 3}$}

\put(235,50){\vector(-1,1){25}}
\put(215,75){\vector(1,-1){25}}
\put(230,65){${}_1$}
\put(225,50){${}_1$}
\end{picture}
\caption{An example of a square-bad \grid{} $\F$, with $\con{\F}$ having two solutions $\xi_1$, $\xi_2$ such that $\sum_{x\in X}\xi_1(x)<\sum_{y\in Y}\xi_1(y)$ and $\sum_{x\in X}\xi_2(x)>\sum_{y\in Y}\xi_2(y)$.}\label{f:bothways} 
\end{center}
\end{figure}


\subsubsection{$X\cup Y$ is infinite}\label{s:caseinf}

If $X\cup Y$ is infinite then,
by Lemma~\ref{l:badsqgrid}~(iii), there are two cases. We suppose that
$X$ is finite, every $x\in X$ is bounded, $Y$ is infinite, and 
there is a finite subgrid $\F^-=(X,Y^-,g^-)$ of $\F$ such that
\begin{equation}\label{alllarge}
\mbox{$\sum_{x\in X }\xi(x) <\sum_{y\in Y^-}\xi(y)$, for every solution $\xi$ of $\con{\F^-}$.}
\end{equation}
(The other case is similar.)
By \eqref{alllarge} and the finiteness of 
$\F^-$, the formula $\solfsub$ is defined in \S\ref{s:casefin}.
For every $x\in X$, take the formula $\finalf{x}$ from \eqref{lfinal}, 
and let 
\[
\badsqf:\quad \solfsub\to \Dh^+\bigwedge_{x\in X}\finalf{x}.
\]
Then $\badsqf$ is clearly a generalised Sahlqvist formula.
An inspection of the proof of Lemma~\ref{l:badsqgrid} shows that by
Lemma~\ref{l:badsqdec} we have the following:

\begin{lemma}\label{l:badsqinfdec}
It is decidable whether a bimodal formula is of the form $\badsqf$ for some infinite \grid{} $\F=(X,Y,g)$ 
for which $X$ is finite, every $x\in X$ is bounded, and 
$\F^-=(X,Y^-,g^-)$ is a finite subgrid of $\F$ such that $\con{\F^-}$ is defined and has a solution, but 
$\sum_{x\in X }\xi(x) <\sum_{y\in Y^-}\xi(y)$ for every solution $\xi$ of $\con{\F^-}$.
\end{lemma}

\begin{lemma}\label{l:infsat} 
$\badsqf$ is not valid in $\F$.
\end{lemma}

\begin{proof}
As $\F^-$ is a subgrid of $\F$,  the proof of  Lemma~\ref{l:solfsat} gives a model $\M$ on $\F$ satisfying $\solfsub$.
As the `$X$-coordinates' of both $\F^-$ and $\F$ are the same, it is easy to see that
 $\neg\Dh^+\bigwedge_{x\in X}\finalf{x}$
is satisfied in $\M$ as well.
\end{proof}

\begin{lemma}\label{l:finalfinsqinf}
$\badsqf$ is valid in every square product of difference frames.
\end{lemma}

\begin{proof}
Suppose $\M$ is a model on $(U,\drel{U})\mprod(V,\drel{V})$ for some $U,V$ with $|U|=|V|>0$, and
$\M,(u,v)\models\solfsub$.
For every $z\in X\cup Y^-$, take the set $\hset{z}$ from Lemma~\ref{l:solgen}.
By Lemma~\ref{l:solgen}~(i),(iii), we have $\sum_{x\in X}|\hset{x}|< \sum_{y\in Y^-}|\hset{y}|$.
As $|U|=|V|$, there is $a\in U-\bigcup_{x\in X}\hset{x}$, and so 
$\M,(u,v)\models\consf$ follows by Lemma~\ref{l:solgen}~(ii).
\end{proof}


\subsection{Infinite Sahlqvist axiomatisation for $\dsqxd$}\label{dxdsqsahl}

Though in general generalised Sahlqvist formulas are more expressive than Sahlqvist
formulas \cite{gv06}, there are special settings when their axiomatic powers coincide \cite{gv01}.
Our bimodal language only has two \emph{monadic} modalities $\Dh$ and $\Dv$.
So our generalised Sahlqvist formulas (as defined in \S\ref{s:Sahlqvist} above) are special cases of the \emph{PCFs}
of \cite{gv01} (and of the \emph{inductive formulas} of \cite{gv06}). The modalities $\Dh$ and $\Dv$ are
\emph{self-reversive} 
in the sense that the formulas $p\to\Bh\Dh p$ and $p\to\Bv\Dv p$ belong to $\dsqxd$ 
(by \eqref{symmf}, \eqref{prodinsq} and \eqref {dxdcomm}).
Therefore, it follows from \cite[Thm.~4.10]{gv01} that 
there is an infinite axiomatisation for $\dsqxd$ consisting of Sahlqvist formulas. Moreover,
the Sahlqvist axioms can be obtained algorithmically from 
the generalised Sahlqvist formulas they are axiomatically equivalent with.


\section{Discussion}

We have shown that the 2D product logic $\dxd$ is non-finitely axiomatisable, and also given an 
infinite set $\Sigma_{\dxd}$ of Sahlqvist formulas axiomatising $\dxd$.
We have also proved that its `square' version \mbox{$\dsqxd$} (the modal counterpart 
of two-variable substitution and equality free first-order logic with counting to $2$) is 
non-finitely axiomatisable over $\dxd$, but can be axiomatised by adding infinitely
many Sahlqvist axioms to $\dxd$.
Here are some related issues and open problems:
\begin{enumerate}
\item
The two-player p-morphism game we defined for \clusters{} in the proof of Lemma~\ref{l:clusterpm}
can easily be generalised to arbitrary countable \grids{} $\F$ such that the analogue of
Claim~\ref{c:game} still holds for the game $\mathbb{G}(\F)$.  (Algebraically, this is the \emph{complete representation
game} \`a la Hirsch and Hodkinson \cite{Hirsch&Hodkinson97b}, for subdirectly irreducible atomic diagonal-free strict-cylindric algebras.)
So by Theorem~\ref{t:axdxd}, $\F$ validates the Sahqvist axioms in $\Sigma_{\dxd}$ iff
player $\exists$ has a winning strategy in $\mathbb{G}(\F)$. 
However, the precise connection between particular plays of $\mathbb{G}(\F)$ and the axioms is not clear.


\item 
One might also consider the `lopsided' product logics $\sxd$ and 
%
\[
\ssqxd=
\Log\bigl\{(U,\urel{U})\mprod(V,\drel{V}) : \mbox{$U,V$ are sets with $|U|=|V|>0$}\bigr\}.
\]
%
$\sxd$ is not finitely axiomatisable by Theorem~\ref{t:nonfinax}, and 
a proof very similar to that of Theorem~\ref{t:sqnonfinax} shows that $\ssqxd$ is not finitely axiomatisable over $\sxd$.
Further, using the proof pattern in \S\ref{proofmethod}, it is easy to show that $\sxd$ and $\ssqxd$
are axiomatisable by adding the Sahlqvist axiom $\Bh p\to p$  (expressing that
$\Rh$ is reflexive) to $\dxd$ and $\dsqxd$, respectively.

However, much simpler axioms for these logics can be obtained by actually repeating the proofs of
Theorems~\ref{t:axdxd} and \ref{t:sqax}, and using that  \grids{} are much simpler in these cases.
In particular, it can be shown directly (without using the algorithm of \cite{gv01}) that $\ssqxd$ is Sahlqvist axiomatisable:
In case of \grids{} $\F$ with reflexive $\Rh$, there are only `local' reasons for $\F$ not being the
p-morphic image of a product of a universal and a difference frame. Thus, switch \clusters{} play no role in an axiomatisation, 
and so there is no need for $\anteftwo$-like conjuncts in the antecedents of the axioms.

\item
Both $\dxd$ and $\dsqxd$ are elementarily generated modal logics 
(by Corollary~\ref{co:prodelem} or Theorem~\ref{t:axdxd}, and Theorem~\ref{t:sqax}, respectively).
Hodkinson \cite{Hodkinson06b} `synthesises' modal axioms for such logics from the first-order defining formulas via hybrid logic formulas.
It would be interesting to consider the connections between our axioms and the axioms obtained 
in this way.
Note that we did not use (or even compute) the first-order correspondents of our axioms.

\item
Hirsch and Hodkinson \cite{Hirsch&Hodkinson97,Hirsch&Hodkinson02} give an explicit infinite axiomatisation for (the algebraic counterpart of) the $n$-dimensional product logic $\Sfive^n$, for any $n\in\nNp$. The axioms are obtained by first expressing `universally' the winning strategy for $\exists$ in a two-player `representation' game, and then turning these `universal expressions' to modal formulas by using that there is a universal modality in $\Sfive^n$-frames. 
By the negative results of \cite{hv05,BulianH13},
infinitely many of these axioms cannot be Sahlqvist/canonical whenever $n\geq 3$.
It is easy to see that the method of \cite{Hirsch&Hodkinson97}  can also be used to give 
an explicit infinite axiomatisation for $\Diff^n$, for any $n\in\nNp$, so in particular for $\dxd$. 
As $\Sfive^n$ is finitely axiomatisable over $\Diff^n$ by \cite[Thm.~2.14]{Kurucz10},  infinitely many of the axioms obtained by the method of  \cite{Hirsch&Hodkinson97} cannot be Sahlqvist/canonical whenever $n\geq 3$.
But what about the $n=2$ case?
Are the axioms obtained for $\dxd$ this way Sahlqvist/canonical?

\item
Our axiomatisations are connected to solutions of some special kinds of
integer programming problems.
It would be interesting to understand these connections further, and possibly use
some known integer programming solver methods in order to find simpler axioms.
Note that 
Pratt-Hartmann \cite{Pratt-Hartmann10} also connects the type-structures of two-variable first-order logic with counting to integer programming.

\item
Here we considered the axiomatisation problem for the modal counterpart of two-variable first-order
logic with counting to $2$ only, and without equality and substitution/transposition of variables.
It would be interesting to get closer to the full two-variable fragment with counting, and
study richer languages that contain (some of the) modal operators
`simulating' these missing features (that is, cylindric and \mbox{(quasi-)} polyadic algebras with `graded'
cylindrifications corresponding to counting quantifiers); see \cite{Henkinetal85,Segerberg73,Marx&Venema97,FBdeCaro85}.
\end{enumerate}


\section*{Acknowledgements}

S\'ergio Marcelino's research was funded by FCT/MCTES through national funds and when applicable co-funded EU funds under the project UIDB/EEA/50008/2020. We thank the anonymous referees for their careful reading and expert criticism of the manuscript.




\end{document}